\documentclass[12pt,onecolumn,draftcls]{IEEEtran}
\IEEEoverridecommandlockouts
\usepackage[english]{babel}
\usepackage{amsfonts}
\usepackage{amssymb}
\usepackage[cmex10]{amsmath}
\usepackage{amssymb}
\usepackage{graphicx}
\usepackage{graphics}
\usepackage{cite}
\usepackage{subcaption}
\usepackage{dsfont}
\newcommand{\CP}{C_\mathcal{P}}
\newcommand{\F}{\mathbb{F}}
\newcommand{\N}{\mathbb{N}}
\newcommand{\R}{\mathbb{R}}
\newcommand{\cS}{\mathcal{S}}
\newcommand{\cR}{\mathcal{R}}

\newcommand{\G}{\mathcal{G}}
\newtheorem{definition}{Definition}
\newtheorem{proposition}{Proposition}
\newtheorem{theorem}{Theorem}
\newtheorem{lemma}{Lemma}
\newtheorem{corollary}{Corollary}
\newtheorem{example}{Example}
\begin{document}
\title{Edge Coloring and Stopping Sets Analysis in Product Codes with MDS components}

\author{Fanny Jardel and Joseph J. Boutros \\
\thanks{This manuscript was submitted to the IEEE Transactions on Information Theory, paper IT-15-1104, Dec. 2015. Fanny Jardel is with Telecom ParisTech, 75013 Paris, France (email: fannjard@gmail.com).
She was with CEA, LIST, Communicating Systems Laboratory
BC 94, Gif Sur Yvette, F91191, France.
Joseph J. Boutros is with the Dept. of Electrical and Computer Engineering, 
Texas~A\&M University at Qatar, Education City, 
23874 Doha, Qatar (email: boutros@tamu.edu).}}


\maketitle

\begin{abstract}
We consider non-binary product codes with MDS components and their
iterative row-column algebraic decoding on the erasure channel. Both independent and block 
erasures are considered in this paper. 
A compact graph representation is introduced on which we define
double-diversity edge colorings via the rootcheck concept. An upper
bound of the number of decoding iterations is given as a function of the graph size
and the color palette size $M$. Stopping sets are defined in the context
of MDS components and a relationship is established with the graph representation.
A full characterization of these stopping sets is given up to a size
$(d_1+1)(d_2+1)$, where $d_1$ and $d_2$ are the minimum Hamming distances
of the column and row MDS components respectively. 
Then, we propose a differential evolution edge coloring algorithm
that produces colorings with a large population of minimal rootcheck order symbols.
The complexity of this algorithm per iteration is $o(M^{\aleph})$, 
for a given differential evolution parameter $\aleph$, where $M^{\aleph}$ itself is small
with respect to the huge cardinality of the coloring ensemble. 
The performance of MDS-based product codes with and without double-diversity coloring
is analyzed in presence of both block and independent erasures.
In the latter case, ML and iterative decoding are proven to coincide at small
channel erasure probability.
Furthermore, numerical results show excellent performance in presence of unequal erasure
probability due to double-diversity colorings.
\end{abstract}

\newpage
\section{Introduction}
The colossal amount of data stored or conveyed by network nodes requires
a special design of coding structures to protect information against loss or errors
and to facilitate its access. At the end-user level, coding is essential for transmitting
information towards the network whether it is located in a single node or distributed over many nodes.
At the network level, coding should help nodes to reliably save a big amount of data
and to efficiently communicate with each others. Powerful capacity-achieving
error-correcting codes developed in the last two decades are mainly efficient 
at large or asymptotic block length, e.g. low-density parity-check codes (LDPC)~\cite{Gallager1963} 
and their spatially-coupled ensembles~\cite{Kudekar2013},
parallel-concatenated convolutional (Turbo) codes~\cite{Berrou1996}\cite{Benedetto1996}, 
and polar codes derived from channel polarization~\cite{Arikan2009}. 
Data transmission and storage in many nowadays networks
may require short-length packets that are not suitable for capacity-achieving codes.
The current interest in finite-length channel coding rates \cite{Polyanskiy2010} put back the light
on code design for short and moderate block length. Many potential candidates are available for
this non-asymptotic length context such as binary and non-binary BCH codes,
including Reed-Solomon (RS) codes,
Reed-Muller (RM) codes, 
and tensor product codes of all these linear block codes~\cite{MacWilliams1977}\cite{Blahut2003}\cite{Lin2004}.

Product codes, introduced by Peter Elias in 1954~\cite{Elias1954}, 
are tensor products of two (or more) simple codes
with a structure that is well-suited to iterative decoding via its graphical description.
In the early decades after their invention, product codes received a great attention 
due to their capability of correcting multiple burst errors~\cite{Wolf1965}\cite{Overveld1987},
the availability of erasure-error bounded-distance decoding algorithms \cite{Wainberg1972},
the ability of correcting many errors beyond the guaranteed correction capacity~\cite{Abramson1968},
and their efficient implementation with a variable rate \cite{Weng1967}. 
The pioneering work by Tanner~\cite{Tanner1981} brought new tools to coding theory
and put codes on graphs, including product codes, 
and their iterative decoding in the heart of modern coding theory~\cite{Kschischang2001}\cite{Kschischang2003}\cite{Richardson2008}.
The graph approach of coding led to new optimal cycle codes on Ramanujan/Cayley graphs~\cite{Tillich1997} 
and to Generalizations of LDPC and product codes, known as GLD codes, studied for 
the binary symmetric channel (BSC) and the Gaussian channel~\cite{Boutros1999}.
The excellent performance of iterative (turbo) decoding of product codes 
on the Gaussian channel~\cite{Pyndiah1998}
made them compete with Turbo codes and LDPC codes for short and moderate block length.
The convergence rate and stability of product codes iterative decoding were studied
based on a geometric framework~\cite{Sella2001}. Product codes with mixed convolutional
and block components were also found efficient in presence of impulsive noise~\cite{Freeman1996}.
In addition, iterated Reed-Muller product codes were shown to exhibit good decoding thresholds
for the binary erasure channel, but at high and low coding rates only~\cite{Varodayan2002}.

The class of product codes in which the row and the column code are both Reed-Solomon codes
was extensively used since more than two decades in DVD storage media 
and in mobile cellular networks~\cite{Wicker1994}. In these systems, the channel
is modeled as a symbol-error channel without soft information, i.e. suited to algebraic decoding.
Improvements were suggested for these RS-based product codes such as soft information provided
by list decoding~\cite{Sarwate2005} within the iterative process in a Reddy-Robinson framework~\cite{Reddy1972}.
Also, RS-based product codes were directly decoded via a Guruswami-Sudan list decoder~\cite{Guruswami1999}
after being generalized to bivariate polynomials~\cite{Augot2006}.
For general tensor products of codes and interleaved, 
a recent efficient list decoding algorithm was published \cite{Gopalan2011}, 
with an improved list size in the binary case.
On channels with soft information, RS-based product codes may be row-column decoded
with soft-decision constituent decoders \cite{El-Khamy2006}\cite{Jiang2004}.

Tolhuizen found the Hamming weight distribution of both binary and non-binary 
product codes up to a weight less than 
$d_1d_2+\max(d_1\lceil d_2/q\rceil, d_2\lceil d_1/q\rceil)$~\cite{Tolhuizen2002}.
Enumeration of erasure patterns up to a weight less than $d_1d_2+\min(d_1,d_2)$ 
was realized by Sendrier for product codes with MDS components \cite{Sendrier1991}.
Rosnes studied stopping sets of binary product codes 
under iterative ML-component-wise decoding \cite{Rosnes2008},
where the defined stopping sets and their analysis are based 
on the generalized Hamming distance~\cite{Wei1991}\cite{Helleseth1992}.
\subsection{Paper content and structure}
In this paper, we consider non-binary product codes with MDS components and their
iterative algebraic decoding on the erasure channel. Both independent and block 
erasures are considered in our paper. The erasure channel is currently a major
area of research in coding theory \cite{Kudekar2015}\cite{Kumar2015} because
of strong connections with theoretical computer science \cite{Kumar2015} 
and its model that easily allows to understand the behavior of codes such as for LDPC codes~\cite{Di2002}, 
for general linear block codes~\cite{Schwartz2006}, and for turbo codes~\cite{Rosnes2007}.
Coding for block erasures was examined by Lapidoth in the context 
of convolutional codes~\cite{Lapidoth2005}. This was a basis to later construct codes for the block-fading
channel with additive white Gaussian noise~\cite{Guillen2006a}\cite{Boutros2010}. 
The notion of {\em rootcheck} introduced in \cite{Boutros2010}\cite{Boutros2009ita} for single-parity
checknodes was applied to more general checknodes in GLD codes \cite{Boutros2008GLD} 
and product codes \cite{Boutros2008} to achieve diversity on non-ergodic block-fading channels.
The rootcheck concept is the main tool in this paper, in a way similar to \cite{Boutros2008},
to define a compact graph representation and study iterative decoding in presence of block erasures.
Edge coloring is one of the most interesting problems in modern graph theory~\cite{Bollobas1998}. 
In this paper, edge coloring is a tool, when combined to the rootcheck concept, yields 
double-diversity product codes. Our work is valid for finite-length MDS-based product codes only.
Product codes for asymptotic block length were studied for single-parity codes constituents \cite{Rankin1999}
and for the erasure channel with a standard regular structure \cite{Schwartz2005} 
and MDS-based irregular structures \cite{Alipour2012}.

Whether a product code is endowed with an edge coloring or not, the analysis of stopping sets,
their characterization and their enumeration is a fundamental task to be able to design
codes for erasure channels and determine the decoder performance. Our work in this sense
is an improvement to previous works cited above by Tolhuizen, Sendrier, and Rosnes. 
Besides this objective of stopping sets characterization which is useful for independent channel
erasures and erasures occurring in blocks of symbols, 
recent works on locality \cite{Gopalan2012} stimulated us
to search for edge colorings with a large population of edges that admit a minimal rootcheck order.
Locality is a concept encountered in distributive storage \cite{Kubiatowicz2000}\cite{Rashmi2013}
where classic coding theory is adapted to the nature of a network with distributed nodes
with its own constraints of load in bandwidth and storage \cite{Dimakis2011}\cite{Oggier2013}.
Furthermore, product codes with MDS components appear to be suited to distributive storage \cite{Esmaili2013}
owing to their simple and mature techniques of erasure resilience.
In our search for good edge colorings, we provide a new algorithm based on the concept of 
differential evolution~\cite{Storn1997}\cite{Onwubolu2009}. We use no crossover in our evolution loop,
only a mutation of the population of bad edges is made to search for a better edge coloring.
Our MDS-based product codes equipped with a double-diversity edge coloring are suited to distributed
storage applications and to wireless networks where diversity is a key parameter.\\

The paper is structured as follows. Section II gives a list of mathematical notations. 
The graph representation of product codes is given in Section III, 
including compact and non-compact graphs. Also the rootcheck concept and its consequences
are also found in  Section III. The analysis of stopping sets is made in Section IV. 
Our edge coloring algorithm for bipartite graphs of product codes is described in Section V.
Finally, in Section VI, we study the performance of product codes with MDS components
on erasure channels and we give theoretical and numerical results before the conclusions
in the last section.
\subsection{Main results \label{sec_main_results}}
The main results in this paper are:
\begin{itemize}
\item Establishing a new compact graph for product codes. The compact graph has many
advantages, the main one being its ability to imitate a Tanner graph with parity-check nodes.
The compact graph is also the basis for the differential evolution edge coloring.
See Section~\ref{sec_graph_representations_sub2}.
\item Iterative decoding analysis of finite-length product codes, mainly the proof
of new bounds on the number of decoding iterations. See Theorem~\ref{th_root_order_max} and Corollary~\ref{cor_rho_max}.
\item Proving new properties of stopping sets for product codes with MDS components.
See Propositions~\ref{prop_not_MDS}\&\ref{prop_MDS_w}, Corollaries~\ref{coro_not_MDS}-\ref{cor_stop_in_Gc}, and Lemmas~\ref{lem_max_support}\&\ref{lem_max_zeros}.

\item Complete enumeration and characterization 
of stopping sets up to a size $(d_1+1)(d_2+1)$,
where $d_1, d_2$ are the minimum Hamming distances of the component codes.
This stopping set enumeration goes beyond the weight $d_1d_2+\max(d_1,d_2)$ of Tolhuizen's 
Theorem~3 for codeword enumeration in the MDS components case.
See Lemmas~\ref{lem_graph_bipartite_deg2}\&\ref{lem_graph_bipartite_deg2_1} and Theorems~\ref{th_stopping_sets_d}\&\ref{th_stopping_sets_d1_d2}.

\item A new edge coloring algorithm (DECA) capable of producing double-diversity colorings
despite the huge size of the coloring ensembles. See Section~\ref{sec_edge_coloring_sub2}.

\item Construction via the DECA algorithm of product codes maximizing the number of edges with root order $1$,
i.e. minimizing the locality when the process of repairing nodes is considered. See Section~\ref{sec_edge_coloring_sub3}.

\item First numerical results for MDS-based product codes on erasure channels
showing how close iterative decoding is to ML decoding, mainly for small $\epsilon$.
We proved that iterative decoding perform as well as ML decoding (the ratio of error probabilities
tends to $1$) for MDS-based product codes at small $\epsilon$.
See Proposition~\ref{prop_stop_codewords}, Corollary~\ref{cor_PewG_PewML}, and other performance results in Section~\ref{perf_erasure_sub2}.

\item Great advantage of double-diversity colorings 
of product codes (with respect to codes without coloring) in presence of unequal probability erasures. 
Thus, double-diversity colorings are efficient on both ergodic and non-ergodic erasure channels.
See Section~\ref{perf_erasure_sub3}.
\end{itemize}

\section{Mathematical notation and Terminology}
We start by the notation related to the product code and its row and column components.
The impatient reader may skip this entire section and then refer to it later to clarify
any notation within the text. Basic notions on product codes and fundamental properties
are found in main textbooks \cite{MacWilliams1977}\cite{Blahut2003}\cite{Lin2004}
and the encyclopedia of telecommunications \cite{Kschischang2003}.\\
The column code $C_1$ is a linear block code over the finite field $\F_q$ with parameters
$[n_1, k_1, d_1]_q$ which may be summarized by $[n_1,k_1]$ when no confusion is possible.
The integer $q$ is the code alphabet size, $n_1$ is the code length, $k_1$ is the code dimension
 as a vector subspace of $\F_q^{n_1}$, and $d_1$ is the minimum Hamming distance of $C_1$.
Similarly, the row code $C_2$ is a linear block code with parameters $[n_2, k_2, d_2]_q$.
Let $G_1$ and $G_2$ be two matrices of size $k_1 \times n_1$ and $k_2 \times n_2$
containing in their row a basis for the subspaces $C_1$ and $C_2$ respectively.
From the two generator matrices $G_1$ and $G_2$ a product code $C_P$ is constructed
as a subspace of $\F_q^N$ with a generator matrix $G_P=G_1 \otimes G_2$, where $N=n_1n_2$
and $\otimes$ denotes the Kronecker product \cite{MacWilliams1977}. $C_P$ has dimension $K=k_1k_2$
and minimum Hamming distance $d_P=d_1d_2$.
$C_1$ and $C_2$ are also called component codes, this is a terminology from concatenated codes.
In \cite{Tanner1981} and \cite{Boutros2008}, 
vertices associated to component codes are called subcode nodes.

A linear $[n,k,d]_q$ code is said to be MDS, i.e. Maximum Distance Separable, 
if it satisfies $d=n-k+1$. Binary MDS codes are the trivial repetition codes
and the single parity-check codes. In this paper, we only consider non-trivial
non-binary MDS codes where $q > n > 2$. 
A linear code over $\F_q$ of rate $R=k/n$ is said to be MDS diversity-wise 
or MDS in the block-fading/block-erasure sense if it achieves a diversity order 
$L$ such that $L=1+\lfloor M(1-R)\rfloor$, 
where $M$ is the number of degrees of freedom in the channel.
The right term $1+\lfloor M(1-R)\rfloor$ 
is known as the block-fading Singleton bound \cite{Malkamaki1999}\cite{Knopp2000}.
In this paper, $M$ shall denote the number of colors, i.e. the palette size of an edge coloring.
Assume that code symbols are partitioned into $M$ sub-blocks,
a code is said to attain diversity $L$ if it is capable of correct decoding
when $L-1$ sub-blocks are erased by the channel. The reader should
refer to \cite{Tse2005book}, chapter 3, for an exact definition of diversity
on fading channels with additive white Gaussian noise.

A product code shall be represented by a non-compact graph $\G=(V_1, V_2, E)$.
$\G$ is a complete bipartite graph where $V_1$ is the set of $n_2$ right vertices,
$V_2$ is the set of $n_1$ left vertices, and $E$ is the set of $N$ edges representing
the code symbols. A compact graph $\G^c$ will also be introduced in the next section
with $\G^c=(V_1^c, V_2^c, E^c)$. The number of edges (also called super-edges)
in the compact graph is $|E^c|=N^c$. A super-edge is equivalent to a super-symbol
that represents $(n_1-k_1)(n_2-k_2)$ symbols from $\F_q$. 
The ensemble of edge colorings is denoted $\Phi(E)$ and $\Phi(E^c)$
for $\G$ and $\G^c$ respectively. An edge coloring will be denoted by $\phi$.
Given $\phi$, the rootcheck order of an edge is $\rho(e)$. The greatest
$\rho(e)$ among all edges will be referred to as $\rho_{max}(\phi)$.
The number of edges $e$ satisfying $\rho(e)=1$ is $\eta(\phi)$,
this is the number of good edges and will be processed by the DECA algorithm
in Section~\ref{sec_edge_coloring}. The DECA parameter $\aleph$ shall
represent the number of edges to be mutated, i.e. those edges being chosen in the population
of bad edges satisfying $\rho(e)>1$.

Under iterative row-column decoding, the rootcheck order $\rho$ is equal
to the number of decoding iterations required to solve the edge value (or the symbol
associated to that edge). In this paper, one decoding iteration is equivalent
to decoding all rows or decoding all columns. A sequence of $n_1$ row decoders
followed by a sequence of $n_2$ column decoders is counted as two decoding iterations.

We give now a general definition of a stopping set. 
A detailed study is found in Section \ref{sec_stopping_sets}.
The notion of a stopping set is useful for iterative decoding in presence of erasures \cite{Di2002}.
\begin{definition}
\label{def_stopping_set}
Let $C[n,k]_q$ be a linear code. Assume that the symbols of a codeword 
are transmitted on an erasure channel.
The decoder $\mathcal{D}$ is using some deterministic decoding method. 
Consider a set $\mathcal{S}$ of $s$ fixed positions  
$i_1, i_2, \ldots, i_s$ where $1 \le i_j \le n$. 
The set $\mathcal{S}$ is said to be a Stopping Set 
if $\mathcal{D}$ fails in retrieving the transmitted codeword when all
symbols on the $s$ positions given by $\mathcal{S}$ are erased.
\end{definition}

This paper focuses on stopping sets of a product code under iterative algebraic row-column decoding,
i.e. referred to as type II stopping sets. The number of stopping sets of size $w$ is $\tau_w$.
The rectangular support $\cR(\cS)$ of a stopping set $\cS$ can be seen as the smallest rectangle 
containing $\cS$. After excluding rows and columns not involved in $\cS$, the rectangular 
support has size $\ell_1 \times \ell_2$ where $w=|\cS| \le \ell_1\ell_2$.
The word error performance of $C_P$ shall be estimated on erasure channels, 
$P_{ew}^{ML}$ is the word error probability under Maximum Likelihood decoding
and $P_{ew}^{\G}$ is the word error probability under iterative row-column decoding.
Three erasure channels are considered: 1- The Symbol Erasure Channel, $SEC(q,\epsilon)$,
where code symbols are independently erased with a probability $\epsilon$, 
2- The Color Erasure Channel, $CEC(q,\epsilon)$, where all symbols associated to the same color
are block-erased with a probability $\epsilon$. On the $CEC(q,\epsilon)$, block-erasure
events are independent from one color to another. 3- The unequal probability Symbol Erasure Channel,
$SEC(q,\{\epsilon_i\}_{i=1}^M)$, where symbol erasures are independent but their erasure probability 
varies from one color to another.

\section{Graph representations for diversity \label{sec_graph_representations}}
Efficient graph representation of codes was established by Tanner for different types
of coding structures \cite{Tanner1981}. Bounds on the code parameters 
and iterative decoding algorithms were also proposed for codes on graphs \cite{Tanner1981}. 
In this paper, we study the edge coloring of a product code graph, where edges represent code symbols.
As shown below, the original graph for a product code is too complex, i.e. it leads to a large ensemble
of colorings. Hence, we introduce a compact graph where symbols are grouped together with the same color
in order to reduce the size of the coloring ensemble. The compact graph also has another asset: 
grouping parity symbols together renders check nodes similar to parity-check nodes found
in standard low-density parity-check codes~\cite{Gallager1963}~\cite{Richardson2008}. 

\subsection{Non-compact graph \label{sec_graph_representations_sub1}}
Consider a product code $C_1[n_1,k_1]_q \otimes C_2[n_2,k_2]_q$ where $C_1$ is the column code
and $C_2$ is the row code. The product code is defined over the finite field $\F_q$ and has length $N$
and dimension $K$ given by \cite{MacWilliams1977}
\begin{equation}
N=n_1 n_2, ~~~~K=k_1 k_2.
\end{equation}
Each code symbol simultaneously 
belongs to one row and to one column. Product codes
studied in this paper are regular, 
in the sense that all columns are codewords of $C_1$ and all rows are codewords
of $C_2$. The graph of $C_1[n_1,k_1]_q \otimes C_2[n_2,k_2]_q$ is built as follows. We use the same
terminology as in \cite{Richardson2008}:
\begin{itemize}
\item $n_1$ check nodes are drawn on the left. A left check node represents the coding constraint
which states that a row belongs to $C_2$. The $n_1$ left check nodes are referred to as $C_2$ check nodes,
or row check nodes, or equivalently left vertices.
\item $n_2$ check nodes are drawn on the right. A right check node represents the coding constraint
which states that a column belongs to $C_1$. The $n_2$ right check nodes are referred to as $C_1$ check nodes,
or column check nodes, or equivalently right vertices.
\item An edge is drawn between a left vertex and right vertex. 
It represents a code symbol located on the row
of the left vertex and on the column of the right vertex. The code symbol belongs to $\F_q$.
\end{itemize}

\begin{figure}[!h]
\begin{center}
\includegraphics[width=0.6\columnwidth]{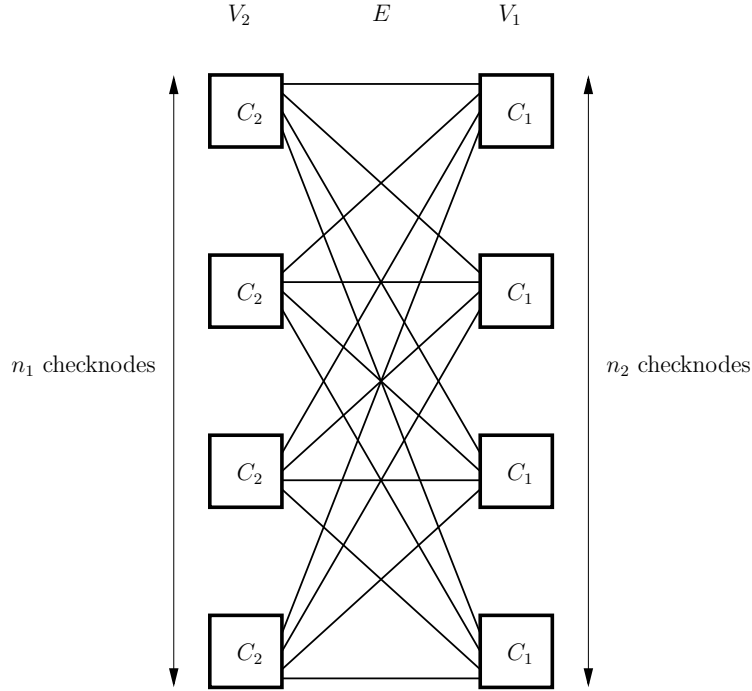}
\caption{Non-compact bipartite graph $\G=(V_1, V_2, E)$ of a product code $[4,2]^{\otimes 2}$, 
i.e. $n_1=n_2=4$, $k_1=k_2=2$, $|V_1|=|V_2|=4$, and $|E|=N=n_1n_2=16$ edges representing 16 symbols in~$\F_q$.\label{complete_graph4_4}}
\end{center}
\end{figure}

In summary, the product code graph $(V_1, V_2, E)$ is a complete biregular bipartite graph built from $n_1$ left vertices, $n_2$ right vertices,
and $N=|E|=n_1 n_2$ edges representing code symbols. The left degree is $n_2$ and the right degree is $n_1$.
Irregular product codes can be found in \cite{Alipour2012}. 
Our paper is restricted to regular product codes.
Figure~\ref{complete_graph4_4} shows the bipartite graph of
a square regular symmetric product code $[4,2] \otimes [4,2]$. The graph structure reveals $n_1$, $n_2$,
and $N=n_1n_2$. The dimensions $k_1$ and $k_2$ of the component codes have no effect on the number of vertices
and edges in the product code graph. Indeed,  a $[4,3] \otimes [4,3]$ code can also be defined by the graph
in Figure~\ref{complete_graph4_4}. The role of the dimensions $k_1$ and $k_2$ is played within the check constraints
inside left and right vertices. Similarly, the size of the finite field defining the code cannot be revealed
from the graph structure, i.e. the product code graph does not depend on $q$.\\

\begin{definition}
The non-compact graph $\G=(V_1, V_2, E)$ for a $[n_1,k_1] \otimes [n_2,k_2]$ product code is a 
complete bipartite graph with $n_1=|V_2|$ left vertices and $n_2=|V_1|$ right vertices.
\end{definition}

\subsection{Compact graph \label{sec_graph_representations_sub2}}
In \cite{Boutros2008} where the diversity of binary product codes was considered, 
vertices of the non-compact graph were grouped together into super-vertices (or supernodes)
because the different channel states lead to multiple classes of check nodes as in root-LDPC codes \cite{Boutros2010}.
To render a graph-encodable code, supernodes in \cite{Boutros2008} were made by putting $n-k$
nodes together for a $[n,k]$ component code. Also, $n-k$ is not necessarily a divisor of $n$.\\

\begin{definition}
The compact graph $\G^c=(V^c_1, V^c_2, E^c)$ for a $[n_1,k_1] \otimes [n_2,k_2]$ product code is a complete bipartite graph
with $\lceil \frac{n_1}{n_1-k_1} \rceil=|V^c_2|$ left vertices and $\lceil \frac{n_2}{n_2-k_2} \rceil=|V^c_1|$ right vertices.
\end{definition}
~\\
From the above definition, the number of edges in the compact graph $\G^c$ is found to be
\begin{equation}
\label{equ_Nc}
N^c=|E^c|=\left\lceil \frac{n_1}{n_1-k_1} \right\rceil \times \left\lceil \frac{n_2}{n_2-k_2} \right\rceil.
\end{equation}
Assuming that $(n_1-k_1)$ divides $n_1$ and $(n_2-k_2)$ divides $n_2$, 
a left check node in $\G^c$ is equivalent to $n_2-k_2$ row constraints 
and a right check node
in $\G^c$ is equivalent to $n_1-k_1$ column constraints. 
An edge in the compact graph carries $(n_1-k_1) \times (n_2-k_2)$ code symbols.
To avoid confusion between edges of $\G$ and $\G^c$, we may refer to those in $\G^c$
as super-edges or equivalently as super-symbols. 
If $n_i$ is not multiple of $n_i-k_i$, then the last row or column supernode will contain less
than $n_i-k_i$ check nodes. Figure~\ref{compact_graph2_2} depicts the compact graph of the $[4,2]^{\otimes 2}$ product code. All $[n,n/2]^{\otimes 2}$ product codes have a compact graph identical
to that of $[4,2]^{\otimes 2}$, for all $n \ge 2$, $n$ even. 

\begin{figure}[!h]
\begin{center}
\includegraphics[width=0.65\columnwidth]{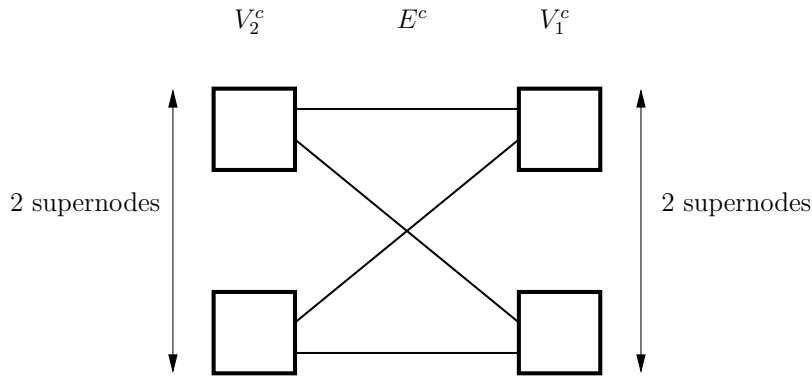}
\caption{Compact bipartite graph $\G^c=(V^c_1, V^c_2, E^c)$ with two supernodes on
each side for the product code $[n,n/2]^{\otimes 2}$, $|V^c_1|=|V^c_2|=2$ and $|E^c|=N^c=4$ supersymbols. Each super-symbol (i.e. super-edge) contains $n^2/4$ symbols (i.e. edges). \label{compact_graph2_2}}
\end{center}
\end{figure}

\subsection{Diversity and codes on graphs \label{sec_graph_representations_sub3}}
From a coding point of view, diversity is the art of creating many replicas of the same information.
From a channel point of view, diversity is the number of degrees of freedom available while transmitting information.
In distributive storage, independent failure of individual machines is modeled by independent erasures of code symbols,
while the outage of a cluster of machines is modeled as block erasures of code symbols.
Assuming a storage domain with a large set of machines partitioned into $M$ clusters, 
diversity of distributed coding is defined
as follows:\\

\begin{definition}
Consider a product code $C_P$ defined over $\F_q$. Assume that symbols
are given $M$ different colors. Erasing one color is equivalent to erasing all symbols
having this color. The code is said to achieve a diversity $L$ if it is capable of filling all erasures
after erasing $L-1$ colors. The code is full-diversity when $L=M$.
\end{definition}
~\\
The integer $L$ may also be called the diversity order. For Gaussian channels with fading, 
the diversity order appears as the slope of the error probability, 
i.e. $L = \lim_{\gamma \rightarrow \infty} -\frac{\log P_e}{\log \gamma}$ \cite{Boutros2010}.
In the above definition, a cluster has been replaced by a color. We will use this terminology
throughout the paper. Notice that coloring symbols is equivalent to edge coloring
of the product code graph. The number of edges is $N$ in the non-compact graph and $N^c$ in the compact graph.
In the sequel, all colorings are supposed to be perfectly balanced, i.e. $M$ divides
both $N$ and $N^c$ and the number of edges having the same color is $N/M$ and $N^c/M$ for
the non-compact graph and the compact graph respectively. More formally, our edge coloring
is defined as follows: an edge coloring $\phi$ of $\G=(V_1, V_2, E)$ 
is a mapping associating one color to every edge in $E$,
\begin{equation}
\phi: E \rightarrow \{1, 2, \ldots, M \},
\end{equation}
such that $|\phi^{-1}(i)|=N/M$ for $i=1 \ldots M$, where $\phi^{-1}(i)$ is the inverse image of $i$.
Similarly, $\phi: E^c \rightarrow \{1, 2, \ldots, M \}$ for $\G^c=(V^c_1, V^c_2, E^c)$
and $|\phi^{-1}(i)|=N^c/M$. The set of such mappings for $\G$ and $\G^c$
is denoted $\Phi(E)$ and $\Phi(E^c)$ respectively.
\\
Consider a coloring $\phi$ in $\Phi(E^c)$. It can be embedded into $\Phi(E)$ by
copying the color of a super-edge to its associated $(n_1-k_1)\times(n_2-k_2)$ edges in $E$.
Thus, let $\Phi(E^c \rightarrow E)$ be the subset of colorings in $\Phi(E)$ obtained by 
embedding all colorings of $\Phi(E^c)$ into $\Phi(E)$. We have
\begin{equation}
\Phi(E^c\rightarrow E) \subset \Phi(E) ~~~\text{and}~~~ |\Phi(E^c \rightarrow E)|=|\Phi(E^c)|.
\end{equation}

The size of the edge coloring ensembles $\Phi(E)$ and $\Phi(E^c)$ is obviously not the same
when $N^c < N$, 
which occurs for both row and column component codes not equal to single parity-check codes.
Indeed, when a palette of size $M$ is used to color edges, the total 
number of colorings of $E$ is
\begin{equation}
\label{equ_PhiE}
|\Phi(E)|=\frac{N!}{((N/M)!)^M}.
\end{equation}
This number for the compact graph is 
\begin{equation}
\label{equ_PhiEc}
|\Phi(E^c)|=\frac{N^c!}{((N^c/M)!)^M}.
\end{equation}
As an example, for the $[12,10]^{\otimes 2}$ code and $M=4$, 
there are $2\cdot 10^{83}$ edge colorings for the non-compact
graph and $2\cdot 10^{19}$ edge colorings for the compact graph.
It is clear that the construction of product codes for diversity is much easier
when based on $\G^c=(V^c_1, V^c_2, E^c)$ because its edge coloring ensemble is smaller. Furthermore,
as described below, vertices in $\G^c$ act in a way similar to standard LDPC check nodes
making the design very simple. Furthermore, we will see in Section~\ref{sec_stopping_sets} 
that edge colorings of the compact graph render 
larger stopping sets than colorings of the non-compact graph.\\

The diversity order $L$ attained by a code can never exceed $M$, the latter being 
the diversity from a channel point of view. A tighter upper bound of $L$ 
showing the rate-diversity tradeoff is the block-fading Singleton bound. 
The Singleton bound for the maximal achievable diversity order is valid
for all types of non-ergodic channels, including block-erasure and block-fading channels.
The block-fading Singleton bound states that~\cite{Knopp2000}~\cite{Malkamaki1999}
\begin{equation}
\label{equ_singleton_bound}
L \le 1+\lfloor M(1-R)\rfloor,
\end{equation}
where $R=K/N$ is the coding rate of the product code. 
Codes satisfying the equality in the above Singleton bound are referred to as 
diversity-wise MDS or block-fading MDS codes.
From (\ref{equ_singleton_bound}), we deduce that $R \le 1/M$ if $L=M$ (full-diversity coding). 
For example, we get $R \le 1/2$ 
with an edge coloring using $L=M=2$ colors and $R \le 1/4$ for $L=M=4$ colors.
The coding rate can exceed $1/M$ when $L<M$ in applications where full diversity is not
mandatory. An example suited to distributed storage is an edge coloring
with a palette of $M=4$ colors, a diversity $L=2$, and $R \le 3/4$.

\subsection{Rootcheck nodes and root symbols \label{sec_graph_representations_sub4}}
In a way similar to root-LDPC codes and product codes built for block-fading channels
\cite{Boutros2008}\cite{Boutros2010},
we introduce now the notion of root symbols and root-check nodes in product codes
to be designed for distributive storage. A linear $[n,k]_q$ code with parity-check matrix $H$
can fill $n-k$ erasures at positions where the columns of $H$ are independent. 
These $n-k$ symbols correspond to $n-k$ separate edges in the non-compact graph
and to a unique edge (supersymbol) in the compact graph. Therefore, for simplicity,
we start by defining a root supersymbol in the compact graph 
where supernodes are equivalent to standard LDPC parity-check nodes.

\begin{definition}
\label{def_root_supersymbol}
Let $\G^c$ be a compact graph of a product code, let $\phi$ be a given edge coloring,
and let $e \in E^c$ be a supersymbol. $e$ is a {\em root supersymbol} with respect to $\phi(e)$
if it admits a neighbor vertex $\upsilon$, $\upsilon \in V^c_1$ or $\upsilon \in V^c_2$,
such that all adjacent edges $f$ in $\upsilon$ satisfy $\phi(f) \ne \phi(e)$.
\end{definition}
~\\
In Definition~\ref{def_root_supersymbol}, if $\upsilon \in V^c_1$ then $e$ is a root supersymbol
thanks to the product code column to which it belongs, i.e. $e$ can be solved in one iteration
by its column component code when the color $\phi(e)$ is erased. Likewise, $e$ is protected
against erasures by its row component code if $\upsilon \in V^c_2$ in the previous definition. 
Finally, a root supersymbol
may be doubly protected by both its row and its column if both right and left neighbors $\upsilon_1 \in V^c_1$
and $\upsilon_2 \in V^c_2$ satisfy the condition of Definition~\ref{def_root_supersymbol}.

\begin{definition}
\label{def_root_symbol}
Let $\G$ be a non-compact graph of a product code, let $\phi$ be a given edge coloring,
and let $e \in E$ be a symbol. $e$ is a {\em root symbol} with respect to $\phi(e)$
if it admits a neighbor vertex $\upsilon$ such that:\\
$\phi(f)=\phi(e)$ for at most $n_2-k_2-1$ adjacent edges $f$ if $\upsilon \in V_1$, or\\ 
$\phi(f)=\phi(e)$ for at most $n_1-k_1-1$ adjacent edges $f$ if $\upsilon \in V_2$.
\end{definition}
~\\
As mentioned in the paragraph before Definition~\ref{def_root_supersymbol}, Definition~\ref{def_root_symbol}
implies that the $n_i-k_i$ root symbols with the same color should belong to positions of independent columns
in the parity-check matrix of the component code $C_i$. This constraint automatically disappears
for MDS component codes since any set of $n_i-k_i$ columns of $H_i$ has full rank.

\subsection{The rootcheck order in product codes \label{sec_graph_representations_sub5}}
Not all symbols of a product code are root symbols. Under iterative row-column decoding on channels with 
block erasures, some symbols may be solved in two decoding iterations or more. Some set of symbols
may never be solved and are referred to as stopping sets \cite{Di2002}\cite{Schwartz2006}\cite{Rosnes2008}.
Our study is restricted to erasing the symbols of one color out of $M$.
Hence, the rest of this paper is restricted to double diversity, $L=2$.
Absence of diversity is equivalent to $L=1$.
We establish now the root order $\rho$ of a symbol. For root symbols satisfying definitions
\ref{def_root_supersymbol} and \ref{def_root_symbol}, the root order is $\rho=1$.
For symbols that can be solved after two decoding iterations, we set $\rho=2$. The formal
definition of the root order $\rho$ can be written in the following recursive manner (for $\rho \ge 2$).\\

\begin{definition}
\label{def_root_order}
Let $\G^c$ be a compact graph of a product code, let $\phi \in \Phi(E^c)$ be an edge coloring,
and let $e \in E^c$ be a super-symbol. $e$ has {\em root order} $\rho(e)=\min(\rho_1, \rho_2)$ where:\\
1- Let $\upsilon_1 \in V^c_1$ be the column neighbor vertex of $e$. 
$\forall f$ adjacent to $e$ in $\upsilon_1$ and $\phi(f)=\phi(e)$, we have $\rho(f)<\rho_1$.\\
2- Let $\upsilon_2 \in V^c_2$ be the row neighbor vertex of $e$. 
$\forall f$ adjacent to $e$ in $\upsilon_2$ and $\phi(f)=\phi(e)$, we have $\rho(f)<\rho_2$.
\end{definition}
~\\
The previous definition implies that $\rho(e)=1$ if there exists no adjacent edge with the same color.
Also, for an edge $e$ that does not admit a finite $\rho(e)$, we set $\rho(e)=\infty$. 
When color $\phi(e)$ is erased, symbols belonging to the so-called stopping sets can never be solved
(even after an infinite number of decoding iterations) and hence their root order is infinite.
In the next section we review stopping sets as known in the literature and we study new stopping sets
for product codes based on MDS components under iterative algebraic decoding. 
Definition \ref{def_root_order} can be rephrased to make it suitable for the non-compact
graph $\G$. 
We pursue this section to establish an upper bound of the largest finite root order valid for all edge colorings $\phi$. 

\begin{theorem}
\label{th_root_order_max}
Let $C_P$ be a product code $[n_1,k_1] \otimes [n_2,k_2]$ with a compact graph $\G^c=(V^c_1, V^c_2, E^c)$.
$\forall \phi \in \Phi(E^c)$ and $\forall e \in E^c$ we have:\\
Case 1: $\nexists f \in E^c$ such that $\phi(f)=\phi(e)$ and $\rho(f)=\infty$, then 
\[
1 \le \rho(e) \le \left\lceil \frac{N^c}{2M} \right\rceil = \rho_u.
\]
Define the minimum number of good edges,
\[
\eta_{min}(\phi)= \min_{i=1 \ldots M} |\{ f \in E^c : \phi(f)=i, \rho(f)=1 \}|.
\]
Then, in Case 1,
\begin{equation}
\label{equ_rho_etamin}
2 \rho(e) + \eta_{min}(\phi) -3 \le \left\lceil \frac{N^c}{M} \right\rceil.
\end{equation}
Case 2: $\exists f \in E^c$ such that $\phi(f)=\phi(e)$ and $\rho(f)=\infty$, then 
\[
\rho(e)=\infty~~~or~~~ 1 \le \rho(e) \le \left\lceil \frac{N^c}{M} \right\rceil - 4,
\]
where $N^c=|E^c|$ is given by (\ref{equ_Nc}).\\
\end{theorem}
\begin{proof}
Case 1 corresponds to a product code with diversity $L=2$, for a given color $\phi(e)$,
which is capable of solving all symbols when that color is erased. 
The graph has no infinite root order symbols. 
$\rho$ is recursively built by starting from $\rho=1$ following two paths in the graph
until reaching a common edge $e$ that has two neighboring vertices with edges of order $\rho(e)-1$.
There are up to $\lceil N^c/M \rceil$ edges, including $e$, having color equal to $\phi(e)$. 
The largest $\rho(e)$ is attained in the middle of the longest path of length $\lceil N^c/M \rceil$, 
hence $2 \rho(e)-1 \le \lceil N^c/M \rceil$ which is translated into the stated result for Case 1.
An illustrated instance is given for the reader in Example~\ref{ex_path_12_10}.
Back to the path of length $2 \rho(e)-1$ ending with edges of order $1$ on both sides,
if the population of order $1$ edges is $\eta_1$ for the color $\phi(e)$, 
then the path can only use a maximum of $\lceil N^c/M \rceil-(\eta_1-2)$ edges.
We get the inequality $2 \rho(e)-1 \le \lceil N^c/M \rceil-(\eta_1-2)$.
By plugging $\eta_{min}(\phi)$ instead of $\eta_1$, 
this inequality becomes independent from the particular color.
The stated inequality in (\ref{equ_rho_etamin}) is obtained after grouping $\rho(e)$ 
and $\eta_{min}(\phi)$ on the left side.
\\
Case 2 corresponds to bad edge coloring where the product code does not have double diversity, i.e. 
stopping sets do exist for the color $\phi(e)$. The order of $e$ may be infinite if $e$ is involved
in a stopping set with another edge $f$ having the same color. Otherwise, consider the smallest
stopping set of size four symbols (the smallest cycle in $\G^c$ with edges of color $\phi(e)$),
then there remains $\lceil N^c/M \rceil-4$ edges of color $\phi(e)$. 
A path of length $\lceil N^c/M \rceil-4$ starting with $\rho=1$
and ending at $\rho=\infty$ may exist. 
The largest finite order in this path before reaching the stopping set is $\rho=\lceil N^c/M \rceil-4$.
\end{proof}

\begin{corollary}
\label{cor_rho_max}
Let $C_P$ be a product code $[n_1,k_1] \otimes [n_2,k_2]$ with a compact graph $\G^c$.
Let $\phi \in \Phi(E^c)$ be an edge coloring.
We define 
\begin{equation}
\rho_{max}(\phi)=\max_{e \in E^c} \rho(e).
\end{equation}
$C_P$ attains double diversity under iterative row-column decoding if and only if $\rho_{max}(\phi)<\infty$.
In this case, we say that $\phi$ is a double-diversity coloring
and $\forall e \in E^c$, $e$ can be solved after at most $\rho_{max}$ decoding iterations where
$\rho_{max}(\phi)\le \rho_u$.\\
\end{corollary}

For colorings in $\Phi(E)$, we extend the same definition as in Corollary~\ref{cor_rho_max}
and we say that $\phi \in \Phi(E)$ is double-diversity if all edges have a finite rootcheck order.
The parameter $\rho_{max}$ is important in practical applications to bound from above 
the amount of conveyed information within a network 
(whether it is a local-area or a wide-area network).
In fact, in coding for distributed storage, the locality of a product code per decoding iteration is $\max(n_1, n_2)$ in $\G$ under algebraic decoding of its row and column components. 
Here, the locality is the number of symbols to be accessed in order to repair an erased symbol~\cite{Gopalan2012}.
Locality is $\max(k_1, k_2)$ for MDS components under ML decoding of the product code components.
Finally, for a product code, the information transfer per symbol is bounded from above by 
\begin{equation}
\rho_{max}(\phi) \times \max(n_1, n_2).
\end{equation}
The exact transfer cost to fill all erasures with iterative decoding can be determined by multiplying
each order $\rho$ with the corresponding edge population size. This exact cost may vary in a wide range
from one coloring to another. The DECA algorithm presented in Section~\ref{sec_edge_coloring} dramatically
reduces $\rho_{max}$ by enlarging the edge population with root order $1$.
The interdependence between $\rho$  and the population of order~$1$ was revealed
in inequality (\ref{equ_rho_etamin}). This inequality is useful in intermediate cases
where $\rho_{max}=1$ is not attained, i.e. outside the case where all edges have order $1$.
The influence of the component decoding method on the performance
of a product code via its stopping sets is discussed in Section~\ref{sec_stopping_sets}.\\

\begin{example}
\label{ex_path_12_10}
Consider a $[12,10]^{\otimes 2}$ product code and a coloring $\phi$ with $M=4$ colors. 
The compact graph has $|E^c|=6\times 6$ edges. Instead of drawing $\G^c$,
we draw the $6\times 6$ compact matrix representation of the product code in Fig.~\ref{matrix_max_order6_6}.
Supersymbols corresponding to a color $\phi(e)=1$ are shaded.
Fig.~\ref{matrix_max_order6_6} also shows a path in $\G^c$ 
such that a maximal order $\rho_{max}=\rho_u=5$
is attained for $\phi(e)=1$. If $\phi$ has double diversity 
then $\rho_{max}$ will not exceed $\rho_u=5$ for all colors
$\phi(e) \in \{ 1, 2, \ldots, M\}$. Note that the parameters of this product code are such that 
$N^c/M-4$ is also equal to 5 for a $\phi$ with a diversity defect.
\end{example}

\begin{figure}[!h]
\begin{center}
\includegraphics[width=0.65\columnwidth]{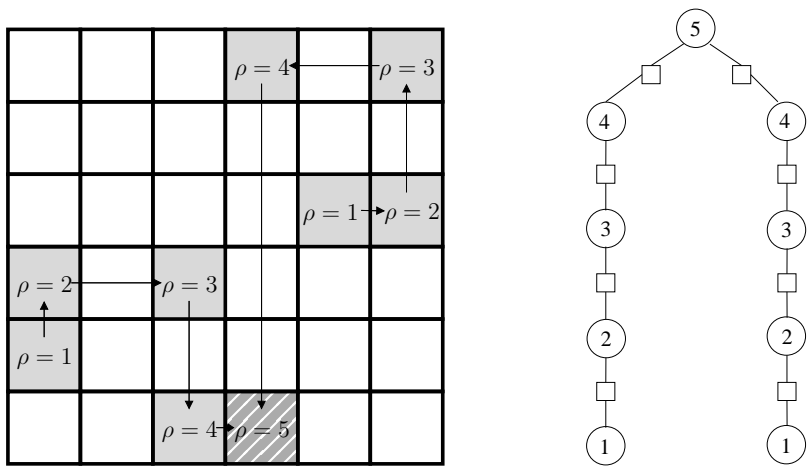}
\caption{Compact matrix (left) and path in compact graph (right) for a product code $[12,10]^{\otimes 2}$
showing a maximal root order of 5.\label{matrix_max_order6_6}}
\end{center}
\end{figure}

\begin{example}
Consider a $[14,12] \otimes [16,14]$ product code and a coloring $\phi$ with $M=4$ colors. 
The compact graph has $|E^c|=7\times 8$ edges. The compact matrix and a path attaining $\rho=10$
are illustrated in Fig.~\ref{matrix_max_order7_8_not_div2}. $\phi$ is chosen
such that the first color has a cycle involving four supersymbols.
Starting from the root supersymbol ($\rho=1$) it is possible to create a path in the graph
such that $\rho=10$ is reached. Note that a double-diversity $\phi$ cannot exceed a root order $\rho_u=7$.
\end{example}

\begin{figure}[!h]
\begin{center}
\includegraphics[width=0.85\columnwidth]{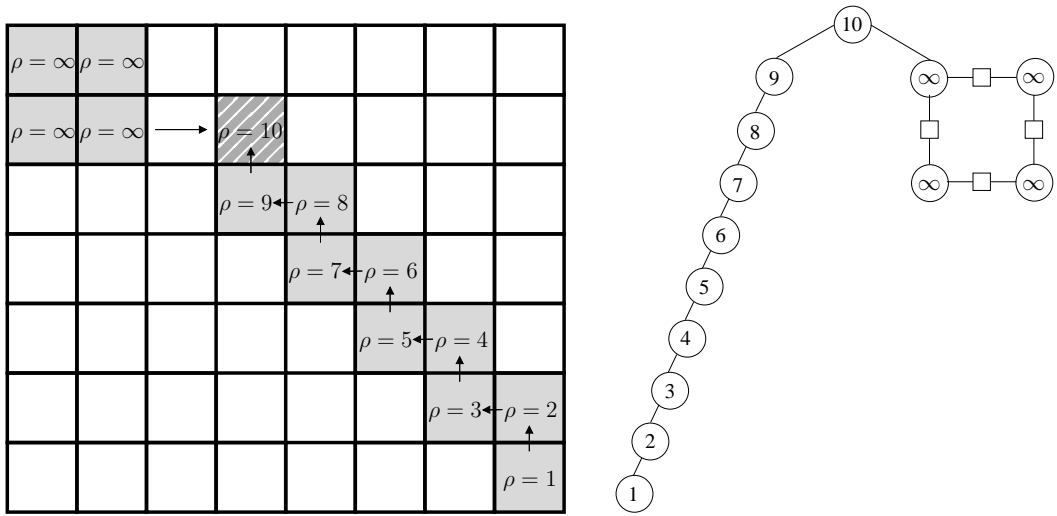}
\caption{Compact matrix (left) and path in compact graph (right) for a product code $[14,12] \otimes [16,14]$
showing a maximal finite root order of 10.\label{matrix_max_order7_8_not_div2}}
\end{center}
\end{figure}

The ideal situation is to construct a product code and its edge coloring in order 
to obtain $\rho(e)=1$ for all edges. We investigate now the conditions on the product code rate
and its components rates in this ideal situation. The analysis based on $\rho_u$
reveals the existence of a trade-off between minimizing the number of decoding iterations
and the valid range of both coding rates for the product code components.\\

Firstly, let us look at the upper bound $\rho_u$ from Theorem~\ref{th_root_order_max}.
Without loss of generality, assume that $n_i-k_i$ divides $n_i$.
Then, we have 
\begin{equation}
\label{equ_Ri_and_Vi}
R_i=1-\frac{n_i-k_i}{n_i}=1-\frac{1}{|V^c_i|}.
\end{equation}
The total coding rate becomes
\begin{equation}
R=R_1R_2=\left(1-\frac{1}{|V^c_1|}\right) \cdot \left(1-\frac{1}{|V^c_2|}\right).
\end{equation}
Using $N^c=|V^c_1| \cdot |V^c_2|$, we get
\begin{equation}
\label{equ_Nc_and_R}
R_1R_2=R_1+R_2-1+\frac{1}{N^c}.
\end{equation}
Finally, from (\ref{equ_Nc_and_R}) and Theorem~\ref{th_root_order_max}, 
the upper bound of the root order for double-diversity edge coloring
of the compact graph can be expressed as
\begin{equation}
\label{equ_rho_u_R1_R2}
\rho_u=\left\lceil \frac{N^c}{2M} \right\rceil
=\left\lceil \frac{1}{2M\times (1+R_1R_2-R_1-R_2)} \right\rceil.
\end{equation}
Fix the product code rate $R$, force the upper bound to $\rho_u=1$, and take $M=4$ colors. 
Then the denominator in (\ref{equ_rho_u_R1_R2}) should be less than $1$
or equivalently $-8R_1^2+(7+8R)R_1-8R~>~0$. 
This second-degree polynomial in $R_1$ is non-negative if and only if
\begin{equation}
\label{equ_R_0.41}
R < \frac{9}{8}-\frac{1}{\sqrt{2}} \approx 0.4178,
\end{equation}
and
\begin{equation}
-\sqrt{64R^2-144R+49} < 16R_1-8R-7 < +\sqrt{64R^2-144R+49}.
\end{equation}
As a result, with a palette of four colors,
(\ref{equ_R_0.41}) tells us that $\rho(e)=1$ for all edges is feasible
for a product code with a rate less than $0.4178$.
It is obvious that (\ref{equ_R_0.41}) 
is a very constraining condition because $\rho_u$ is an upper bound
of $\rho_{max}(\phi)$ for all $\phi \in \Phi(E^c)$. It is worth noting
that $R_1$ and $R_2$ vary in a smaller range when $R$ approaches $\frac{9}{8}-\frac{1}{\sqrt{2}}$,
which corresponds to a product code with balanced components.\\

In Section~\ref{sec_edge_coloring_sub1}, 
we will show unbalanced product codes where a sufficient condition
on the component rates imposes order $1$ to all edges.
The sufficient condition, not based on $\rho_u$, is given by Lemma~\ref{lem_rho=1}.
But before introducing an efficient edge coloring algorithm in Section~\ref{sec_edge_coloring},
we analyze stopping sets in product codes with MDS components in the next section, 
we describe the relationship between stopping sets and the product code graph representation,
and finally we enumerate obvious and non-obvious stopping sets. Stopping sets
enumeration is useful to determine the performance 
of a product code with and without edge coloring.
\section{Stopping sets for MDS components \label{sec_stopping_sets}}
The purpose of this section is to prepare the way for determining the performance
of iterative decoding of non-binary product codes. 
The analysis of stopping sets in a product code will yield a tight upper bound 
of its iterative decoding performance over a channel with independent erasures. 
The same analysis will be useful to accurately estimate the performance under 
edge coloring in presence of block and multiple erasure channels.
\subsection{Decoding erasures \label{sec_stopping_sets_sub1}}

\begin{definition}
An erasure pattern is said to be ML-correctable
if the ML decoder is capable of solving all its erased symbols.
\end{definition}
For an erasure pattern which is not correctable under ML or iterative decoding, 
the decoding process may fill none or some of the erasures and then stay stuck on the remaining ones. 
Before describing the stopping sets of a product code, let us recall some
fundamental results regarding the decoding of its row and column component codes.
The ML erasure-filling capability of a linear code satisfies the following property.
\begin{proposition}
\label{prop_not_MDS}
Let $C[n,k,d]_q$ be a linear code with $q \ge 2$. Assume that $C$ is not MDS and
the $n$ symbols of a codeword are transmitted on an erasure channel.
Then, there exists an erasure pattern of weight greater than $d-1$ that is ML-correctable.
\end{proposition}
\begin{proof}
Let $H$ be an $(n-k) \times n$ parity-check matrix of $C$ with rank $n-k > d-1$.
For any integer $w$ in the range $[d, n-k]$, there exists a set of $w$ linearly independent columns in $H$. 
Choose an erasure pattern of weight $w$ with erasures located at the positions of the $w$
independent columns. 
Then, the ML decoder is capable of solving all these erasures by simple Gaussian reduction of $H$.
\end{proof}
For MDS codes, based on a proof similar to the proof of Proposition~\ref{prop_not_MDS}, 
we state a well-known result in the following corollary. 
\begin{corollary}
\label{coro_not_MDS}
Let $C[n,k,d]_q$ be an MDS code. All erasure patterns of weight greater than $d-1$ are not ML-correctable.\\
\end{corollary}

We conclude from the previous corollary that an algebraic decoder for an MDS code
attains the word-error performance of its ML decoder. What about symbol-error performance?
Indeed, for general binary and non-binary codes, 
the ML decoder may outperform an algebraic decoder since
it is capable of filling some of the erasures when
dealing with a pattern which is not ML-correctable.
In the MDS case, the answer comes from the absence of spectral holes for any MDS code
beyond its minimum distance. This basic result is proven via standard tools from
algebraic coding theory \cite{MacWilliams1977}\cite{Blahut2003}:
\begin{proposition}
\label{prop_MDS_w}
Let $C[n,k,d]_q$ be a non-binary MDS code ($q>n>2$). For any $w$ satisfying $d \le w \le n$ 
and any support $\mathcal{X}=\{i_1, i_2, \ldots, i_w\}$, where $ 1 \le i_j \le n$,
there exists a codeword in $C$ of weight $w$ having $\mathcal{X}$ as its own support.
\end{proposition}
\begin{proof}
By assumption we have $w > r=n-k$. Let $H$ be a parity-check matrix of $C$ with rank $r=n-k$.
Recall that the MDS property makes full-rank any set of $n-k$ columns of $H$ \cite{MacWilliams1977}.
$w$ is written as $w=r+\ell$, where $\ell=1 \ldots k$. 
The $w$ positions of $\mathcal{X}$ are anywhere inside the range $[1,n]$,
but for simplicity let us denote $h_1 \ldots h_r$ the $r$ columns of $H$
in the first $r$ positions. The last $\ell$ columns are denoted 
$\zeta_1 \ldots \zeta_{\ell}$.
For any $j=1 \ldots \ell$, we have
\[
\zeta_j=\sum_{i=1}^{r} a_{i,j} h_i,
\]
where $a_{i,j} \in \F_q \setminus \{0\}$ otherwise it contradicts $d=n-k+1$.
Now, select $\alpha_1 \ldots \alpha_{\ell}$ from $\F_q \setminus \{0\}$ such that:
$\alpha_1$ is arbitrary, $\alpha_2$ is chosen outside the set $\{-\alpha_1a_{i,1}/a_{i,2} \}_{i=1}^r$,
then $\alpha_3$ is chosen outside the set $\{(-\alpha_1a_{i,1}-\alpha_2a_{i,2})/a_{i,3} \}_{i=1}^r$,
and so on, up to $\alpha_{\ell}$ which is chosen outside the set 
$\{-\sum_{u=1}^{\ell-1}\alpha_ua_{i,u}/a_{i,\ell} \}_{i=1}^r$.
Here, the notation $a/b$ in $\F_q \setminus \{0\}$ is equivalent to the standard
algebraic notation $ab^{-1}$.
The equality
\[
\sum_{j=1}^{\ell} \alpha_j \zeta_j = \sum_{i=1}^{r} \sum_{j=1}^{\ell}  \alpha_j a_{i,j} h_i
\]
produces a codeword of Hamming weight $w$.
Hence, there exists a codeword of weight $w$ with non-zero symbols in all positions given by $\mathcal{X}$. 
\end{proof}
Now, at the symbol level for an MDS code and an erasure pattern which is not ML-correctable ($w > d-1$),
we conclude from Proposition~\ref{prop_MDS_w} that the ML decoder cannot solve any of the $w$ erasures
because they are covered by a codeword. Consequently, an algebraic decoder for an MDS code
also attains the symbol-error performance of the ML decoder.
This behavior will have a direct
consequence on the iterative decoding of a product code with MDS components: 
stopping sets are identical when dealing with algebraic and ML-per-component decoders.\\

A general description of a stopping set was given by Definition~\ref{def_stopping_set}.
The exact definition of a stopping set depends on the iterative decoding type.
For product codes, four decoding methods are known:
\begin{itemize}
\item Type I: ML decoder. This is a non-iterative decoder. 
It is based on a Gaussian reduction of the parity-check matrix of the product code. 
\item Type II: Iterative algebraic decoder. At odd decoding iterations, 
component codes $C_1$ on each column are decoded via an algebraic
decoder (bounded-distance) that fills up to $d-1$ erasures.
Similarly, at even decoding iterations, 
component codes $C_2$ on each row are decoded via an algebraic decoder.
\item Type III: Iterative ML-per-component decoder. 
This decoder was considered by Rosnes in \cite{Rosnes2008} for binary product codes.
At odd decoding iterations, column codes $C_1$ are decoded via an optimal decoder (ML for $C_1$).
At even decoding iterations, row codes $C_2$ are decoded via a similar optimal decoder (ML for $C_2$).
\item Type IV: Iterative belief-propagation decoder based on the Tanner graph of $\CP$, 
as studied by Schwartz et al. for general linear block codes~\cite{Schwartz2006} 
and by Di et al. for low-density parity-check codes~\cite{Di2002}.
\end{itemize}
The three iterative decoders listed above give rise to three different kinds
of stopping sets. As previously indicated, from Corollary~\ref{coro_not_MDS} 
and Propositions~\ref{prop_MDS_w},
we concluded that type-II and type-III stopping sets are identical
if component codes are MDS.

\subsection{Stopping set definition\label{sec_stopping_sets_sub2}}

Let $C$ be a $q$-ary linear code of length $n$, i.e. $C$ is a sub-space of dimension $k$ of $\F_q^n$. 
The support of $C$, denoted by $\mathcal{X}(C)$, is the set of $\ell$ distinct 
positions $\{ i_1, i_2, \ldots, i_{\ell}\}=\{ i_j\}_{j=1}^{\ell}$,
$1~\le~i_j~\le ~n$, such that, for all $j$, there exists a codeword $c=(c_1 \ldots c_n) \in C$ with $c_{i_j} \ne 0$.
This notion of support $\mathcal{X}$ is applied to rows and columns in a product code.\\

Now, we define a rectangular support which is useful to represent a stopping set in a bi-dimensional
product code. Let $\mathcal{S} \subseteq \{1,\ldots,n_1\} \times \{1,\ldots,n_2\}$ be a set of symbol positions
in the product code. The set of row positions associated to $\mathcal{S}$ is 
$\mathcal{R}_1(\mathcal{S})=\{ i_1, \ldots, i_{\ell_1}\}$ where $|\mathcal{R}_1(\mathcal{S})|=\ell_1$
and for all $i \in \mathcal{R}_1(\mathcal{S})$ there exists $(i,\ell) \in \mathcal{S}$.
The set of column positions associated to $\mathcal{S}$ is 
$\mathcal{R}_2(\mathcal{S})=\{ j_1, \ldots, j_{\ell_2}\}$ where $|\mathcal{R}_2(\mathcal{S})|=\ell_2$
and for all $j \in \mathcal{R}_2(\mathcal{S})$ there exists $(\ell,j) \in \mathcal{S}$.
The rectangular support of $\mathcal{S}$ is 
\begin{equation}
\mathcal{R}(\mathcal{S})=\mathcal{R}_1(\mathcal{S}) \times \mathcal{R}_2(\mathcal{S}), 
\end{equation}
i.e. the smallest $\ell_1 \times \ell_2$ rectangle including all columns and all rows of $\mathcal{S}$.

\newpage
\begin{definition}
\label{def_exact_stop_sets}
Consider a product code $C_P=C_1 \otimes C_2$. Let $\mathcal{S} \subseteq \{1,\ldots,n_1\} \times \{1,\ldots,n_2\}$
with $|\mathcal{R}_1(\mathcal{S})|=\ell_1$ and $|\mathcal{R}_2(\mathcal{S})|=\ell_2$.
Consider the $\ell_1$ rows of $\mathcal{S}$ given by $\mathcal{S}_r^{(i)}=\{ j : (i,j) \in \mathcal{S} \}$
and the $\ell_2$ columns of $\mathcal{S}$ given by $\mathcal{S}_c^{(j)}=\{ i : (i,j) \in \mathcal{S} \}$. The set $\mathcal{S}$ is a stopping set of type~III for $C_P$
if there exist linear subcodes $C_c^{(j)} \subseteq C_1$ and $C_r^{(i)} \subseteq C_2$ such that
$\mathcal{X}(C_c^{(j)})=\mathcal{S}_c^{(j)}$ and $\mathcal{X}(C_r^{(i)})=\mathcal{S}_r^{(i)}$ for all
$i \in \mathcal{R}_1(\mathcal{S})$ and for all $j \in \mathcal{R}_2(\mathcal{S})$.
\end{definition}

The cardinality $|\mathcal{S}|$ is called the size of the stopping set and will also be referred
to in the sequel as the weight of $\mathcal{S}$.
Recall that type II and type III stopping sets are identical when both $C_1$ and $C_2$ are MDS.
Stopping sets of type III were studied for binary product codes by Rosnes \cite{Rosnes2008}.
His analysis is based on the generalized Hamming distance~\cite{Wei1991}\cite{Helleseth1992}
because sub-codes involved in Definition~\ref{def_exact_stop_sets} may have a dimension greater than 1.
In the non-binary MDS case, according to Proposition~\ref{prop_MDS_w}, all these sub-codes
have dimension 1, i.e. they are generated by a single non-zero codeword. Consequently,
the generalized Hamming distance is not relevant when using MDS components.
In such a case, the analysis of type II stopping sets is mainly combinatorial and
does not require algebraic tools.\\

Stopping sets for decoder types II-IV can be characterized by four main properties summarized as follows.
\begin{itemize}
\item Obvious or not obvious sets, also known as rank-1 sets. 
A stopping set $\mathcal{S}$ is obvious if $\mathcal{S}=\mathcal{R}(\mathcal{S})$.
\item Primitive or non-primitive stopping sets. 
A stopping set is primitive if it cannot be partitioned 
into two or more smaller stopping sets. Notice that all stopping sets, 
whether they are primitive or not, are involved in the code performance.
\item Codeword or non-codeword. A stopping set $\mathcal{S}$  is said to be a codeword stopping set
if there exists a codeword $c$ in $C_P$ such that $\mathcal{X}(c)=\mathcal{S}$.
\item ML-correctable or non-ML-correctable.
A stopping set $\mathcal{S}$ cannot be corrected via ML decoding if 
it includes the support of a non-zero codeword.\\ 
\end{itemize}

In the remaining material of this paper, we restrict our study to type II stopping sets.
\newpage
\begin{example}
\label{ex_w=9}
Consider a $[n_1,n_1-2,3]_q \otimes [n_2,n_2-2,3]_q$ product code. A stopping set $\mathcal{S}$
of size $w=9$ is shown as a weight-$9$ matrix of size $n_1 \times n_2$, where $1$
corresponds to an erased position:
\begin{equation}
\label{equ_stop_w=9_large}
\cS~=~\left(
\begin{array}{ccccccc}
0 & 0 & 0 & 0 & 0 & 0  \\
0 & 0 & 0 & 0 & 0 & 0  \\
0 & 1 & 0 & 1 & 1 & 0  \\
0 & 0 & 0 & 0 & 0 & 0  \\
0 & 1 & 0 & 1 & 1 & 0  \\
0 & 1 & 0 & 1 & 1 & 0  \\
0 & 0 & 0 & 0 & 0 & 0 
\end{array}
\right).
\end{equation}
We took $n_1=n_2=7$ for illustration. 
The rectangular support is shown in a compact representation 
as a matrix of size $ \ell_1 \times \ell_2 = 3 \times 3$,
\begin{equation}
\label{equ_stop_w=9}
\mathcal{R}(\mathcal{S}) ~=~
\left(
\begin{array}{ccc}
1 & 1 & 1 \\
1 & 1 & 1 \\
1 & 1 & 1 
\end{array}
\right).
\end{equation}
The stopping set in (\ref{equ_stop_w=9_large}) is obvious, 
it has the same size as its rectangular support. 
It corresponds to a matrix of rank 1. Each row and each column of $\mathcal{S}$
has weight $3$. Iterative row-column decoding based on component algebraic decoders 
fails in decoding rows and columns since the number of erasures exceeds 
the erasure-filling capacity of the MDS components. This stopping set is not
ML-correctable because it is a product-code codeword. 
In the sequel, all stopping sets (type II) shall be represented in this compact manner
by a smaller rectangle of size $\ell_1 \times \ell_2$.
\end{example}

\begin{example}
\label{ex_w=12}
For the same $[n_1,n_1-2,3]_q \otimes [n_2,n_2-2,3]_q$ product code used in the previous example,
the following stopping sets of size $12$ are not obvious.
\begin{equation}
\cS_1~=~\left(
\begin{array}{ccccccc}
0 & 0 & 0 & 0 & 0 & 0 & 0 \\
0 & 0 & 0 & 0 & 0 & 0 & 0 \\
0 & 1 & 1 & 1 & 0 & 0 & 0 \\
0 & 0 & 1 & 1 & 1 & 0 & 0 \\
0 & 1 & 0 & 1 & 1 & 0 & 0 \\
0 & 1 & 1 & 0 & 1 & 0 & 0 \\
0 & 0 & 0 & 0 & 0 & 0 & 0 
\end{array}
\right),
\end{equation}

\begin{equation}
\cS_2~=~\left(
\begin{array}{ccccccc}
0 & 0 & 0 & 0 & 0 & 0 & 0 \\
0 & 0 & 0 & 0 & 0 & 0 & 0 \\
0 & 1 & 0 & 1 & 1 & 0 & 0 \\
0 & 0 & 0 & 1 & 1 & 1 & 0 \\
0 & 1 & 0 & 0 & 1 & 1 & 0 \\
0 & 1 & 0 & 1 & 0 & 1 & 0 \\
0 & 0 & 0 & 0 & 0 & 0 & 0 
\end{array}
\right).
\end{equation}
In compact form, their rectangular support is
\begin{equation}
\label{equ_stop_w=12}
\mathcal{R}(\mathcal{S}_1) ~=~\mathcal{R}(\mathcal{S}_2) ~=~
\left(
\begin{array}{cccc}
1 & 1 & 1 & 0 \\
0 & 1 & 1 & 1 \\
1 & 0 & 1 & 1 \\
1 & 1 & 0 & 1 
\end{array}
\right).
\end{equation}
These stopping sets have size $12$ and a $4 \times 4$ rectangular support. 
For $w=12$, it is also possible
to build an obvious stopping set in a $3 \times 4$ rectangle or a $4 \times 3$ rectangle full of $1$.
$\cS_1$ is ML-correctable since it does not cover a product code codeword. 
$\cS_2$ covers a codeword hence it is not ML-correctable.
\end{example}


\subsection{Stopping sets and subgraphs of product codes\label{sec_stopping_sets_sub3}}
A stopping set as defined by Definition~(\ref{def_exact_stop_sets}) corresponds to erased
edges in the non-compact graph $\G$ introduced in Section~\ref{sec_graph_representations_sub1}.
Indeed, consider the size-$9$ stopping set given by (\ref{equ_stop_w=9_large}) 
or (\ref{equ_stop_w=9}). The nine symbol positions involve nine edges in $\G$, three
row checknodes, and three column checknodes. Each of these six checknodes has three erased
symbols making the $[12,10,3]$ decoder fail. This stopping set is equivalent
to a subgraph of $9$ edges in $\G$ as shown in Figure~\ref{fig_G_12_12_w=9}.

\begin{figure}[!h]
\begin{center}
\includegraphics[width=0.36\columnwidth]{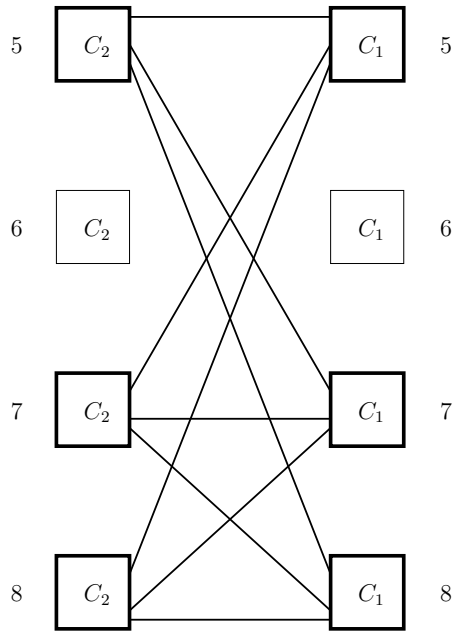}
\caption{A sub-graph of $\G$ representing the size-$9$ obvious stopping set. The graph $\G$
has $|E|=144$ edges, $|V_2|=12$ left (row) checknodes, and $|V_1|=12$ left (column) checknodes.
Only the stopping set edges are drawn.
\label{fig_G_12_12_w=9}}
\end{center}
\end{figure}

The subgraph in Figure~\ref{fig_G_12_12_w=9} has three length-$4$ cycles and two length-$6$ cycles.
The small cycles of length-$4$ are associated to an erasure pattern with a $2\times 2$ rectangular support which is not a stopping set ($d_1=d_2=3$). Similarly, length-$6$ cycles
are not stopping sets and are associated to erasure patterns with a $2\times 3$ rectangular support.
We will see in the next section that the minimum stopping set 
size is $d_1d_2=9$, i.e. it is equal to the minimum Hamming distance of the product code.

A subgraph of $\G^c$ can be embedded into $\G$ by splitting each super-edge 
into $(n_1-k_1)\times (n_2-k1)$ edges. The converse is not always true.
The subgraph with nine edges in Figure~\ref{fig_G_12_12_w=9} cannot be compressed
into a subgraph of $\G^c$. For the $[12,10,3]^{\otimes 2}$ product code, a supersymbol in $\G^c$
contains four edges. Hence, a necessary condition for a stopping set in $\G$
to become a valid stopping set in $\G^c$ is to erase edges in groups of $4$.
Knowing that type II and type III stopping sets are identical when row and column
codes $C_1$ and $C_2$ are MDS, 
Definition~(\ref{def_exact_stop_sets}) leads to the following corollaries.

\begin{corollary}
\label{cor_stop_in_G}
Let $C_P=C_1 \otimes C_2$ be a product code with MDS components $C_1$ and $C_2$
having minimum Hamming distance $d_1$ and $d_2$ respectively. Assume that symbols (edges)
of $\G=(V_1, V_2, E)$ are sent over an erasure channel. A stopping set for the iterative decoder
is a subgraph of $\G$ such that all column vertices in $V_1$ have a degree greater than or equal
to $d_1$ and all row vertices in $V_2$ have a degree greater than or equal to $d_2$.
\end{corollary}

\begin{corollary}
\label{cor_stop_in_Gc}
Let $C_P=C_1 \otimes C_2$ be a product code with MDS components $C_1$ and $C_2$
having minimum Hamming distance $d_1$ and $d_2$ respectively. Assume that supersymbols 
(super-edges) of $\G^c=(V^c_1, V^c_2, E^c)$ are sent over an erasure channel. 
A stopping set for the iterative decoder is a subgraph of $\G^c$ such that all 
column vertices in $V^c_1$ have a degree greater than or equal
to $2$ and all row vertices in $V_2$ have a degree greater than or equal to $2$.
\end{corollary}

The above corollaries suppose a symbol (or a supersymbol) channel with independent erasures.
When $\G$ is endowed with an edge coloring $\phi$, we get the same constraint on the validity
of a subgraph embedding from $\G^c$ into $\G$. 
We know from Section~\ref{sec_graph_representations_sub1} 
that $\Phi(E^c \rightarrow E)$ is a subset of $\Phi(E)$, 
i.e. some edge colorings of $\G$ are not edge colorings
of $\G^c$. Consequently, on a block-erasure channel, 
if all super-edges of the same color are erased, 
stopping sets in $\G^c$ are a subset of those in $\G$. 
The non-compact graph $\G$ has a larger ensemble of stopping sets, with or without edge coloring. 
As an example, for the  $[12,10,3]^{\otimes 2}$ product code,
the smallest stopping set in $\G^c$ has size $2 \times 2$ when four super-edges are erased
which yields a stopping set of size $16$ in $\G$.\\

\begin{example} Consider the $[9,6,4]_q^{\otimes 2}$ product code where $d_1=d_2=4$
and $q>9$. Assume that our palette has $M=3$ colors. 
The non-compact graph admits an ensemble of 
$|\Phi(E)|~=~4490186382903298862950669893074864640$ edge colorings! 
The compact graph has $|\Phi(E^c)|~=~1680$ only. In $\G^c$, each color is used $N^c/M=3$ times.
For a channel erasing all symbols of the same color,
the compact graph has no stopping sets (the $2\times 2$ rectangular support cannot be filled
by a single color). A compact matrix representation of $\G^c$ attaining double diversity
with all symbols of order 1 is given by the trivial matrix
\begin{equation}
\left[
\begin{array}{ccc}
R & G & B \\
B & R & G \\
G & B & R 
\end{array}
\right],
\end{equation}
where the color $\phi(e)=1$ is replaced by the letter 'R', $\phi(e)=2$ is replaced by the letter 'G',
and $\phi(e)=3$ is replaced by the letter 'B'. 
The non-compact graph has $9 \times 9$ edges, each color is used $27$ times.
Double diversity is lost in $\G$ if one of the $4 \times 4$, $4\times 5$, or $5\times 5$ obvious 
stopping sets is covered by a unique color. Clearly, $\G^c$ makes the design much easier.
This double-diversity product code has a relatively low coding rate.
More challenging product code designs are given in Section~\ref{sec_edge_coloring} 
with higher rates up to the one imposed by the block-fading/block-erasure Singleton bound.
\end{example}

\subsection{Enumeration of stopping sets\label{sec_stopping_sets_sub4}}
For a fixed non-zero integer $w$, the number of stopping sets of size $w$, denoted as $\tau_w$,
falls in two different cases. Firstly, $\tau_w~=~0$ if $w$ is small with respect to the minimum
Hamming distance of the product code. Also, $\tau_w=0$ for special erasure patterns obtained
by adding a small neighborhood to a smaller obvious set. Secondly, for both obvious and non-obvious stopping sets,
$\tau_w$ is non-zero and the weight $w$ may correspond to many rectangular supports of different height and width.
The code performance over erasure channels is dominated by not-so-large stopping sets.
Non-empty stopping sets of the second case satisfy the general property stated in the following lemma.
\begin{lemma}
\label{lem_max_support}
Given a weight $w \le (d_1+1)(d_2+1)$ and assuming $\tau_w >0$, then $\exists \mathcal{S}^0$ such that $\forall \mathcal{S}$ with $|\mathcal{S}|=w$, we have
$\|\mathcal{R}(\mathcal{S})\| \le \|\mathcal{R}(\mathcal{S}^0)\|=(\ell_1^0, \ell_2^0)$, where
\begin{align}
\label{equ_ell1_0} \ell_1^0 &\le d_1+1+\left\lfloor \frac{d_1+1}{d_2} \right\rfloor, \\
\label{equ_ell2_0} \ell_2^0 &\le d_2+1+\left\lfloor \frac{d_2+1}{d_1} \right\rfloor. 
\end{align}
\end{lemma}
\begin{proof}
Let $w$ be equal to $(d_1+1)(d_2+1)$. In order to establish an upper bound of the height $\ell_1$, 
we build the highest possible rectangular support for this weight $w$. Assume the rectangle
is $\ell_1^0 \times \ell_2$, each of its rows should have at least $d_2$ erasures to make
the type-II decoder fail. Then $d_2 \ell_1^0 \le (d_1+1)(d_2+1)$
which becomes the upper bound given by (\ref{equ_ell1_0}).
Now, if $w$ is less than $(d_1+1)(d_2+1)$, the rectangular support of the stopping set
can only shrink in size. The upper bound of the width in (\ref{equ_ell2_0}) is proven
in a similar way. 
\end{proof}
The above lemma states the existence of a maximal rectangular support for a given stopping set size.
The example given below cites stopping sets with a unique-size rectangular support
and stopping sets with multiple-size rectangular supports.
\begin{example}
Consider a $C_1 \otimes C_2$ product code where $C_1$ and $C_2$ are both MDS
with minimum Hamming distance $3$. The stopping set given by (\ref{equ_stop_w=9}) 
cannot have a large rectangular support.
In general, all stopping sets of size $d_1d_2$ have a rectangular support of fixed dimensions
$d_1 \times d_2$. Now, let $w=12$. As indicated in Example~\ref{ex_w=12}, stopping sets
of size $12$ may be included in rectangular supports of dimensions $3 \times 4$, $4 \times 3$,
and $4 \times 4$. For $w=12$, it is impossible to build a $4 \times 5$ rectangular support
(reductio ad absurdum) making $\ell_1^0=4$ and $\ell_2^0=4$. A similar proof by contradiction
yields $\ell_1^0=5$ and $\ell_2^0=5$ for $w=15$.
\end{example}
The next lemma gives an obvious upper bound of the size of $\cR(\cS)$ by stating a simple
limit on the number of zeros (non-erased positions) inside $\cR(\cS)$.
\begin{lemma}
\label{lem_max_zeros}
Let  $\cR(\cS)$ be the $\ell_1 \times \ell_2$ rectangular support of a stopping set $\cS$ of size $w$.
Let $\beta=\ell_1\ell_2-w$ be the number of zero positions, or equivalently $\beta$ is the size of the
set $\cR(\cS) \setminus \cS$. Then 
\begin{equation}
\beta \le \min((\ell_1-d_1)\ell_2, \ell_1(\ell_2-d_2)).
\end{equation}
\end{lemma}
Before stating and proving Theorem~\ref{th_stopping_sets_d}, we announce two results
in Lemma~\ref{lem_graph_bipartite_deg2} and Lemma~\ref{lem_graph_bipartite_deg2_1} 
on bipartite graphs enumeration. We saw in the previous section that stopping sets
are sub-graphs of $\G$ and $\G^c$, see Corollary~\ref{cor_stop_in_G} and Corollary~\ref{cor_stop_in_Gc}.
In other words, the enumeration of stopping sets represented as matrices of a given
distribution of row weight and column weight is equivalent to enumerating
bipartite graphs where left vertices stand for rows and right vertices stand for columns.
An edge should be drawn between a left vertex and a right vertex according to some rule,
e.g. the rule used in the previous section draws an edge in the bipartite graph
for each $1$ in the stopping set matrix. Stopping sets enumeration in the next theorem
is based on $\beta$, the number of zeros or the number of non-erased positions.
Hence, we shall use the opposite rule. A stopping set of weight $w$ and having a $\ell_1 \times \ell_2$
rectangular support shall be represented by a bipartite graph with $\ell_1$ left vertices,
$\ell_2$ right vertices, and a total of $\beta=\ell_1\ell_2-w$ edges. Notice
that these bipartite graphs have no length-2 cycles because parallel edges are forbidden.

For finite $\ell_1$ and $\ell_2$, given the left degree distribution and the right degree distribution,
there exists no exact formula for counting bipartite graphs. The best recent results 
are asymptotic in the graph size for sparse and dense matrices 
\cite{Canfield2009}~\cite{Greenhill2012} 
and cannot be applied in our enumeration. 
The following two lemmas solve two cases encountered in Theorem~\ref{th_stopping_sets_d}
for $w=d(d+2)$ and $w=(d+1)(d+1)$ both inside a $(d+2) \times (d+2)$ 
rectangular support. The definition of special partitions is required
before introducing the two lemmas.

\begin{definition}
\label{def_special_partition}
Let $\ell \ge 2$ be an integer. A {\em special partition} of length $j$ of $\ell$
is a partition defined by a tuple $(\ell_1, \ell_2, \ldots, \ell_j)$ 
such that its integer components satisfy:
\begin{itemize}
\item $\ell_1 \le \ell_2 \le \ldots \le \ell_j$.
\item $\sum_{i=1}^j \ell_i = \ell$.
\item $\ell_i \ge 2$, $\forall j$.
\item $1 \le j \le \ell/2$. 
\end{itemize}
A special partition shall be denoted by $((\ell_1,\ldots, \ell_j))$.
\end{definition}

\begin{definition}
The { \em group number} of a special partition, 
denoted by $\kappa=\kappa(\ell_1, \ell_2, \ldots, \ell_j)$,
is the number of different integers $\ell_j$, for $j=1 \ldots \ell/2$.
In other words, following set theory, the set including the $j$ integers $\ell_i$'s
is $\{\ell_{i_1}, \ell_{i_2}, \ldots, \ell_{i_{\kappa}} \}$.
The group number divides the partition of $\ell$ into $\kappa$ groups
where the $m^{th}$ group includes $\ell_{i_m}$ repeated $g_m$ times,
and $\sum_{m=1}^{\kappa} g_m = j$.
\end{definition}

\begin{lemma}
\label{lem_graph_bipartite_deg2}
Consider bipartite graphs defined as follows: $\ell$~left vertices,
$\ell$ right vertices, all vertices have degree 2, and no length-2 cycles
are allowed. For $\ell \ge 2$, 
the total number $x_{\ell}$ of such bipartite graphs is given by the expression
\begin{equation}
\label{equ_x_ell}
x_{\ell}=
\sum_{((\ell_1,\ldots, \ell_j))} \frac{1}{\prod_{m=1}^{\kappa(\ell_1, \ldots, \ell_j)} g_m!}
\prod_{k=1}^j \frac{\prod_{u=0}^{\ell_k-1} (\ell-\sum_{i=1}^{k-1} \ell_i-u)^2}{2 \ell_k}
\end{equation}
where $\sum_{((\ell_1,\ldots, \ell_j))}$ is a summation over all special partitions
of the integer $\ell$, $\kappa(\ell_1, \ldots, \ell_j)$ is the group number
of the special partition $((\ell_1, \ldots, \ell_j))$, and $g_m$ is the size
of the $m^{th}$ group.
\end{lemma}
\begin{proof}
Firstly, let us find the number of Hamiltonian bipartite graphs having $\ell_k$ 
left vertices, $\ell_k$ right vertices, all vertices of degree $2$,
and no length-$2$ cycles allowed. There are $(\ell_k!)^2$ ways to choose the order
of all left and right vertices. If the Hamiltonian cycle is represented by
a sequence of $2\ell_k$ integers corresponding to the $2\ell_k$ vertices of the bipartite
graph, then there are $2\ell_k$ ways to shift the Hamiltonian cycle without changing
the graph. Hence, the number of Hamiltonian bipartite graphs of degree $2$ is
\begin{equation}
\frac{(\ell_k!)^2}{2\ell_k}.
\end{equation}
Secondly, given the half-size $\ell$ of the bipartite graph stated in this lemma,
all special partitions of $\ell$ are considered. For a fixed special partition 
$((\ell_1, \ell_2, \ldots, \ell_j))$ the bipartite graph is decomposed into $j$
Hamiltonian graphs each of length $\ell_k$, $k=1 \ldots j$. 
The number of choices for selecting the vertices of the $j$ Hamiltonian graphs
is
\begin{equation}
\prod_{k=1}^j {\ell-\sum_{i=1}^{k-1} \ell_i \choose \ell_k}^2. 
\end{equation}
The above number should be multiplied by the number of Hamiltonian graphs
for each selection of vertices to get
\begin{equation}
\prod_{k=1}^j {\ell-\sum_{i=1}^{k-1} \ell_i \choose \ell_k}^2 
\frac{\left( \ell_k ! \right)^2}{2 \ell_k}.
\end{equation}
But for a given special partition, each group of size $g_m$ is creating $g_m!$ 
identical bipartite graphs. Hence, the final result for a fixed partition
becomes
\begin{equation}
\label{equ_xell_oneterm}
\frac{1}{\prod_{m=1}^{\kappa(\ell_1, \ldots, \ell_j)} g_m!}
\prod_{k=1}^j {\ell-\sum_{i=1}^{k-1} \ell_i \choose \ell_k}^2 
\frac{\left( \ell_k ! \right)^2}{2 \ell_k}.
\end{equation}
Then, $x_{\ell}$ is obtained by summing (\ref{equ_xell_oneterm}) over all special partitions
of the integer $\ell$ to yield
\begin{equation}
x_{\ell}=
\sum_{((\ell_1,\ldots, \ell_j))} \frac{1}{\prod_{m=1}^{\kappa(\ell_1, \ldots, \ell_j)} g_m!}
\prod_{k=1}^j {\ell-\sum_{i=1}^{k-1} \ell_i \choose \ell_k}^2 
\frac{\left( \ell_k ! \right)^2}{2 \ell_k}.
\end{equation}
The simplification of the factors $(\ell_k!)^2$ yields the expression stated by this lemma.
\end{proof}

\begin{lemma}
\label{lem_graph_bipartite_deg2_1}
Consider bipartite graphs defined as follows: $\ell$~left vertices,
$\ell$ right vertices, all left vertices have degree $2$ except one vertex of degree $1$, 
all right vertices have degree $2$ except one vertex of degree $1$, and finally 
no length-$2$ cycles are allowed. For $\ell \ge 3$, 
the total number $y_{\ell}$ of such bipartite graphs is
\begin{equation}
y_{\ell}=\ell^2 \cdot \left( (2\ell-1)\cdot x_{\ell-1} + (\ell-1)^2 \cdot x_{\ell-2} \right),
\end{equation}
where $x_{\ell}$ is determined via Lemma~\ref{lem_graph_bipartite_deg2} and $x_1=0$.
\end{lemma} 
\begin{proof}
Let the first $\ell-1$ left vertices and the first $\ell-1$ right vertices be of degree $2$.
There exists two ways to complete this bipartite graph such that the two remaining vertices
have degree~$1$. 
\begin{itemize}
\item Each of the $x_{\ell-1}$ sub-graphs has $2(\ell-1)$ edges. Break one edge into two edges
and connect them to the remaining left and right vertices, the number of such graphs
is $2(\ell-1)x_{\ell-1}$. Another set of $x_{\ell-1}$ bipartite graphs is built by directly
connecting the last two vertices together without breaking any edge in the upper sub-graph.
Now, we find  $2(\ell-1)x_{\ell-1}+x_{\ell-1}=(2\ell-1)x_{\ell-1}$ bipartite graphs. 
\item Fix a vertex among the $\ell-1$ upper left vertices and fix one among
the $\ell-1$ upper right vertices ($(\ell-1)^2$ choices). 
Consider a length-$2$ cycle including these two vertices.
One edge of this cycle can be broken into two edges and then attached to the degree-$1$
vertices at the bottom. The remaining $\ell-2$ left and right vertices may involve $x_{\ell-2}$ 
sub-graphs. Consequently, the number of graphs in this second case is $(\ell-1)^2x_{\ell-2}$.
\end{itemize}
The total number of bipartite graphs enumerated in the above cases is
\begin{equation}
\label{equ_yell_part}
(2\ell-1)x_{\ell-1}+(\ell-1)^2x_{\ell-2}.
\end{equation}
Finally, the degree-$1$ left vertex has $\ell$ choices
and so has the degree-$1$ right vertex. 
The number of graphs in (\ref{equ_yell_part}) should be multiplied by $\ell^2$.
\end{proof}

We make no claims about a possible generalization of Lemma~\ref{lem_graph_bipartite_deg2} and Lemma~\ref{lem_graph_bipartite_deg2_1} to finite bipartite graphs with higher vertex degrees. 
As mentioned before, for general degree distributions, results on 
enumeration of asymptotic bipartite graphs were published by Brendan McKay and his
co-authors \cite{Canfield2009}~\cite{Greenhill2012}.
Table~\ref{tab_seq_spec_parts} shows the number of special partitions for $\ell=2 \ldots 32$.
The number of standard partitions (the partition function) can be found by a recursion resulting from
the pentagonal number theorem \cite{Conway1996book}. To our knowledge, there exists no such recursion
for special partitions. The number of bipartite graphs under the assumptions
of Lemma~\ref{lem_graph_bipartite_deg2} and Lemma~\ref{lem_graph_bipartite_deg2_1}
is found in Table~\ref{tab_xell_yell} for a graph half-size up to $8$. 
Finally, we are ready to state and prove the first theorem on stopping sets enumeration.

\begin{table}[h!]
\begin{center}
\begin{tabular}{|l|}
\hline
1, 1, 2, 2, 4, 4, 7, 8, 12, 14, 21, 24, 34, 41, 55, 66, 88, 105, 137,\\
165, 210, 253, 320, 383, 478, 574, 708, 847, 1039, 1238, 1507  \\ \hline
\end{tabular}
\end{center}
\caption{\label{tab_seq_spec_parts}Sequence of the number of special partitions of the integer $\ell$, for $\ell=2\ldots 32$. Special partitions are described in Definition~\ref{def_special_partition}.
The sequence for standard partitions can be found in~\cite{A000041}.}
\end{table}
\vspace{-1cm}
\begin{table}[h!]
\begin{center}
\begin{tabular}{|c|c|c|c|c|c|c|c|}
\hline
$\ell$    & 2 & 3 &  4  & 5 & 6 & 7 & 8 \\ \hline \hline 
$x_{\ell}$ & 1 & 6  & 90  & 2040 & 67950 & 3110940 & 187530840\\ \hline 
$y_{\ell}$ & 0 & 45 & 816 & 22650 & 888840 & 46882710 & 3199593600\\ \hline
\end{tabular}
\end{center}
\caption{\label{tab_xell_yell}Number of bipartite graphs not including length-2 cycles
from Lemma~\ref{lem_graph_bipartite_deg2} and Lemma~\ref{lem_graph_bipartite_deg2_1}.}
\end{table}

In the sequel, the open interval between two real numbers $a$ and $b$ 
will be denoted $]a,b[$,
\[
]a,b[ ~=~ \{x \in \R : a < x < b\}.
\]

\begin{theorem}
\label{th_stopping_sets_d}
Let $C_P$ be a product code $[n_1,k_1,d_1]_q \otimes [n_2,k_2,d_2]_q$ built from row and column MDS
component codes, where the alphabet size $q$ is greater than $\max(n_1, n_2)$. 
Let $\tau_w$ be the number of stopping sets of size $w$. We write $\tau_w=\tau^a+\tau^b$, where
$\tau^a$ counts obvious stopping sets and $\tau^b$ counts non-obvious stopping sets.
Under (type-II) iterative algebraic decoding and for $d_1=d_2=d \ge 2$, 
stopping sets are characterized as follows:
\begin{itemize}
\itemsep=2mm
\item For $w<d^2$,
\[
\tau^a=\tau^b=0.
\]
\item For $w=d^2$,
\[
\tau^a={n_1 \choose d} {n_2 \choose d}, ~~~\tau^b=0.
\]
\item For $w\in ]d^2, d(d+1)[$,
\[
\tau^a=\tau^b=0.
\]
\item For $w=d(d+1)$,
\begin{align}
\tau^a & ={n_1 \choose d}  {n_2 \choose d+1} + {n_1 \choose d+1}  {n_2 \choose d}, 
\nonumber \\
\tau^b & =(d+1)!  {n_1 \choose d+1}  {n_2 \choose d+1}.\nonumber
\end{align}
\item For $w \in ~]d(d+1),d(d+2)[$.\\
Let us write $w=d^2+d+\lambda$, where $\lambda \in [1, d-1]$.
\begin{align}
\tau^a & =0, \nonumber \\
\tau^b & =(d+1-\lambda)!  {d+1 \choose \lambda}^2  {n_1 \choose d+1}  {n_2 \choose d+1}. \nonumber
\end{align}
\item For $w=d(d+2)$,
\begin{align}
\tau^a & ={n_1 \choose d}  {n_2 \choose d+2}+{n_1 \choose d+2}  {n_2 \choose d}, \nonumber \\
\tau^b & =(d+1)^2  {n_1 \choose d+1}  {n_2 \choose d+1}  \nonumber \\
~& +\sum_{2r_0+r_1=d} {d+1 \choose r_0} {d+1-r_0 \choose r_1} \frac{(d+2)!}{2^{r_2}}   
\left[ 
{n_1 \choose d+1}  {n_2 \choose d+2} + {n_1 \choose d+2}  {n_2 \choose d+1}
\right] \nonumber \\
~& + x_{d+2}  {n_1 \choose d+2}  {n_2 \choose d+2}, \nonumber
\end{align}
where $\sum_{2r_0+r_1=d}$ is a summation over $r_0$ and $r_1$, both being non-negative
and satisfying $2r_0+r_1=d$, $r_2=d+1-r_0-r_1$, and $x_{d+2}$ is determined 
from Lemma~\ref{lem_graph_bipartite_deg2}.
\item For $w=(d+1)(d+1)$ 
\begin{align}
\tau^a & = {n_1 \choose d+1}  {n_2 \choose d+1},  \nonumber \\
\tau^b & = \sum_{2r_0+r_1=d+1} {d+1 \choose r_0} {d+1-r_0 \choose r_1} \frac{(d+2)!}{2^{r_0}} 
\left[ 
{n_1 \choose d+1}  {n_2 \choose d+2} + {n_1 \choose d+2}  {n_2 \choose d+1}
\right] \nonumber \\
 & + y_{d+2}  {n_1 \choose d+2}  {n_2 \choose d+2}, \nonumber
\end{align}
where $y_{d+2}$ is determined from Lemma~\ref{lem_graph_bipartite_deg2_1}.
\end{itemize}
\end{theorem}
\begin{proof} For $w$ satisfying $d^2 \le w \le (d+1)^2$, the admissible size of $\cR(\cS)$ varies  
from $d^2$ up to $(d+2)^2$ as given by Lemma~\ref{lem_max_support}.
All cases stated in the theorem shall use the following sequence of $\cR(\cS)$
listed in the order of increasing size $\ell_1\ell_2$: $d^2$, $d(d+1)$, $d(d+2)$, $(d+1)^2$, $(d+1)(d+2)$,
and $(d+2)^2$. For these rectangular supports, the stopping set weight also has six
cases to be considered, where $w$ takes the following values (or ranges) in increasing order:
$w=d^2$, $w \in ]d^2,d(d+1)[$, $w=d(d+1)$, $w \in ]d(d+1),d(d+2)[$, $w=d(d+2)$, and $w=(d+1)^2$.\\

\begin{itemize}
\itemsep=3mm
\item The case $w < d^2$.\\
Consider a stopping set of size $w<d^2$. Its rectangular support $\cR(\cS)$ has size 
$\ell_1\ell_2 \ge w$. All columns should have a weight greater than or equal to $d$,
we find that $ w \ge d\ell_2$. Similarly, all rows must have a weight greater than or equal to $d$,
then $ w \ge d\ell_1$. By combining the two inequalities, we find $w^2 \ge d^2 \ell_1\ell_2 \ge d^2w$,
so we get $w \ge d^2$ which is a contradiction unless these stopping sets do not exist,
i.e. $\tau_w=0$ for $w < d^2$ under type II iterative decoding.
\item The case $w=d^2$.\\
We use similar inequalities as in the previous case.
We have $w=d^2 \ge d\ell_2$ because column decoding must fail. We obtain $\ell_2 \le d$.
In a symmetric way, $w=d^2 \ge d\ell_1$ because row decoding must fail. We obtain $\ell_1 \le d$.
But $\cR(\cS)$ cannot be smaller than $\cS$, i.e. we get $\ell_1=d$ and $\ell_2=d$.
We just proved that all stopping set of size $d^2$ are obvious. Their number is given by choosing
$d$ rows out of $n_1$ and $d$ columns out of $n_2$. 
\item The case $ d^2 < w < d(d+1)$.\\
Given that $\ell_1\ell_2 \ge w > d^2$, we get $\ell_1 \ge d$ and $\ell_2 \ge d$ 
since the support $\cR(\cS)$ is larger than a $d \times d$ rectangle, the latter being
the smallest stopping set as proven in the previous case.
Take $\ell_1=d$, then $\ell_2 \ge d+1$ because $w > d^2$. The weight of each column must be at least $d$
giving us $w \ge d\ell_2 \ge d(d+1)$, which is a contradiction unless $\tau_w=0$.
For $\ell_1>d$, the same arguments hold.
\item The case $w=d(d+1)$.
\begin{itemize}
\item The smallest $\cR(\cS)$ is $d\times (d+1)$ or $(d+1)\times d$. According to Lemma~\ref{lem_max_zeros}, we have $\beta=0$. All these stopping sets are obvious.
Their number is
\[
{n_1 \choose d}  {n_2 \choose d+1}+{n_1 \choose d+1}  {n_2 \choose d}.
\]
\item $\cR(\cS)$ has size $d(d+2)$. Each column must have at least $d$ erasures. 
Then $\cS$ can only be obvious with weight $d(d+2)$ which contradicts $w=d(d+1)$.
Hence, this size of rectangular support yields no stopping sets, $\tau_w=0$ in this sub-case.
\item $\cR(\cS)$ has size $(d+1)(d+1)$. Let $\beta$ be the number of zeros in $\cR(\cS)$,
$\beta=(d+1)^2-w=d+1$.
All these stopping sets are found by considering the $(d+1)!$ permutations where a unique $0$
is placed per row and per column. Then, the binomial coefficient must be multiplied 
by $(d+1)!$ which yields the $\tau^b$ announced in the theorem for $w=d(d+1)$.
\item $\cR(\cS)$ has size $(d+1)(d+2)$. The number of zeros is $\beta=(d+1)(d+2)-w=2d+2$. 
Then $\beta > d+2 = (\ell_1-d)\ell_2$ which contradicts Lemma~\ref{lem_max_zeros}.
We get $\tau_w=0$ in this sub-case.
The same arguments are valid for larger rectangles.
\end{itemize}
\item The case $d(d+1) < w < d(d+2)$.\\
Let us write $w=d^2+d+\lambda$, where $\lambda \in [1, d-1]$. 
We consider below three sub-cases corresponding to admissible sizes of $\cR(\cS)$.
\begin{itemize}
\item The smallest $\cR(\cS)$ is $d\times (d+2)$ or $(d+2)\times d$. Take the rectangle
of size $d\times (d+2)$.
Each column must have at least $d$ erasures. Then $\cS$ can only be obvious with weight $d(d+2)$ which
is outside the range for $w$ in this case. 
Hence, this size of the rectangular support yields
no stopping sets, $\tau_w=0$ in this sub-case.
\item $\cR(\cS)$ has size $(d+1)\times (d+1)$. 
The number of zeros is $\beta=(d+1)^2-w=d+1-\lambda$, where $\beta \in [2,d]$.\\
Put the zeros in $\cR(\cS)$ not exceeding one per column and not exceeding one per row. 
The enumeration of these stopping sets
is given by selecting the $\beta$ rows and the $\beta$ columns, then filling all $\beta \times \beta$
permutation matrices in the zero positions. Hence, given that ${d+1 \choose \beta}={d+1 \choose \lambda}$,
we get for this sub-case
\[
\tau_w=\beta!  {d+1 \choose \lambda}^2  {n_1 \choose d+1}  {n_2 \choose d+1}. 
\]
All corresponding stopping sets are not obvious (the rank is greater than 1).
\item $\cR(\cS)$ has size $(d+1)(d+2)$. The number of zeros is 
$\beta=(d+1)(d+2)-w=2d+2-\lambda \in [d+3, 2d+1]$. Then $\beta > d+2=(\ell_1-d)\ell_2$
which contradicts Lemma~\ref{lem_max_zeros}.
In a similar way, it can be proven that $\tau_w=0$ in the sub-case $\cR(\cS)$ with size $(d+2)^2$.
\end{itemize}
\item The case $w=d(d+2)$.\\
The admissible rectangular support can have four sizes:
$d(d+2)$, $(d+1)(d+1)$, $(d+1)(d+2)$, and $(d+2)(d+2)$.
\begin{itemize}
\item $\cR(\cS)$ has size $d(d+2)$. 
According to Lemma~\ref{lem_max_zeros}, we have $\beta=0$. All these stopping sets are obvious.
Their number is
\[
{n_1 \choose d}  {n_2 \choose d+2}+{n_1 \choose d+2}  {n_2 \choose d}.
\]
\item $\cR(\cS)$ has size $(d+1)(d+1)$. We have $\beta=1$. The number of these
stopping sets is
\[
(d+1)^2  {n_1 \choose d+1}  {n_2 \choose d+1}.
\]
\item $\cR(\cS)$ has size $(d+1)(d+2)$. The number of zeros is $\beta=d+2$. 
Each column must have
a unique zero and each row cannot have more than two zeros. 
Let $r_i$ be the number of rows containing $i$ zeros, $i=0, 1, 2$. 
Then $r_0+r_1+r_2=d+1$ and $\beta=2r_2+r_1$, so the constraint is $2r_0+r_1=d$.
Given a stopping set satisfying this constraint, a permutation can be applied
on the $(d+2)$ columns to create another stopping set.
But a row with two zeros creates two identical columns, so the number of stopping
sets should be divided by $2^{r_2}$, where $r_2=d+1-r_0-r_1$.
The number of stopping sets in this sub-case is
\[
 \sum_{2r_0+r_1=d} {d+1 \choose r_0} {d+1-r_0 \choose r_1} \frac{(d+2)!}{2^{r_2}} 
\left[ 
{n_1 \choose d+1}  {n_2 \choose d+2} + {n_1 \choose d+2}  {n_2 \choose d+1}
\right]. 
\]
\item $\cR(\cS)$ has size $(d+2)(d+2)$. We have $\beta=2d+4$ reaching
the upper bound in Lemma~\ref{lem_max_zeros}. $\cR(\cS)$ must have
two zeros in each column and two zeros in each row. 
A first group of these stopping sets can be enumerated
by building $\cR(\cS)$ with two zero length-$(d+2)$ diagonals (to be folded
if not the main diagonal) and then
applying all row and column permutations. This generates all Hamiltonian
bipartite graphs with $d+2$ left vertices and $d+2$ right vertices,
their number is 
\[
\frac{\left( (d+2)! \right)^2}{2(d+2)},
\]
as known from Lemma~\ref{lem_graph_bipartite_deg2}.
In fact, the full exact enumeration of stopping sets in this case  
is already made by Lemma~\ref{lem_graph_bipartite_deg2} and its proof,
just take $\ell=d+2$. 
Then, in this sub-case, the number of stopping sets is given by
\[
x_{d+2}  {n_1 \choose d+2}  {n_2 \choose d+2}.
\]
\end{itemize}
\item The case $w=(d+1)(d+1)$.\\
The admissible rectangular support can have three possible sizes
$(d+1)(d+1)$, $(d+1)(d+2)$, and $(d+2)(d+2)$.
\begin{itemize}
\item $\cR(\cS)$ has size $(d+1)(d+1)$. We have $\beta=0$, i.e. $\cR(\cS)=\cS$.
The number of these obvious stopping sets is
\[
{n_1 \choose d+1}  {n_2 \choose d+1}.
\]
\item $\cR(\cS)$ has size $(d+1)(d+2)$. We have $\beta=d+1$. 
A column of $\cR(\cS)$ should contain at most one zero
and a row should contain at most two zeros.
Let $r_i$ be the number of rows containing $i$ zeros, $i=0, 1, 2$. 
Then $r_0+r_1+r_2=d+1$ and $\beta=2r_2+r_1$, so the constraint is $2r_0+r_1=d+1$.
Given a stopping set satisfying this constraint, a permutation can be applied
on the $(d+2)$ columns to create another stopping set.
The number of stopping sets in this sub-case is 
\[
 \sum_{2r_0+r_1=d+1} {d+1 \choose r_0} {d+1-r_0 \choose r_1} \frac{(d+2)!}{2^{r_2}}  
\left[ 
{n_1 \choose d+1}  {n_2 \choose d+2} + {n_1 \choose d+2}  {n_2 \choose d+1}
\right], 
\]
where $r_2=r_0$.
\item $\cR(\cS)$ has size $(d+2)(d+2)$. $\beta=2d+3$ which is less than the upper
bound in Lemma~\ref{lem_max_zeros}. These stopping sets are equivalent to bipartite
graphs considered in Lemma~\ref{lem_graph_bipartite_deg2_1}. 
Then, in this sub-case, the number of stopping sets is given by
\[
y_{d+2}  {n_1 \choose d+2}  {n_2 \choose d+2}.
\]
\end{itemize}

\end{itemize}
\end{proof}
From the proof of Theorem~\ref{th_stopping_sets_d}, in the case $w=d(d+2)$
with a $(d+1)\times (d+2)$ rectangular support, the enumeration of stopping sets
is directly converted into enumeration of trivial bipartite graphs defined by:
a- $\ell$ left vertices and a left degree $0$, $1$, or $2$, and b- $\ell+1$ right vertices all of degree~$1$.
Similarly, the proof for the case $w=d(d+2)$ with a $(d+1)\times (d+2)$ rectangular support
is directly related to the enumeration of bipartite graphs with one edge less.

\begin{theorem}
\label{th_stopping_sets_d1_d2}
Let $C_P$ be a product code $[n_1,k_1,d_1]_q \otimes [n_2,k_2,d_2]_q$ built from row and column MDS
components, where the alphabet size $q$ is greater than $\max(n_1, n_2)$. 
Let $\tau_w$ be the number of stopping sets of Hamming weight~$w$. 
We write $\tau_w=\tau^a+\tau^b$, where
$\tau^a$ counts obvious stopping sets and $\tau^b$ counts non-obvious stopping sets.
It is assumed that $2 < d_1 < d_2 < 3d_1-1$ or $2=d_1 < d_2 < 4d_1-1$. 
Under iterative algebraic decoding, stopping sets are characterized as follows.
\begin{itemize}
\itemsep=2mm
\item For $w<d_1d_2$,
\[
\tau^a=\tau^b=0.
\]
\item For $w=d_1d_2$,
\[
\tau^a ={n_1 \choose d_1}  {n_2 \choose d_2}, ~~~\tau^b=0.
\]
\item For $w\in ]d_1d_2, d_1(d_2+1)[$,
\[
\tau^a=\tau^b=0.
\]
\item For $w=d_1(d_2 + 1)$,
\[
\tau^a ={n_1 \choose d_1}  {n_2 \choose d_2+1}, ~~~\tau^b=0. 
\]
\end{itemize}
For larger weights, the enumeration of stopping sets distinguishes
three cases: A, B, and C.\\
\noindent
{\bf Case A: $d_2 <  2d_1$.}
\begin{itemize}
\itemsep=2mm
\item For $w\in ]d_1(d_2 + 1),(d_1+1)d_2[$.\\
\[
\tau^a=\tau^b=0.
\]
\item For $w=(d_1+1)d_2$.\\
\begin{align}
\tau^a &={n_1 \choose d_1+1}  {n_2 \choose d_2}, \nonumber \\
\tau^b &=(d_1+1)!  {d_2+1 \choose d_2-d_1}  {n_1 \choose d_1+1}  {n_2 \choose d_2+1}. \nonumber
\end{align}
\item For $w\in ](d_1+1)d_2,d_1(d_2+2)[$, write $w=(d_1+1)d_2+\lambda$.\\
\[
\tau^b=\mathds{1}_{\{d_2<2d_1-1\}} \times
(d_1+1-\lambda)!  {d_1+1 \choose \lambda} {d_2+1 \choose d_1+1-\lambda} 
{n_1 \choose d_1+1}  {n_2 \choose d_2+1}.
\]
\item For $w=d_1(d_2+2)$.\\
\begin{align}
\tau^a &= {n_1 \choose d_1}  {n_2 \choose d_2+2}, \nonumber \\
\tau^b &= (d_2-d_1+1)! {d_1+1 \choose d_2-d_1+1}  {d_2+1 \choose d_2-d_1+1}{n_1 \choose d_1+1}  {n_2 \choose d_2+1} \nonumber \\
& + \sum_{2r_0+r_1=2d_1-d_2} {d_1+1 \choose r_0} {d_1+1-r_0 \choose r_1} \frac{(d_2+2)!}{2^{r_0+d_2-d_1+1}} {n_1 \choose d_1+1}  {n_2 \choose d_2+2}. \nonumber 
\end{align}

\item For $w\in ]d_1(d_2+2), (d_1+1)(d_2+1)[$, write $w=d_1(d_2+2)+\lambda$.\\
\begin{align}
\tau^b &=(d_2-d_1+1-\lambda)! {d_1+1 \choose 2d_1-d_2+\lambda}  {d_2+1 \choose d_1+\lambda} {n_1 \choose d_1+1}  {n_2 \choose d_2+1} \nonumber \\
& + \sum_{2r_0+r_1=2d_1-d_2+\lambda} {d_1+1 \choose r_0} {d_1+1-r_0 \choose r_1} \frac{(d_2+2)!}{2^{r_2}  \lambda!} {n_1 \choose d_1+1}  {n_2 \choose d_2+2}, \nonumber
\end{align}
where $r_2=r_0+\lambda-d_1-1$.
\item For $w=(d_1+1)(d_2+1)$.\\
\begin{align}
\tau^a&={n_1 \choose d_1+1}  {n_2 \choose d_2+1}
+ \mathds{1}_{\{d_2=2d_1-1\}} {n_1 \choose d_1}  {n_2 \choose d_2+3} 
+ \mathds{1}_{\{d_2=d_1+1\}} {n_1 \choose d_1+2}  {n_2 \choose d_2}, \nonumber \\
\tau^b&=\sum_{2r_0+r_1=d_1+1} {d_1+1 \choose r_0} {d_1+1-r_0 \choose r_1}
\frac{(d_2+2)!}{2^{r_0}(d_2-d_1+1)!} {n_1 \choose d_1+1}  {n_2 \choose d_2+2} \nonumber \\
& +\mathds{1}_{\{d_2=d_1+1\}} \times  
\sum_{2r_0+r_1=d_2+1} {d_2+1 \choose r_0} {d_2+1-r_0 \choose r_1} 
\frac{(d_1+2)!}{2^{r_0}} {n_1 \choose d_1+2}  {n_2 \choose d_2+1} \nonumber \\
& + \mathds{1}_{\{d_2=2d_1-1\}} \times 
\sum_{3r_0+2r_1+r_2=d_1+1} {d_1+1 \choose r_0} {d_1+1-r_0 \choose r_1} {d_1+1-r_0-r_1 \choose r_2}\times \nonumber \\
&  ~~~~~~~~~~~~~~~~~~~~~~~~~~~~~~~~~~~~~~~~~~~~~ \frac{(d_2+3)!}{2^{r_2}6^{r_3}} {n_1 \choose d_1+1}  {n_2 \choose d_2+3} \nonumber \\
& +\mathds{1}_{\{d_2=d_1+1\}} 
\left((d_2+2)x_{d_1+2} + \frac{(d_2+2)y_{d_1+2}}{2} \right) 
{n_1 \choose d_1+2}  {n_2 \choose d_2+2} \nonumber \\
& + \mathds{1}_{\{d_1=2,d_2=3\}}  \times 1860  {n_1 \choose d_1+2}  {n_2 \choose d_2+3}.\nonumber
\end{align}
where $r_3=d_1+1-r_0-r_1-r_2$, and $x_{d_1+2}$ and $y_{d_1+2}$ are determined 
from Lemma~\ref{lem_graph_bipartite_deg2} and Lemma~\ref{lem_graph_bipartite_deg2_1} respectively.
\end{itemize}

\noindent
{\bf Case B: $d_2 =  2d_1$.}
\begin{itemize}
\itemsep=2mm
\item For $w\in ]d_1(d_2 + 1),(d_1+1)d_2[$.\\
\[
\tau^a=\tau^b=0.
\]
\item For $w=(d_1+1)d_2=d_1(d_2+2)$.\\
\begin{align}
\tau^a &={n_1 \choose d_1+1} {n_2 \choose d_2} + {n_1 \choose d_1} {n_2 \choose d_2+2}, \nonumber \\
\tau^b &= (d_1+1)! {d_2+1 \choose d_1+1} {n_1 \choose d_1+1} {n_2 \choose d_2+1} 
+ \frac{(d_2+2)!}{2^{d_1+1}} {n_1 \choose d_1+1} {n_2 \choose d_2+2}. \nonumber
\end{align}
\item For $w\in ](d_1+1)d_2,d_1(d_2+3)[$, write $w=(d_1+1)d_2+\lambda$.\\
\begin{align}
\tau^b &= (d_1+1-\lambda)! {d_1+1 \choose \lambda} {d_2+1 \choose d_1+\lambda} {n_1 \choose d_1+1} {n_2 \choose d_2+1} \nonumber \\
&+ \sum_{2r_0+r_1=\lambda} {d_1+1 \choose r_0} {d_1+1-r_0 \choose r_1} 
\frac{(d_2+2)!}{2^{r_2} \lambda!} {n_1 \choose d_1+1} {n_2 \choose d_2+2}, \nonumber
\end{align}
where $r_2=d_1+1-r_0-r_1=d_1+1+r_0-\lambda$.
\item For $w=d_1(d_2+3)$.\\
\begin{align}
\tau^a &= {n_1 \choose d_1} {n_2 \choose d_2+3},\nonumber \\
\tau^b &= (d_1+1) (d_2+1) {n_1 \choose d_1+1} {n_2 \choose d_2+1} \nonumber \\
& + \sum_{2r_0+r_1=d_1} {d_1+1 \choose r_0} {d_1+1-r_0 \choose r_1} 
\frac{(d_2+2)!}{2^{r_2} d_1!} {n_1 \choose d_1+1} {n_2 \choose d_2+2} \nonumber \\
& + \sum_{3r_0+2r_1+r_2=d_1} {d_1+1 \choose r_0} {d_1+1-r_0 \choose r_1} 
{d_1+1-r_0-r_1 \choose r_2} \frac{(d_2+3)!}{2^{r_2}6^{r_3}} {n_1 \choose d_1+1} {n_2 \choose d_2+3}, \nonumber
\end{align}
where $r_3=d_1+1-r_0-r_1-r_2=2r_0+r_1+1$.
\item For $w=(d_1+1)(d_2+1)$.\\
\begin{align}
\tau^a &={n_1 \choose d_1+1} {n_2 \choose d_2+1}, \nonumber \\
\tau^b &=\sum_{2r_0+r_1=d_1+1} {d_1+1 \choose r_0} {d_1+1-r_0 \choose r_1} 
\frac{(d_2+2)!}{2^{r_0}(d_1+1)!} {n_1 \choose d_1+1} {n_2 \choose d_2+2} \nonumber \\
& + \sum_{3r_0+2r_1+r_2=d_1+1} {d_1+1 \choose r_0} {d_1+1-r_0 \choose r_1}  {d_1+1-r_0-r_1 \choose r_2}\frac{(d_2+3)!}{2^{r_2}6^{2r_0+r_1}} {n_1 \choose d_1+1} {n_2 \choose d_2+3}. \nonumber 
\end{align}
\end{itemize}

\noindent
{\bf Case C: $4 < 2d_1 <  d_2 < 3d_1-1$\\
or $4=2d_1 <  d_2 < 4d_1-1$.}
\begin{itemize}
\itemsep=2mm
\item For $w\in ]d_1(d_2 + 1),d_1(d_2+2)[$.\\
\[
\tau^a=\tau^b=0.
\]

\item For $w=d_1(d_2+2)$.\\
\[
\tau^a={n_1 \choose d_1} {n_2 \choose d_2+2}.
\]

\item For $w\in ]d_1(d_2+2),(d_1+1)d_2[$.\\
\[
\tau^a=\tau^b=0.
\]

\item For $w=(d_1+1)d_2$.\\
\begin{align}
\tau^a &={n_1 \choose d_1+1}  {n_2 \choose d_2},  \nonumber \\
\tau^b &=(d_1+1)!  {d_2+1 \choose d_1+1}  {n_1 \choose d_1+1}  {n_2 \choose d_2+1} 
 + \frac{(d_2+2)!}{2^{d_1+1}(d_2-2d_1)!}  {n_1 \choose d_1+1}  {n_2 \choose d_2+2} \nonumber \\
& + \mathds{1}_{\{d_1=2,d_2=6\}} \times 1680 {n_1 \choose d_1+1}  {n_2 \choose d_2+3}. \nonumber 
\end{align}
\newpage
\item For $w\in ](d_1+1)d_2 , d_1(d_2+3)[$, write $w=(d_1+1)d_2 +\lambda$.\\
\begin{align}
\tau^b & = \mathds{1}_{\{d_1>2\}}  \times \Bigg[ (d_1+1-\lambda)!   {d_2+1 \choose d_1+1-\lambda}  {d_1+1 \choose \lambda} {n_1 \choose d_1+1}  {n_2 \choose d_2+1} \nonumber\\
& + \sum_{2r_0+r_1=\lambda}  {d_1+1 \choose r_0} {d_1+1-r_0 \choose r_1} \frac{(d_2+2)!}{2^{r_2}(d_2-2d_1+\lambda)!}  {n_1 \choose d_1+1}  {n_2 \choose d_2+2} \Bigg]. \nonumber
\end{align}
where $r_2=d_1+1+r_0-\lambda$.  

\item For $w=d_1(d_2+3)$.
\begin{align}
\tau^a & ={n_1 \choose d_1} {n_2 \choose d_2+3},\nonumber \\
\tau^b & =\mathds{1}_{\{d_1=2,d_2=6\} \bigcup \{ d_1>2 \} } \times
\Bigg[ (d_2-2d_1+1)! {d_1+1 \choose 3d_1-d_2} {d_2+1 \choose 2d_1} {n_1 \choose d_1+1} {n_2 \choose d_2+1} \nonumber \\
&+\sum_{2r_0+r_1=3d_1-d_2}  {d_1+1 \choose r_0} {d_1+1-r_0 \choose r_1} \frac{(d_2+2)!}{2^{r_2}d_1!}  {n_1 \choose d_1+1} {n_2 \choose d_2+2} \nonumber\\
&+\sum_{3r_0+2r_1+r_2=3d_1-d_2} {d_1+1 \choose r_0} {d_1+1-r_0 \choose r_1} 
  {d_1+1-r_0-r_1 \choose r_2}\frac{(d_2+3)!}{2^{r_2}6^{r_3}}
  {n_1 \choose d_1+1} {n_2 \choose d_2+3}\Bigg], \nonumber
\end{align}
where $r_2=d_2-2d_1+1+r_0$ and $r_3=d_1+1-r_0-r_1-r_2$.

\item For $w=d_1(d_2+4)$.\\
\[
\tau^a=\mathds{1}_{\{d_1=2\}} \times { n_1 \choose d_1} {n_2 \choose d_2+4}.
\]

\item For $w\in ]d_1(d_2+3) , (d_1+1)(d_2+1)[$, write $w=d_1(d_2+3) +\lambda$.\\
\begin{align}
\tau^b &=(d_2-2d_1+1-\lambda)!  {d_1+1 \choose 3d_1-d_2+\lambda}  {d_2+1 \choose 2d_1+\lambda} 
{n_1 \choose d_1+1} {n_2 \choose d_2+1} \nonumber \\
&+\sum_{2r_0+r_1=3d_1-d_2+\lambda}  {d_1+1 \choose r_0} {d_1+1-r_0 \choose r_1} 
\frac{(d_2+2)!}{2^{r_2}(d_1+\lambda)!} {n_1 \choose d_1+1}  {n_2 \choose d_2+2} \nonumber \\
&+\sum_{3r_0+2r_1+r_2=3d_1-d_2+\lambda} {d_1+1 \choose r_0} {d_1+1-r_0 \choose r_1} {d_1+1-r_0-r_1 \choose r_2}\frac{(d_2+3)!}{2^{r_2}6^{r_3}\lambda!}{n_1 \choose d_1+1}{n_2 \choose d_2+3} \nonumber \\
&+\mathds{1}_{\{d_1=2,d_2=6\}} \times 22050{n_1 \choose d_1+1}{n_2 \choose d_2+4}, \nonumber
\end{align}
where $r_2=d_1+1-r_0-r_1$ and $r_3=d_1+1-r_0-r_1-r_2$.
\item For $w=(d_1+1)(d_2+1)$.\\
\begin{align}
\tau^a &={n_1 \choose d_1+1}  {n_2 \choose d_2+1}, \nonumber \\
\tau^b &=(d_1+1)!  {d_2+2 \choose d_1+1}   {n_1 \choose d_1+1} {n_2 \choose d_2+2}\nonumber \\
&+\sum_{3r_0+2r_1+r_2=d_1+1}  {d_1+1 \choose r_0} {d_1+1-r_0 \choose r_1}{d_1+1-r_0-r_1 \choose r_2}\frac{(d_2+3)!}{2^{r_2}6^{r_3}(d_2-2d_1+1)!}\times\nonumber \\
& ~~~~~~~~~~~~~~~~~~~~~~~~{n_1 \choose d_1+1}{n_2 \choose d_2+3}\nonumber \\
&+\mathds{1}_{\{d_1=2,d_2=5\}} \times 11130{n_1 \choose d_1+1}{n_2 \choose d_2+4}\nonumber \\
&+\mathds{1}_{\{d_1=2,d_2=6\}} \times 111300{n_1 \choose d_1+1}{n_2 \choose d_2+4},\nonumber
\end{align}
where $r_3=d_1+1-r_0-r_1-r_2$.
\end{itemize}
\end{theorem}
\vspace{5mm}
The detailed proof of Theorem~\ref{th_stopping_sets_d1_d2} is found in Appendix~A.
From a stopping set perspective, 
both theorems \ref{th_stopping_sets_d}\&\ref{th_stopping_sets_d1_d2} match Tolhuizen's
results on weight distribution for a weight less than $d_1d_2+d_2$ \cite{Tolhuizen2002}. 
Our theorems found stopping sets that are only obvious for $w$ in the range $[d_1d_2, d_1d_2+d_2[$. 
For any weight $w$, there exists an equivalence in support between codewords
and obvious stopping sets (thanks to Proposition~\ref{prop_stop_codewords}). 
Trivial lower and upper bounds of the number of obvious weight-$w$ product code codewords are
\[
(q-1)\tau_w \le A_w \le \mathds{1}_{\{\tau_w \ne 0\}} A_w.
\]
For non-obvious stopping sets and non-obvious codewords, establishing a clear relationship
is still an open problem. This is directly related to solving the weight enumeration
beyond $d_1d_2+\max(d_1,d_2)$. In the special case $d_1=d_2=d$, Sendrier gave 
upper bounds of the number of erasure patterns for a weight up to $d^2+2d-1$ \cite{Sendrier1991}.

\section{Edge coloring algorithm under constraints \label{sec_edge_coloring}}
In section~\ref{sec_graph_representations}, we described graph representations of product codes
and we introduced the root order $\rho(e)$ of an edge with respect to its color $\phi(e)$. 
Our objective is to find a coloring $\phi$ such that the maximum diversity order is reached
under block erasures. The notion of root order in Definition~(\ref{def_root_order}) is for
double diversity ($L=2$) because it indirectly assumes that all symbols of one color out of $M$ 
can be erased by the channel.
Given the Singleton bound tradeoff stated in (\ref{equ_singleton_bound}),
double diversity is sufficient in distributed storage applications 
where the required coding rate should be sufficiently high.
Definition~(\ref{def_root_order}) may be generalized to take into account two or more erased colors,
e.g. see Figure~11 in \cite{Boutros2010} for $L=3$ with $M=3$ colors where an information symbol 
is protected by multiple root checknodes. In this paper,
we restrict both Definition~(\ref{def_root_order}) and the design in this section to a double-diversity product code.
This double diversity on a block-erasure channel is achieved if all stopping sets, 
as defined and counted in the previous section,
can be colored in a way such that at least two distinct colors are found within the symbols of a stopping set 
(valid for both $\G$ and $\G^c$).
This task is intractable. Imagine an edge coloring $\phi$ designed in a way to guarantee that all
weight-$w$ stopping sets include at least two colors. This task is already very hard (or almost impossible) 
for a fixed $w$. There is no coloring design tool for non-trivial product codes 
to ensure that all stopping sets of all weights incorporate at least two distinct colors.

\subsection{Hand-made edge coloring and its limitations\label{sec_edge_coloring_sub1}}
The aim of this section is to give more insight on designing edge coloring,
before introducing the differential evolution algorithm.\\ 

The compact graph $\G^c$ makes the design much simpler, as we saw in Section~\ref{sec_stopping_sets_sub3}.
The number of super-edges with the same color is $N^c/M$. 
We also know from (\ref{equ_Ri_and_Vi})-(\ref{equ_Nc_and_R}) that the size, height, and width
of $G^c$ are directly related to the component and total coding rates.

\begin{lemma}
\label{lem_rho=1}
Let $C_P=C_1 \otimes C_2$ be a product code with a column component $C_1[n_1,k_1]_q$
and a row component $C_2[n_2,k_2]_q$ whose coding rates are $R_1=k_1/n_1$ and $R_2=k_2/n_2$ respectively. 
Assume that $n_i-k_i$ divides $n_i$, for $i=1,2$, and assume that $M$ divides $N^c$.
$G^c$ admits an edge coloring $\phi$ such
that $\rho_{max}(\phi)=1$ if the coding rates satisfy
\begin{equation}
\min(R_1,R_2) \le 1-\frac{1}{M}.
\end{equation}
\end{lemma}
\begin{proof}
Consider the $|V_1^c| \times |V_2^c|$ matrix representation of $\G^c$. 
A sufficient condition to get $\rho_{max}(\phi)=1$ is to assign the $N^c/M$ edges having the same color
to a single row or a single column. The sufficient condition for $\rho_{max}(\phi)=1$ is expressed as
$N^c/M \le \max(n_1/(n_1-k_1), n_2/(n_2-k_2))$, the $\max$ let us select the longest item among a row or a column.
Recall also that $|V_i^c|=n_i/(n_i-k_i)$. 
Using (\ref{equ_Nc}), the sufficient condition becomes $n_1n_2 \le M\cdot\max(n_1(n_2-k_2), n_2(n_1-k_1))$.
Divide by $n_1n_2$ to get the inequality announced in the Lemma statement.
\end{proof}

When the palette has $M=4$ colors, the sufficient condition in Lemma~\ref{lem_rho=1} is written
as $\min(R_1,R_2) \le 3/4$. In order to achieve the block-fading Singleton bound for $M=4$,
we should take $R_1=3/4$ and $R_2=1$, i.e. the product code degenerates to a single component code.
It is possible to approach $R=3/4$ by keeping $R_1=3/4$ and letting $R_2=\frac{n_2-1}{n_2}$ be very close to $1$.
In this case, the row code $C_2$ is a single-parity check code over $F_q$. The product code is very unbalanced.
An example of such an unbalanced product code is
\[
C_P=[12,9,4]_q \otimes [14,13,2]_q.
\]
From the proof of Lemma~\ref{lem_rho=1}, the edge coloring of $\G^c$ satisfying $\rho_{max}=1$
is given by the following $4 \times 14$ matrix:
\begin{equation}
\left[
\begin{array}{cccccc}
R & R & R & \ldots  & R & R \\
G & G & G & \ldots  & G & G \\
B & B & B & \ldots  & B & B \\
Y & Y & Y & \ldots  & Y & Y 
\end{array}
\right],
\end{equation}
where the colors $\phi(e)=1,2,3,4$ are replaced by the four letters 'R', 'G', 'B', and 'Y'.
The rate of $[12,9,4]_q \otimes [14,13,2]_q$ is comparable to the rate of $[12,10,3]_q^{\otimes 2}$, $R \approx 0.69$
but it is sill far from reaching three quarters as the product code $[14,12,3]_q \otimes [16,14,3]_q$.
Of course, if the practical constraints allow for it, 
it is possible to consider an extremely unbalanced code such as $[12,9,4]_q \otimes [100,99,2]_q$!\\

Let us build balanced product codes by relaxing the constraint $\rho_{max}=1$. We may authorize
a $\rho_{max}$ greater than $1$ but not too large in order to limit the number of decoding iterations.
On the other hand, the double diversity condition on the edge coloring is maintained.
Firstly, let us find a hand-made edge coloring for the $[12,10,3]_q^{\otimes 2}$ product code with $M=4$ colors.
$\G^c$ has $6$ left supernodes, $6$ right supernodes, and a total of $36$ edges.
Each color is used $N^c/M=9$ times. The hint is to place a color on the rows of the matrix representation
of $\G^c$, row by row from the top to the bottom in a way that avoids stopping sets. 
The smallest stopping set is the $2 \times 2$ square.
Other non-obvious stopping sets may not be visible without a tedious row-column decoding
which is equivalent to determining the root order of all edges. 
We start with the first color 'R' and use the following number of letters per row:
\begin{equation}
\left[
\begin{array}{cccccc}
R & R & R & G & B & Y \\
R &   &   & R &   &  \\
R &   &   &   &   &  \\
R &   &   &   &   &  \\
R &   &   &   &   &  \\
R &   &   &   &   &  
\end{array}
\right].
\end{equation}
As seen above, we completed the first row with the three other colors.
On the second row, we moved the second 'R' to the right to avoid a $2 \times 2$ stopping set.
Next, we can start filling the second color 'G' from the third row, then the third color 'B' from the fifth row.
There will be no choice for the $9$ positions of 'Y'.
We allow some extra permutations to avoid small stopping sets. After filling the $36$ positions,
we found the following hand-made edge coloring for the $[12,10,3]_q^{\otimes 2}$ product code:
\begin{equation}
\left[
\begin{array}{cccccc}
R & R & R & G & B & Y \\
R & B & Y & R & Y & G \\
B & G & G & G & R & Y \\
Y & G & B & Y & G & R \\
R & G & B & B & B & Y \\
R & G & B & Y & Y & B 
\end{array}
\right].
\end{equation}
This coloring $\phi$ gives $24$ super-edges of order $1$ ($96$ edges in the non-compact graph $\G$) and $\rho_{max}(\phi)=3$.
Can we find a better $\phi$? Yes, in Section~\ref{sec_edge_coloring_sub3}, 
the DECA algorithm outputs an edge coloring with a population of $32$ super-edges of order $1$ 
($128$ edges in the non-compact graph $\G$) and reaching $\rho_{max}(\phi)=2$ only.\\

In a similar way, we attempt to build a double-diversity coloring for a well-balanced rate-$3/4$ product code,
e.g. the $[14,12,3]_q \otimes [16,14,3]_q$ product code where $R_1=6/7$, $R_2=7/8$, and $R=3/4$.
The compact graph $\G^c$ has $7$ left vertices and $8$ right vertices. For $M=4$ colors, each color
is used $N^c/M=56/4=14$ times. Again, we try to avoid small obvious stopping sets like $2 \times 2$, $2\times 3$, $3\times 3$, etc.
We start by putting five 'R' on the first row, three 'R' on the second row, two 'R' on the third row,
and one 'R' on the remaining rows as follows:
\begin{equation}
\left[
\begin{array}{cccccccc}
R & R & R & R & R & G & B & Y \\
R &   &   &   &   & R & R &   \\
R &   &   &   &   &   &   & R \\
R &   &   &   &   &   &   &   \\
R &   &   &   &   &   &   &   \\
R &   &   &   &   &   &   &   \\
R &   &   &   &   &   &   &   
\end{array}
\right].
\end{equation}
We repeat the same number of color entries 'G' starting on the fourth row.
The color 'B' starts with five entries on the seventh row. 
We allow some extra permutations to avoid small stopping sets.
Colors were exchanged within a row or within a column.
The coloring process was tedious. Many permutations had to be applied.
Some non-obvious stopping sets appeared, a computer software was used to reveal those sets (only for this task).
We reached the following hand-made double-diversity edge coloring 
for the $[14,12,3]_q \otimes [16,14,3]_q$ product code:
\begin{equation}
\left[
\begin{array}{cccccccc}
Y & R & R & Y & R & G & B & R \\
R & Y & B & G & Y & R & R & B \\
B & B & B & Y & Y & R & G & R \\
R & G & G & G & G & G & B & Y \\
R & G & Y & Y & B & Y & G & G \\
G & G & R & R & B & Y & Y & Y \\
R & G & B & B & B & B & B & Y
\end{array}
\right].
\end{equation}
This coloring gives $30$ super-edges of order $1$ in $\G^c$ ($120$ edges in the non-compact graph $\G$) and $\rho_{max}(\phi)=5$.
In Section~\ref{sec_edge_coloring_sub3}, for the same rate-$3/4$ product code,
the DECA algorithm outputs an edge coloring with a population of $40$ super-edges of order $1$ 
($160$ edges in the non-compact graph $\G$) and reaching $\rho_{max}(\phi)=3$ only.

\subsection{The algorithm\label{sec_edge_coloring_sub2}}
We propose in this section an algorithm for product codes that searches for an edge
coloring with a large number of root-order-$1$ edges ({\em good edges}) 
and achieving double diversity. The search is made in the ensemble of edge colorings $\Phi(E^c)$
of the compact graph $\G^c$. 
A necessary condition on the coding rate $R$ to get double diversity is
\begin{equation}
R \le 1-\frac{1}{M}, 
\end{equation}
i.e. those satisfying inequality (\ref{equ_singleton_bound}), 
where $M$ is the color palette size.
Codes attaining equality in~(\ref{equ_singleton_bound}) 
are referred to as MDS in the block-fading/block-erasure
sense \cite{Guillen2006a}\cite{Boutros2010}.
The main loop of our algorithm is a differential evolution loop
that mutates a fraction of the population of bad edges. The algorithm will be referred
to as the Differential Edge Coloring Algorithm (DECA).\\

\noindent
The population of bad edges is defined by the following set
\begin{equation}
B=\{ e \in E^c : \rho(e) > 1 \}.
\end{equation}
It should be remembered that $B=B(\phi)$ because of Definition~(\ref{def_root_order}), 
but $\phi$ is dropped here for the sake of simplifying the notations. 
The number of good edges is given by
\begin{equation}
\label{equ_eta_phi}
\eta(\phi)=|E^c \setminus B|=\left| \{ e \in E^c : \rho(e)= 1 \} \right|.
\end{equation}

Among the $|B|$ bad edges, colors of a fraction of $\aleph$ edges are modified in order
to maximize $\eta(\phi)$, $\aleph \in \N$. The fraction $\aleph/|B|$ should be large enough to allow 
for a population evolution but it should stay small enough in order to limit the algorithm complexity.
The DECA algorithm proceeds as follows.\\

\noindent
{\bf Initialization.} The compact graph $(V_1^c, V_2^c, E^c)$, the number of colors $M$, 
the differential evolution parameter $\aleph$, a maximum number of rounds $MaxIter$,
and an initial edge coloring $\phi_0$ are made ready as an input to DECA.\\

\noindent
{\bf Pre-processing.} Build all weak compositions of $\aleph$ with $M$ parts, 
i.e. write $\aleph$ as the some of $M$ non-negative integers,
\begin{equation}
\aleph=\gamma_1+\gamma_2+\ldots + \gamma_M,
\end{equation}
the number of weak compositions being
\begin{equation}
\Gamma={\aleph-M+1 \choose M-1}.
\end{equation}
For each weak composition, prepare the $\Lambda$ permutations
that permute colors among the $\aleph$ edges, the total number
of these permutations is
\begin{equation}
\Lambda(\gamma_1, \ldots, \gamma_M)=\frac{(\gamma_1+\ldots + \gamma_M)!}{\prod_{i=1}^{M} \gamma_i!}.
\end{equation}
This pre-processing step is completed by setting a loop counter to zero.\\

\noindent
{\bf Differential evolution loop.} This looping phase of DECA includes three
main steps.
\begin{itemize}
\item {\bf Edge sets initialization.} Set $\phi=\phi_0$ and $\eta_{max}=0$. 
Build $B=B(\phi)$ and randomly select a subset $B_{\aleph}$. There is a unique
weak composition $(\gamma_1, \ldots, \gamma_M)$ of $\aleph$ associated to $B_{\aleph}$ determined
by
\begin{equation}
\gamma_i=\left| \{ e \in B_{\aleph} : \phi(e)=i \} \right|.
\end{equation}

\item {\bf Color permutations.} For $\lambda=1\ldots \Lambda(\gamma_1, \ldots, \gamma_M)$, replace
the image of $B_{\aleph}$ in the mapping $\phi$ by a permutation of $\phi_0(B_{\aleph})$.
The color permutation is denoted by $\pi^{\lambda}$. This step is a modification
of the mapping $\phi_0$ at the $\aleph$ bad edges, i.e. 
$\phi(B_{\aleph}) \leftarrow \pi^{\lambda}(\phi_0(B_{\aleph}))$.
Record the mapping with the largest number of good edges, i.e. the edge coloring
with the best $\eta(\phi)$, 
in $\phi_1$ and update $\eta_{max}$. 
\item {\bf Termination.} Increment the counter of evolution loops.
Stop and output $\phi_1$ if this counter reaches $MaxIter$,
otherwise set $\phi_0=\phi_1$ and go back to the edge sets initialization.\\
\end{itemize}

\begin{figure}[!h]
\begin{center}
\includegraphics[width=0.57\columnwidth]{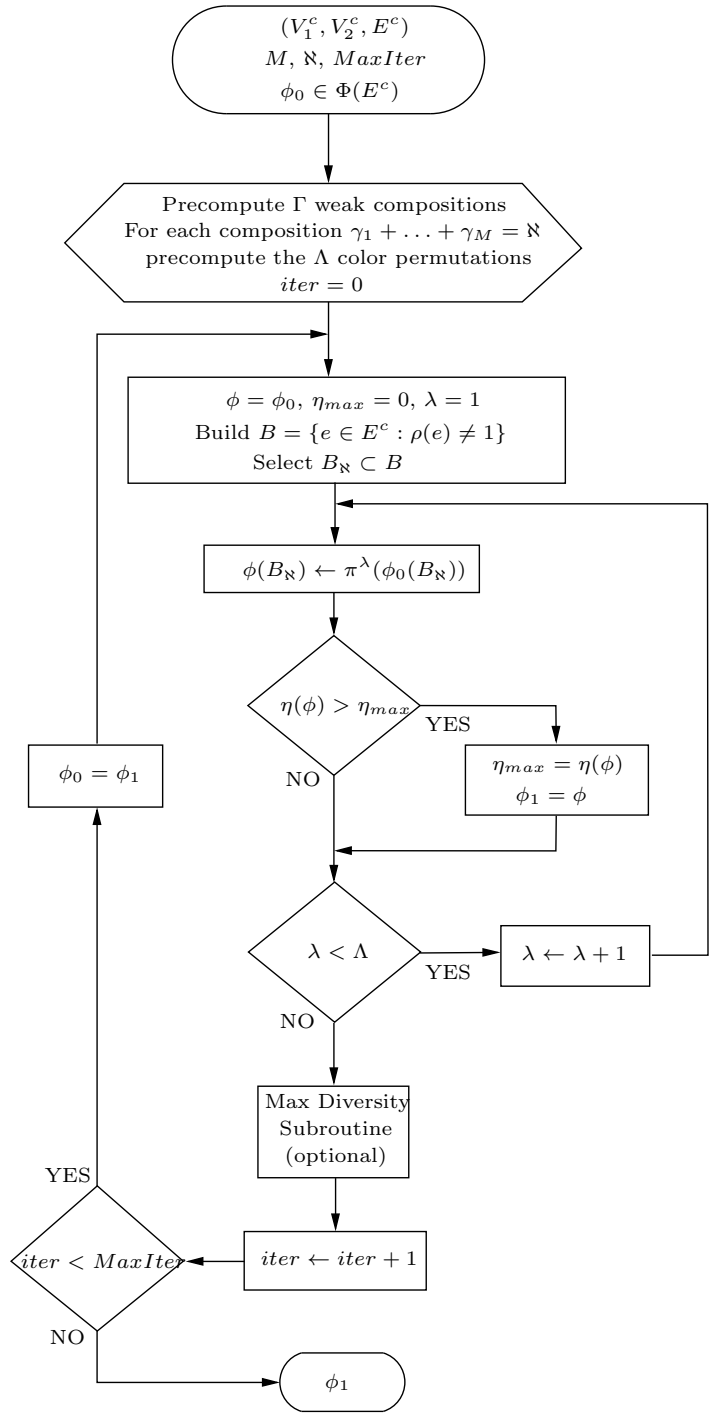}
\caption{Flowchart of the edge coloring algorithm (DECA) for designing double-diversity product codes.\label{fig_flowchart}}
\end{center}
\end{figure}

A detailed functional flowchart of DECA is drawn in Figure~\ref{fig_flowchart}. 
The complexity of DECA is mainly due to the differential evolution loop. 
The complexity is proportional to $\Lambda(\gamma_1, \ldots, \gamma_M)$ per round.
Hence, the number of operations in DECA behaves as
\begin{equation}
\label{equ_complexity}
\Lambda \le \Lambda_{max}(\aleph, M)=\frac{\aleph!}{((\aleph/M)!)^M}.
\end{equation}
When $\aleph$ is not multiple of $M$, the denominator in the
right term should be rewritten as 
$\prod_{i=1}^{i_0} \lfloor \aleph/M \rfloor \times \prod_{i=i_0+1}^{M} \lceil \aleph/M \rceil$,
where $i_0$ is chosen such that the sum of all elements involved in both products is equal
to $\aleph$.
All $\Gamma$ compositions of $\aleph$ are not considered by the algorithm.
In fact, the total number of permutations for all weak compositions is
\begin{equation}
\sum_{j=1}^{\Gamma} \frac{(\gamma_1(j)+\ldots + \gamma_M(j))!}{\prod_{i=1}^{M} \gamma_i(j)!} = M^{\aleph}.
\end{equation}
Fortunately, the per-round complexity of DECA given in (\ref{equ_complexity}) 
is much smaller that $M^{\aleph}$, i.e. $\Lambda_{max}=o(M^{\aleph})$. In practical product code design,
we will also have $\Lambda_{max} \ll M^{\aleph} \ll M^{N^c}$.\\

The proposed edge coloring algorithm aims at maximizing $\eta(\phi)$ but does not guarantee
that $\forall e \in E^c, \rho(e) < \infty$. In some cases, the algorithm may terminate all its
rounds with some edges having an infinite order,
i.e. the coloring is not double-diversity. This occurs when trying to design a product code
with a coding rate very close or equal to $1-1/M$, the block-fading/block-erasure Singleton bound rate.
To remedy for this weakness, DECA is endowed with an extra subroutine called {\em Max Diversity},
as shown in Figure~\ref{fig_flowchart}. Likewise the second step in the differential 
evolution loop, this subroutine applies color permutations to a subset $B_{\aleph_1}$ of edges,
$|B_{\aleph_1}|=\aleph_1$, $B_{\aleph_1} \subset B^{\infty}$, and
\begin{equation}
B^{\infty}=\{ e \in E^c : \rho(e)= \infty \}.
\end{equation}

\subsection{Applications \label{sec_edge_coloring_sub3}}
Now, let us apply DECA to design two double-diversity product codes with MDS components.
Numerical values are selected to make these codes suitable to distributed storage applications
and to diversity systems in wireless networks. The parameter $MaxIter$ is 100. DECA with its 
hundred iterations runs in a small fraction of a second on a standard computer machine.\\
 
\begin{example}\label{ex_DECA_12x12} 
The first application of DECA is to color edges in the compact graph of 
$C_{P1}=[n,k,d]_q^{\otimes 2}$, where $n=12$, $k=10$, $d=3$, and the finite-field alphabet size is $q > 12$.
The coding rate of $C_{P1}$ is $R(C_{P1})=25/36 < 1-1/M=3/4$, 
i.e. the gap to (\ref{equ_singleton_bound}) is $1/18$. 
This small gap is enough to render an uncomplicated double-diversity design.
The coloring in $\Phi(E^c)$ can be easily converted into its counterpart in $\Phi(E)$
by replacing each supersymbol with $4$ symbols.
From (\ref{equ_PhiE}) and (\ref{equ_PhiEc}), the total number of edge colorings is 
$|\Phi(E)| \approx 10^{83}$ in the non-compact graph 
and $|\Phi(E^c)|\approx 10^{19}$ in the compact graph. 
The differential evolution parameter $\aleph$ is set to $8$.
The diversity subroutine is deactivated.
We have
\begin{equation*}
\Lambda_{max}(8, 4)=2520 \ll  |\Phi(E^c)| \ll |\Phi(E)|.
\end{equation*}
For almost any choice of the initial coloring $\phi_0$ uniformly distributed in $\Phi(E^c)$,
DECA yields a double-diversity coloring $\phi_1$. For roughly one choice out of three for $\phi_0$, 
the algorithm outputs a coloring $\phi_1$ such that $\eta(\phi_1) \ge 28$. 
Figure~\ref{fig_matrix_color_6_6_a} shows the matrix representation of a special $\phi_1$ found by DECA.
It has $\eta(\phi_1)=32$ which corresponds to $\eta=128$ in $(V_1, V_2, E)$. 
The corresponding rootcheck order matrix is shown in Figure~\ref{fig_matrix_color_6_6_b}.
The highest attained order for this coloring is $\rho_{max}(\phi) =2$. The maximal
order for all colorings in $\Phi(E^c)$ from Theorem~\ref{th_root_order_max} is $\rho_u=5$.
This coloring satisfies equality in (\ref{equ_rho_etamin}) since $2\rho_{max}(\phi)+\eta_{min}(\phi)=12$.\\ 
\end{example}

\begin{figure}[!h]
\begin{subfigure}{0.48\columnwidth}
\begin{center}
\includegraphics[width=0.7\columnwidth]{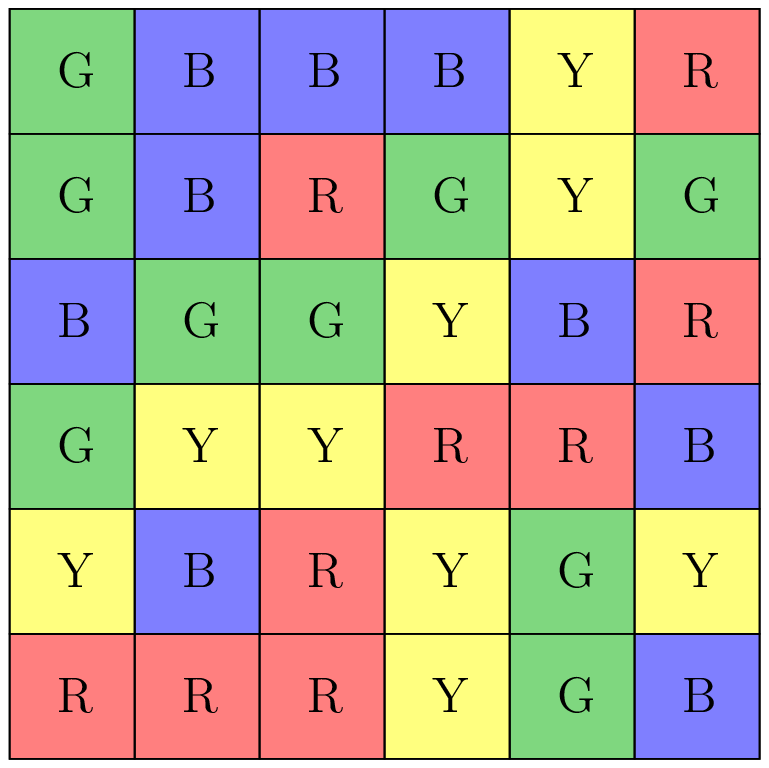}
\end{center}
\caption{\label{fig_matrix_color_6_6_a}}
\end{subfigure}
\hspace*{\fill}
\begin{subfigure}{0.48\columnwidth}
\begin{center}
\includegraphics[width=0.7\columnwidth]{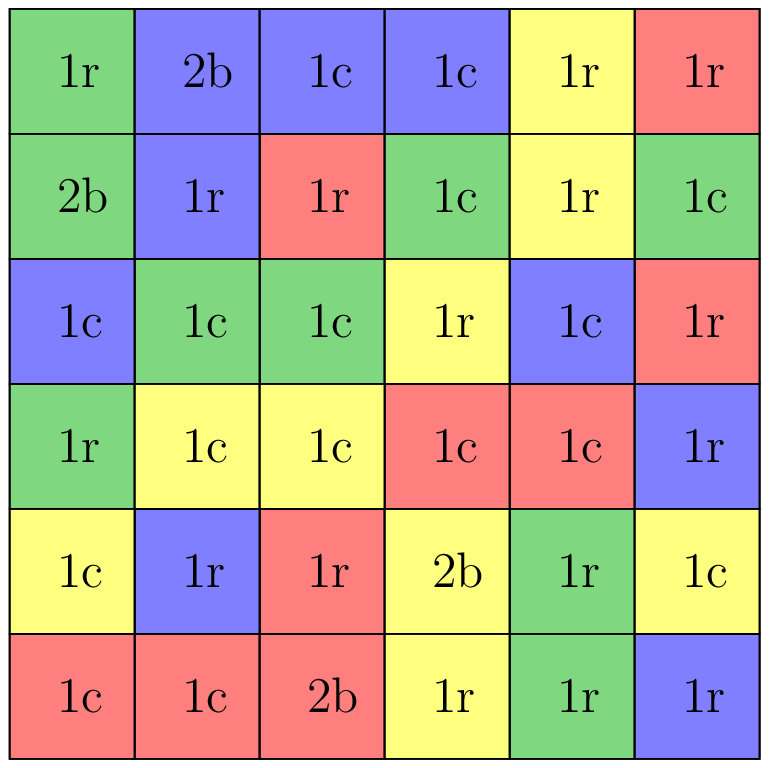}
\end{center}
\caption{\label{fig_matrix_color_6_6_b}}
\end{subfigure}
\caption{Compact coloring matrix (figure a) and the corresponding rootcheck-order matrix (figure b)
for the $[12,10]^{\otimes 2}$ product code $C_{P1}$ found by DECA,
$\eta(\phi)=32$ and $\rho_{max}=2$.\label{fig_matrix_color_6_6}}
\end{figure}

\begin{example}\label{ex_DECA_14x16}
The second more challenging application of DECA is the design of a double-diversity product
code attaining the block-fading/block-erasure Singleton bound. Let us consider
$C_{P2}=[n_1,k_1,d_1]_q \otimes [n_2,k_2,d_2]_q$, where $n_1=14$, $k_1=12$, $n_2=16$, $k_2=14$, $d_1=d_2=3$, 
and the finite-field alphabet size is $q > 16$. The coding rate is $R(C_{P2})=1-1/M=3/4$.
From (\ref{equ_PhiE}) and (\ref{equ_PhiEc}), the total number of edge colorings is 
$|\Phi(E)| \approx 10^{131}$ in the non-compact graph 
and $|\Phi(E^c)|\approx 10^{31}$ in the compact graph.
The differential evolution parameter $\aleph$ is set to $7$.
The diversity subroutine is activated with $\aleph_1=8$.
We have
\begin{equation*}
\Lambda_{max}(7,4)+\Lambda_{max}(8,4)=3150 \ll  |\Phi(E^c)| \ll |\Phi(E)|.
\end{equation*}
The initial coloring $\phi_0$ is taken to be uniformly distributed in $\Phi(E^c)$.
For almost three $\phi_0$ choices out of four, DECA yields a double-diversity coloring $\phi_1$.
Roughly one $\phi_0$ choice out of two guarantees $\eta(\phi_1) \ge 34$.
Figure~\ref{fig_matrix_color_7_8_a} shows the matrix representation of a special $\phi_1$ found by DECA.
It has $\eta(\phi_1)=40$ which corresponds to $\eta=160$ in $(V_1, V_2, E)$. 
The rootcheck order matrix is shown in Figure~\ref{fig_matrix_color_7_8_b}.
The highest attained order for this coloring is $\rho_{max}(\phi) =3$. The maximal
order for all colorings in $\Phi(E^c)$ from Theorem~\ref{th_root_order_max} is $\rho_u=7$. 
This coloring satisfies $2\rho_{max}(\phi)+\eta_{min}(\phi)=16$ 
while the right term in (\ref{equ_rho_etamin}) is $17$.\\ 
\end{example}

\begin{figure}[!h]
\begin{subfigure}{0.48\columnwidth}
\begin{center}
\includegraphics[width=0.9\columnwidth]{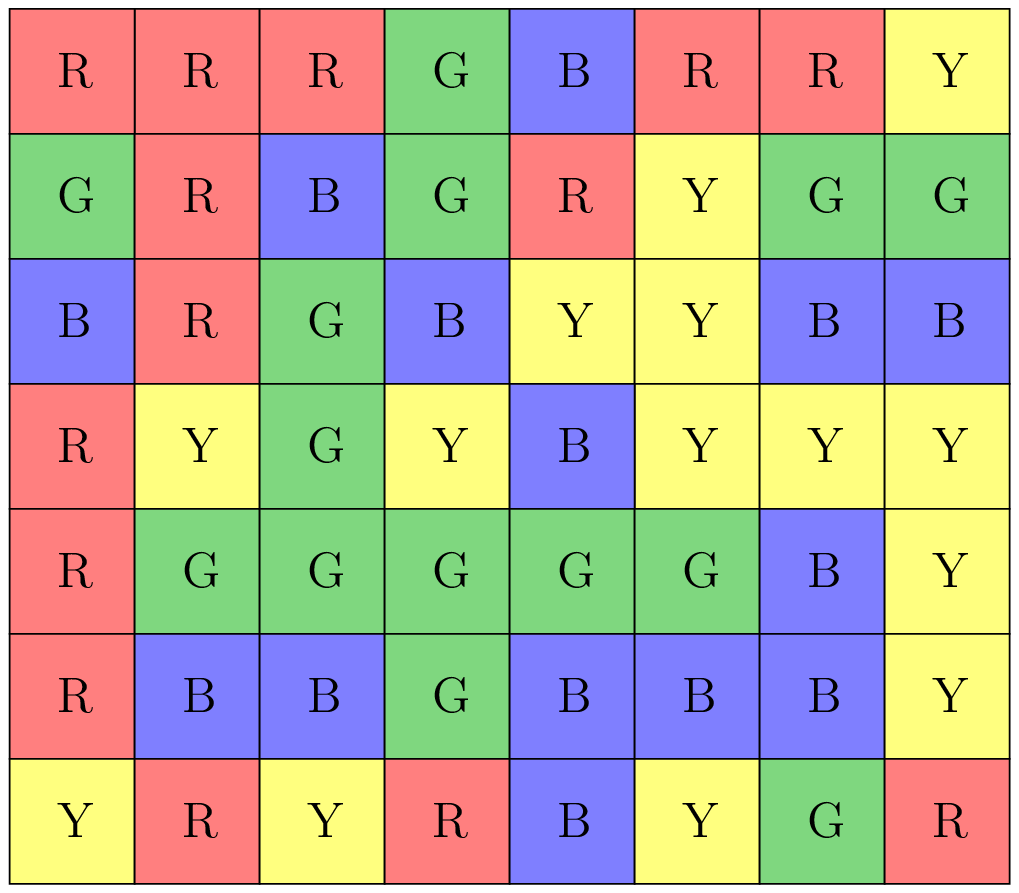}
\end{center}
\caption{\label{fig_matrix_color_7_8_a}}
\end{subfigure}
\hspace*{\fill}
\begin{subfigure}{0.48\columnwidth}
\begin{center}
\includegraphics[width=0.9\columnwidth]{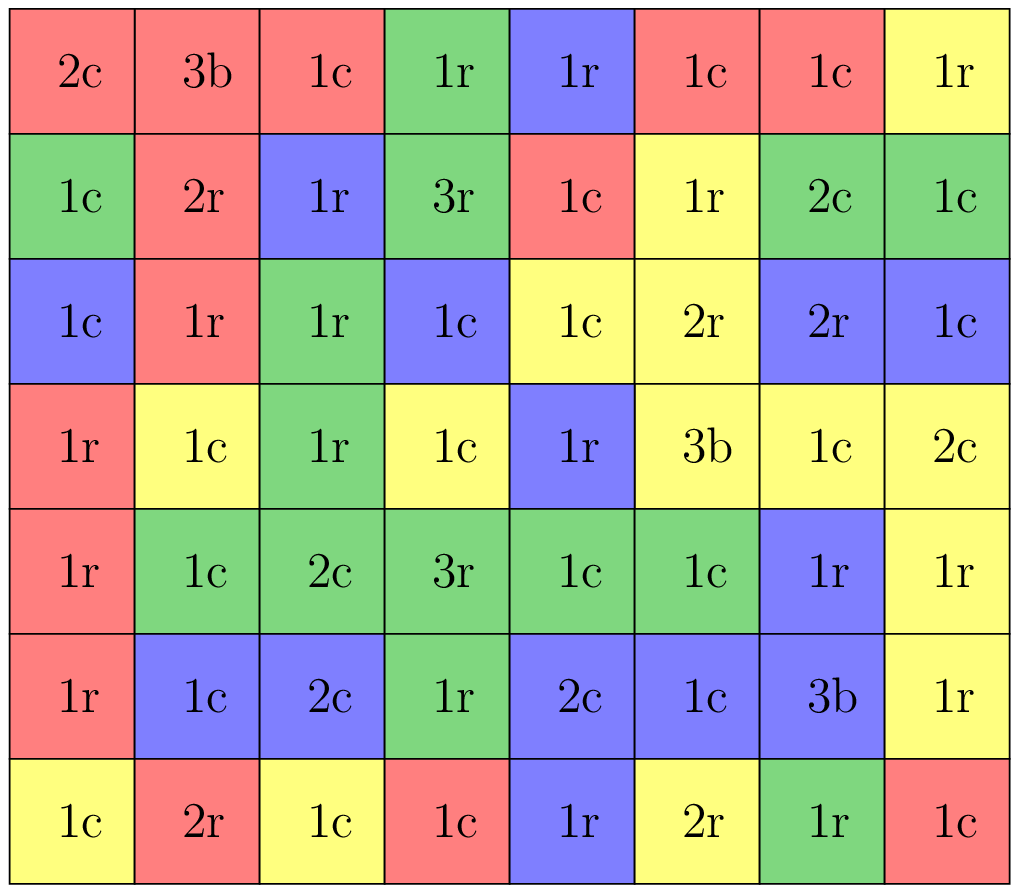}
\end{center}
\caption{\label{fig_matrix_color_7_8_b}}
\end{subfigure}
\caption{Compact coloring matrix (figure a) and the corresponding rootcheck-order matrix (figure b)
for the $[14,12] \otimes [16,14]$ product code $C_{P2}$ found by DECA,
$\eta(\phi)=40$ and $\rho_{max}=3$.\label{fig_matrix_color_7_8}}
\end{figure}

\begin{example}\label{ex_DECA_10x10}
A third example suitable for nowadays distributed storage warehouses is 
$C_{P3}=[10,8,3]_q \otimes [10,9,2]_q$. The coding rate is $R=18/25$ with a minimum
distance $d_1d_2=6$ and the locality is $n_1=n_2=10$, i.e. this code is an improvement 
to the standard $RS[14,10]$ used by Facebook~\cite{Rashmi2013}. 
The coloring ensembles have sizes $|\Phi(E)| \approx 10^{57}$ and $|\Phi(E^c)|\approx 10^{27}$
respectively. 
The DECA algorithm produced double-diversity edge colorings where we distinguish two classes:
a first class of colorings with $\rho_{max}=3$ and $\eta(\phi)~=~41$, and a second
class with $\rho_{max}=2$ and $\eta(\phi)~=~40$. An edge coloring of the second class
is shown in Figure~\ref{fig_matrix_color_5_10}. The reader is invited to determine
the rootcheck order matrix and verify that $40$ super-edges
have root order $1$ and $10$ super-edges have a root order equal to $2$.
\end{example}

\begin{figure}[!h]
\begin{center}
\includegraphics[width=0.5\columnwidth]{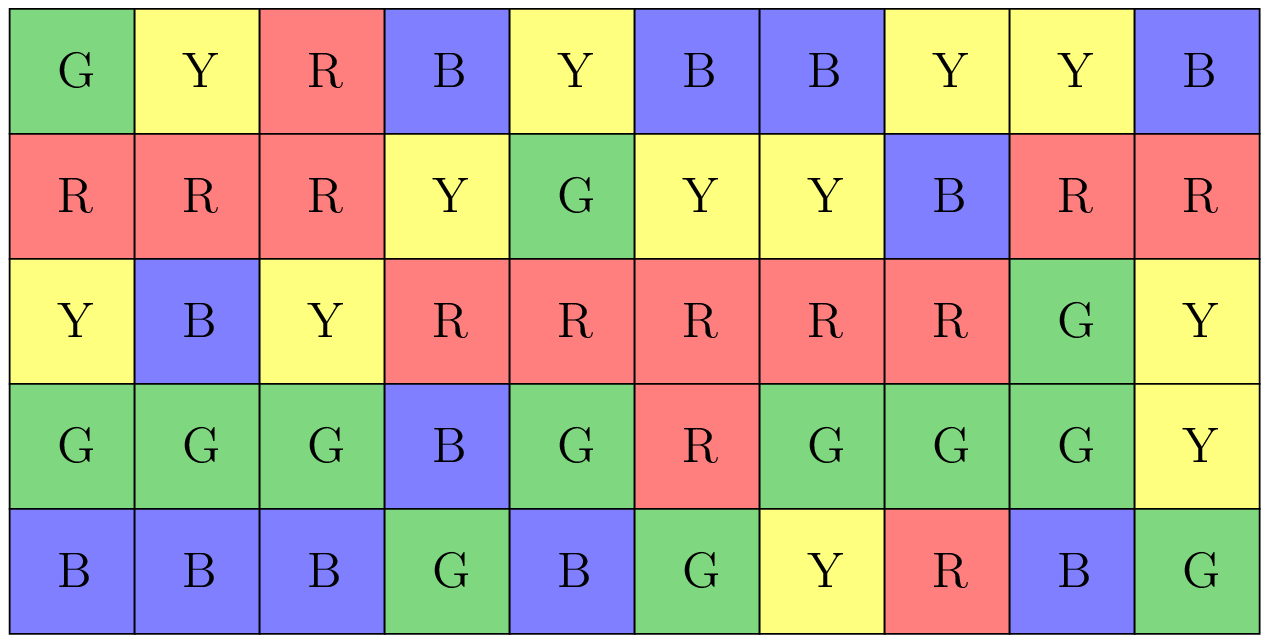}
\end{center}
\caption{Compact coloring matrix for the $[10,8] \otimes [10,9]$ product code found by DECA,
$\eta(\phi)=40$ and $\rho_{max}=2$.\label{fig_matrix_color_5_10}}
\end{figure}

In figures of the previous examples, the four colors were also indicated by the first letter of the color name, Red, Green, Blue, and Yellow. The rootcheck order $\rho(e)$ for an edge $e$ in $E^c$
(which is also the order of the four code symbols associated to that edge) is indicated by an integer
in the right part of each figure for the first two examples. 
In the rootcheck order matrix, $2r$ means that this supersymbol has order $2$ 
and its root checknode is a row. Similarly, $2c$ designates a supersymbol with order $2$ 
and a column rootcheck. The letter '$b$' is written when a supersymbol has both rootchecks, 
a row and a column rootcheck.\\

Product codes in Examples \ref{ex_DECA_12x12}-\ref{ex_DECA_10x10} do not satisfy 
the $\rho_u$ condition given in (\ref{equ_R_0.41}) and the sufficient condition of Lemma~\ref{lem_rho=1} either.
An interesting question arises. Does an edge coloring with $\rho_{max}=1$ exist for a $6 \times 6$ compact graph?
We provide a partial answer in the sequel. A similar answer is valid for the $7 \times 8$ compact graph.

The $6 \times 6$ compact graph is perfectly balanced. Let us start with the first color 'R'. The unique solution
to get $\rho(e)=1$ for all edges $e$ with $\phi(e)=R$ is to place 'R' entries separately on the first row and the first column.
Hence, no row or a column contain the same color twice. The first $9$ edges are located as follows:\\
\begin{equation}
\left[
\begin{array}{cccccc}
  & R & R & R & R & R \\
R &   &   &   &   &  \\
R &   &   &   &   &  \\
R &   &   &   &   &  \\
R &   &   &   &   &  \\
  &   &   &   &   &  
\end{array}
\right].
\end{equation}
We start over with the second color 'G' using the same rule. Given the lack of space on the second row
and the second column, the ninth green edge is placed on the top left corner. We get\\
\begin{equation}
\left[
\begin{array}{cccccc}
G & R & R & R & R & R \\
R &   & G & G & G & G \\
R & G &   &   &   &   \\
R & G &   &   &   &   \\
R & G &   &   &   &   \\
  & G &   &   &   &  
\end{array}
\right].
\end{equation}
At this point, $18$ super-edges have a rootcheck order $\rho=1$. 
Seven edges only can be colored in blue, three edges on the third row, three edges on the third column,
and one edge at the intersection of the second row and the second column. 
One color 'R' can be moved down to the last row leading to the following coloring:
\begin{equation}
\left[
\begin{array}{cccccc}
G & R & R & R & R & B \\
R & B & G & G & G & G \\
R & G &   & B & B & B \\
R & G & B &   &   &   \\
R & G & B &   &   &   \\
B & G & B &   &   & R
\end{array}
\right].
\end{equation}
Finally, we reached an edge coloring where all edges of three colors satisfy $\rho(e)=1$.
Unfortunately, there is no space left for edges of 'Y' to achieve $\rho(e)=1$.
The situation is even worse, the remaining edges for 'Y' make five primitive stopping sets
(three $2 \times 2$, one $2 \times 3$, and one $3 \times 2$). This edge coloring has no diversity.

\subsection{Random edge coloring \label{sec_edge_coloring_sub4}}
The efficiency of the DECA algorithm was validated in the previous section in terms
of number of edges of first order and the maximal order over all edges. 
Clearly, while evolving from one coloring to another in order to get a large $\eta(\phi)$,
DECA also produced a very small maximal order $\rho_{max}(\phi)$. 
Any deterministic construction seems to be destined to fail given the huge size
of the ensembles $\Phi(E)$ and $\Phi(E^c)$.\\

In this sub-section, another way to show the efficiency of our coloring algorithm 
is to make random selections from $\Phi(E)$ and $\Phi(E^c)$ and get an estimate 
of the probability distributions of $\eta(\phi)$ and $\rho_{max}(\phi)$.
Indeed, a uniformly distributed permutation in the symmetric group of order $N$
yields a uniformly distributed edge coloring $\phi$ in $\Phi(E)$. This is also true
for $\Phi(E^c)$ when the symmetric group has order $N^c$. 
Thus, in a uniform manner, we selected 2 billion edge colorings through our computer application
from $\Phi(E)$ and $\Phi(E^c)$ respectively.
For each coloring, rootcheck orders of all edges were computed, i.e. for the $N$ edges in the non-compact
graph and the $N^c$ edges in the compact graph. Only double-diversity colorings are counted
in this comparison, i.e. colorings with at least one edge of infinite rootcheck order are excluded.
As an illustration, the characteristics of double-diversity random coloring for $C_{P1}$ 
are plotted in Figure \ref{fig_random_12x12} where numerical estimations
of all probability distributions are compared to colorings designed via DECA.\\
\begin{figure*}[!h]
\begin{subfigure}{0.49\textwidth}
\begin{center}
\includegraphics[height=0.98\columnwidth,angle=270]{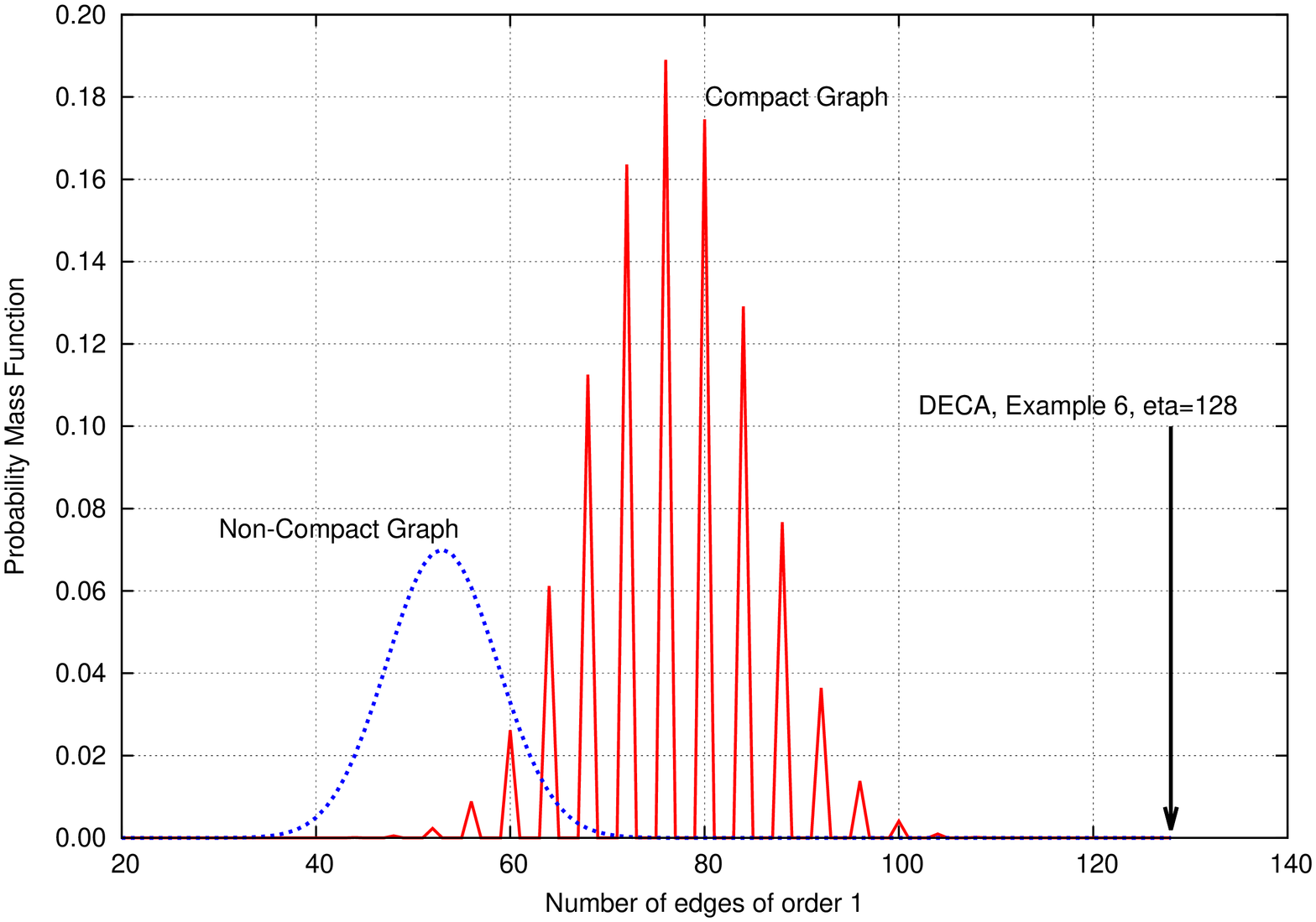}
\end{center}
\caption{\label{fig_random_eta_12x12}}
\end{subfigure}
\hspace*{\fill}
\begin{subfigure}{0.49\textwidth}
\begin{center}
\includegraphics[height=0.98\columnwidth,angle=270]{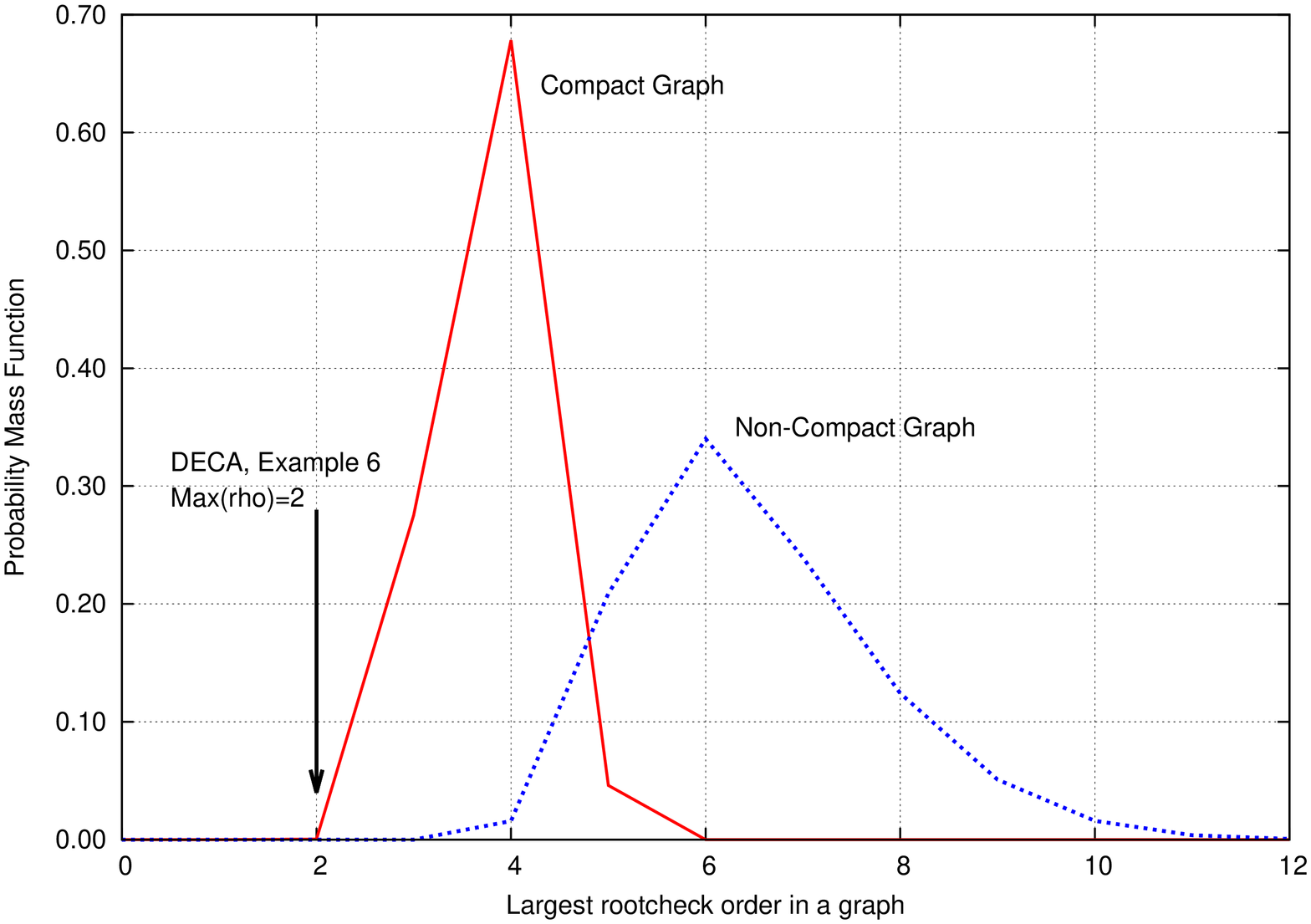}
\end{center}
\caption{\label{fig_random_rho_12x12}}
\end{subfigure}
\caption{Distribution of $\eta(\phi)$ (figure a) and $\rho_{max}(\phi)$ (figure b) 
for double-diversity random edge colorings uniformly distributed in $\Phi(E)$ and $\Phi(E^c)$. 
Product code $[12,10]^{\otimes 2}$.\label{fig_random_12x12}}
\end{figure*}
Double diversity design is more arduous for the rate-$3/4$ $C_{P2}$ product code 
than for the rate-$25/36$ $C_{P1}$ product code because of the rate-diversity tradeoff
given by the Singleton bound.
For $C_{P1}$, the $[12,10]^{\otimes 2}$ code, 
$8.97\%$ of uniformly sampled colorings have double diversity in $\Phi(E^c)$,
whereas this fraction is  $43.6\%$ in $\Phi(E)$. 
For $C_{P2}$, the $[14,12] \otimes [16,14]$ code, 
only $0.00039\%$ of uniformly sampled colorings have double diversity in $\Phi(E^c)$,
and we found no double-diversity colorings in $\Phi(E)$ despite the 2 billion samples. 
As expected, compact graphs exhibit better characteristics than non-compact
graphs thanks to their simpler structure, i.e. $n_i-k_i$ parity symbols are grouped 
inside a unique supersymbol: for $C_{P1}$, one double-diversity random coloring has $\eta(\phi)=88$,
$\rho_{max}(\phi)=4$ for non-compact graphs, 
seven double-diversity colorings have $\eta(\phi)=120$, 
and $\rho_{max}(\phi)=2$ for compact graphs. 
There exists a double-diversity coloring in $\Phi(E)$ with $\rho_{max}(\phi)=3$ but
its $\eta$ is $85$.
The estimated probability mass functions for $C_{P1}$ are plotted in Figures~\ref{fig_random_eta_12x12}
and \ref{fig_random_rho_12x12}.
For $C_{P2}$, one double-diversity random coloring reached $\eta(\phi)=128$ and $\rho_{max}(\phi)=4$
out of the 2 billion samples from $\Phi(E^c)$. 
In all cases, for both $\eta$ and $\rho$,
double-diversity random colorings are not as efficient as colorings designed via the DECA algorithm. 
The situation is worse for random colorings if a double-diversity code with maximal rate $1-1/M$ is to be designed.
The DECA algorithm exhibits excellent values, $\eta=160$ and $\rho=3$, for the rate-$3/4$ 
$[14,12] \otimes [16,14]$ product code.

\section{Code performance in presence of erasures \label{perf_erasure}}
Iterative decoding performance of $C_P=C_1 \otimes C_2$ is studied in presence of
channel erasures, with and without edge coloring. The iterative decoder makes
row and column iterations where the component decoder of $C_i$ can be an algebraic erasure-filling decoder
(limited by $d_i-1$) or a maximum-likelihood decoder of $C_i$. 
As stated in Section~\ref{sec_stopping_sets_sub1},
type II and type III stopping sets are identical because the non-binary codes $C_1$ and $C_2$ are MDS.
The word error probability of the iterative decoder is denoted by $P_{ew}^{\G}$.
The product code can also be decoded via an ML decoder, i.e. maximum likelihood decoding of $C_P$
based on a Gaussian reduction of its parity-check matrix. The word error probability
under ML decoding of $C_P$ is denoted by $P_{ew}^{ML}$.

\subsection{Block erasures \label{perf_erasure_sub1}}
Consider the block-erasure channel $CEC(q,\epsilon)$. The $N$ symbols of a codeword are partitioned
into $M$ blocks, each block contains symbols associated to edges in $\G$ with the same color.
The $CEC(q,\epsilon)$ channel erases a block with a probability~$\epsilon$. The block
is correctly received with a probability~$1-\epsilon$. Erasure events
are independent from one block to another. We say that a {\em color is erased} 
if the associated block of $N/M$ symbols is erased. Assume that $\G$ is endowed with a double-diversity
edge coloring $\phi$ (i.e. $L(\phi)=2$) as defined in Corollary~\ref{cor_rho_max}. 
Then, on the block-erasure channel $CEC(q,\epsilon)$, for a rate satisfying
\begin{equation}
1-\frac{2}{M} < R \le 1-\frac{1}{M},
\end{equation}
we have
\begin{equation}
\label{equ_perf_CEC}
\epsilon^2 ~\le~ P_{ew}^{ML} ~\le~ P_{ew}^{\G} ~\le~ \sum_{i=2}^M {M \choose i} \epsilon^i (1-\epsilon)^{M-i}.
\end{equation}
Since $\phi$ has a double diversity, there exist two colors among the $M$ colors
such that the iterative decoder must fail if both colors are erased. This explains the upper bound
of $P_{ew}^{\G}$ in~(\ref{equ_perf_CEC}). The upper bound is valid for any rate less than 
the maximal achievable rate for double diversity, i.e. $1-\frac{1}{M}$. 
Now, since $R > 1-\frac{2}{M}$, the ML decoder for $C_P$ cannot attain a diversity $L=3$ otherwise
the block-fading/block-erasure Singleton bound would be violated. Consequently, the ML
decoder of $C_P$ can only reach $L=2$ and so there exists a pair of erased colors that cannot be
solved by the ML decoder. This explains the lower bound in (\ref{equ_perf_CEC}).
The reader can easily verify that
\begin{equation}
\lim_{\epsilon \rightarrow 0} \frac{\log P_{ew}^{ML}}{\log \epsilon}
= \lim_{\epsilon \rightarrow 0} \frac{\log P_{ew}^{\G}}{\log \epsilon} = L =2.
\end{equation}
The slope of $P_{ew}$ versus the erasure probability $\epsilon$ in a double-logarithmic scale 
is equal to $2$. Under the stated constraint on $R$, the upper bound in (\ref{equ_perf_CEC}) is the exact expression
of the outage probability on a block-erasure channel valid for
$q$-ary codes with asymptotic length \cite{Guillen2006b}.
For double-diversity edge colorings found by DECA 
in Examples \ref{ex_DECA_12x12} and \ref{ex_DECA_14x16},
$P_{ew}^{\G}$ equals its upper bound in (\ref{equ_perf_CEC}). These examples achieve the outage probability
although a code may perform better than the outage probability at finite length.
For these colorings where $M=4$, the error probability on $CEC(q,\epsilon)$ behaves
like $P_{ew}^{\G}=6\epsilon^2+O(\epsilon^3)$. One possible interpretation of this behavior is:
the optimization of $\eta(\phi)$ (equivalent in some sense to minimizing $\rho(\phi)$) 
pushed the performance of edge colorings found by the DECA algorithm as far as possible 
from the lower bound $\epsilon^2$.
As can be observed in Figures \ref{fig_matrix_color_6_6} and \ref{fig_matrix_color_7_8}, 
all rows and all columns include the four colors. When any two colors out of four
are erased, the iterative decoder will completely fail without correcting a single supersymbol.
A double-diversity edge coloring guarantees that all stopping sets are covered by at least two colors
but it cannot cover all stopping sets with three colors or more otherwise we get $L=3$
which contradicts $R> 1-\frac{2}{M}$.
Fortunately, these product codes are diversity-wise MDS and the second code in Example \ref{ex_DECA_14x16}
has the maximal coding rate for double diversity. In the sequel, we will see that these codes
also perform well in presence of independent erasures.

\subsection{Independent erasures \label{perf_erasure_sub2}}
Consider the i.i.d. erasure channel $SEC(q,\epsilon)$. 
The $N$ symbols of a codeword are independently erased by the channel.
A symbol is erased with a probability $\epsilon$ and is correctly received with 
a probability $1-\epsilon$. Edge coloring has no effect on the performance of $C_P$
on the $SEC(q,\epsilon)$ channel. Before studying the performance on the $SEC(q,\epsilon)$,
following Examples \ref{ex_w=9} \& \ref{ex_w=12} 
and Theorems \ref{th_stopping_sets_d} \& \ref{th_stopping_sets_d1_d2},
we state an obvious result about obvious stopping sets in the following proposition.
\begin{proposition}
\label{prop_stop_codewords}
Let $C_P=C_1 \otimes C_2$ be a product code with non-binary MDS components.
All obvious stopping sets are supports of product code codewords.
\end{proposition}
\begin{proof}
Consider an $\ell_1 \times \ell_2$ obvious stopping set. Its rectangular support
is $\mathcal{R}(\mathcal{S})=\mathcal{R}_1(\mathcal{S}) \times  \mathcal{R}_2(\mathcal{S})$.
We have $\ell_1 \ge d_1$ and $\ell_2 \ge d_2$. 
From Proposition~\ref{prop_MDS_w}, there exists a column codeword 
$x=(x_1, x_2, \ldots, x_{n_1}) \in C_1$ of weight $\ell_1$ 
with support $\mathcal{R}_1(\mathcal{S}) \times \{j_1\}$, where $j_1 \in \mathcal{R}_2(\mathcal{S})$.
Similarly, there exists a row codeword $y=(y_1, y_2, \ldots, y_{n_2}) \in C_2$ of weight $\ell_2$ 
with support $\{i_1\} \times \mathcal{R}_2(\mathcal{S})$, where $i_1 \in \mathcal{R}_1(\mathcal{S})$.
Now, the Kronecker product of $x$ and $y$ satisfies $\mathcal{X}(x\otimes y)=\mathcal{S}$.
\end{proof}

\begin{corollary}
\label{cor_PewG_PewML}
Consider a product code $C_P=C_1 \otimes C_2$ with non-binary MDS component codes.
Assume the symbols of $C_P$ are transmitted over a $SEC(q,\epsilon)$ channel.
Then, for $\epsilon~\ll~1$, the error probabilities satisfy $P_{ew}^{\G} \sim P_{ew}^{ML}$.
\end{corollary}
\begin{proof}
On the $SEC(q,\epsilon)$, the word error probabilities are given by \cite{Schwartz2006},
\begin{equation}
\label{equ_PewML_Psi}
P_{ew}^{ML}=\sum_{i=d_1d_2}^{N} \Psi_i(ML) \epsilon^i (1-\epsilon)^{N-i},
\end{equation}
where $\Psi_i(ML)$ is the number of weight-$i$ erasure patterns covering a product code codeword,
and
\begin{equation}
\label{equ_PewG_Psi}
P_{ew}^{\G}=\sum_{i=d_1d_2}^{N} \Psi_i(\G) \epsilon^i (1-\epsilon)^{N-i},
\end{equation}
where $\Psi_i(\G)$ is the number of weight-$i$ erasure patterns covering a stopping set.
Of course, here we refer to stopping sets in the non-compact graph $\G$, i.e. in the $n_1 \times n_2$
product code matrix. Next, since $N$ is fixed (asymptotic length analysis is not considered in this paper)
we write $P_{ew}^{ML}=\Psi_{d_1d_2}(ML) \epsilon^{d_1d_2} + o(\epsilon^{d_1d_2})$ 
and $P_{ew}^{\G}=\Psi_{d_1d_2}(\G) \epsilon^{d_1d_2} + o(\epsilon^{d_1d_2})$. 
From Proposition~\ref{prop_stop_codewords}, we get the equality $\Psi_{d_1d_2}(\G)=\Psi_{d_1d_2}(ML)$
and so we obtain $\lim_{\epsilon \rightarrow 0} P_{ew}^{\G}/P_{ew}^{ML}=1$.
\end{proof}

The erasure patterns can be decomposed according to the size of the covered stopping set.
The coefficient $\Psi_i(\G)$ becomes $\Psi_i(\G)=\sum_{w=d_1d_2}^i \Psi_{i,w}(\G)$, where
$\Psi_{i,w}(\G)$ is the number of weight-$i$ patterns covering a stopping set of size $w$.
It is clear that $\Psi_{w,w}(\G)=\tau_w$. For small $i-w$, 
$\Psi_{i,w}(\G)$ can be approximated by $\sum_{\mathcal{A}} { N-\mathcal{A} \choose i-w} \tau_{w,\mathcal{A}}$,
where $\tau_{w,\mathcal{A}}$ is the number of stopping sets of size $w$ having 
$|\mathcal{R}(\mathcal{S})|=\mathcal{A}$. For $w \le d_1d_2+d_1+d_2+1$, the area $\mathcal{A}$
is bounded from above by the product $\ell_1^0 \times \ell_2^0$ from Lemma~\ref{lem_max_support}. 
Numerical evaluations of $\Psi_i(\G)$ are tractable for very short codes ($N \le 25$)
and become very difficult for codes of moderate size and beyond, e.g. $N=144$ and $N=224$ for the
$[12,10]^{\otimes 2}$ and the $[14,12] \otimes [16,14]$ codes respectively.
For this reason, expressions (\ref{equ_PewML_Psi}) and (\ref{equ_PewG_Psi}) are not practical 
to predict the $SEC(q,\epsilon)$ performance of product codes with significant characteristics.\\

For $P_{ew}^{\G}$, thanks to Theorems \ref{th_stopping_sets_d} and \ref{th_stopping_sets_d1_d2}, 
a union bound can be easily established. Indeed, we have
\begin{align*}
P_{ew}^{\G} & =Prob(\exists \cS~covered) \\
          &\le \sum_w Prob(\exists \cS: |\cS|=w, \cS~covered),
\end{align*}
leading to
\begin{equation}
\label{equ_Union_bound}
P_{ew}^{\G} \le P^U(\epsilon)=\sum_{w=d_1d_2}^{N} \tau_w \epsilon^w.
\end{equation}

\begin{figure}[!h]
\begin{center}
\includegraphics[height=0.98\columnwidth,angle=270]{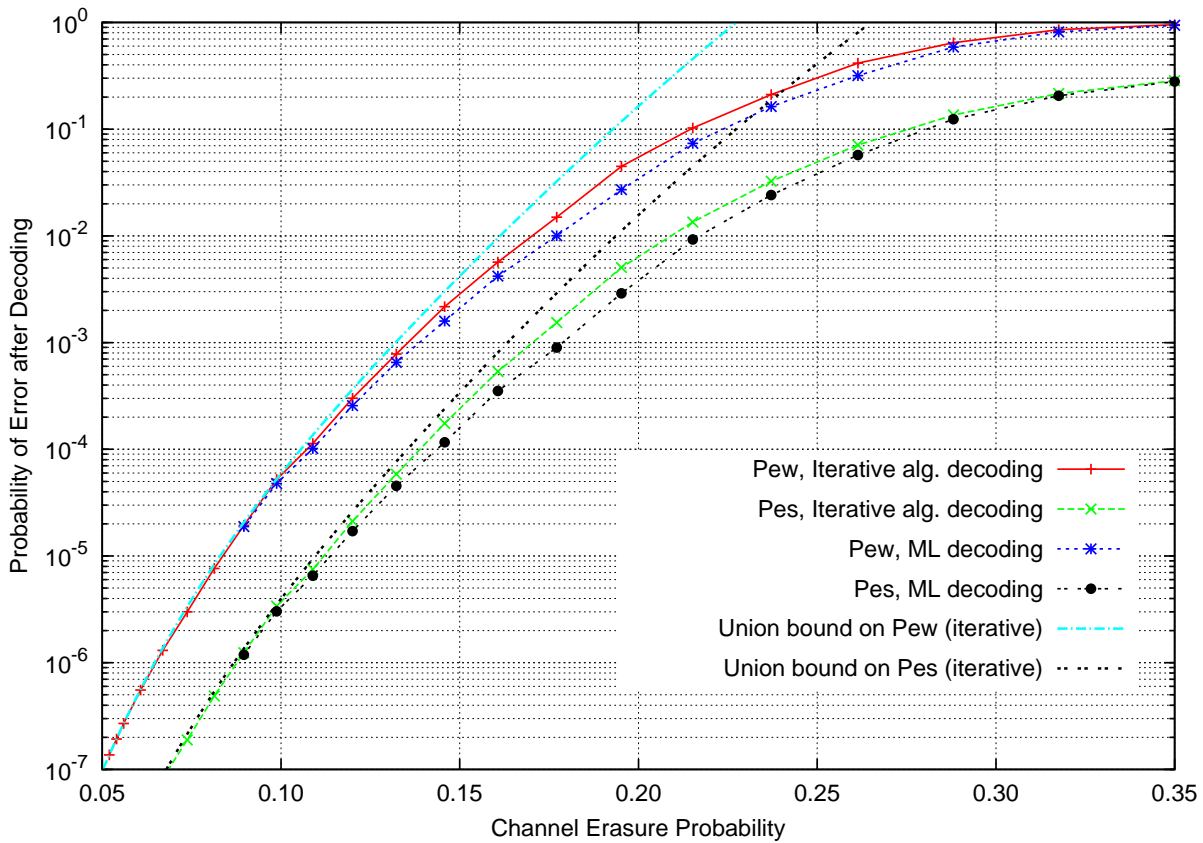}
\caption{Product code $[12,10]_q^{\otimes 2}$, no edge coloring. 
Word and symbol error rate performance for iterative decoding versus its union bound 
and ML decoding.\label{fig_ErasurePerf12}}
\end{center}
\end{figure}

\noindent
From Theorem~\ref{th_stopping_sets_d}, 
the union bound $P^U(\epsilon)$ for the $[12,10,3]_q^{\otimes 2}$ product code is
\begin{align*}
P^U(\epsilon) =&48400\epsilon^9+6098400\epsilon^{12}+23522400\epsilon^{13} +17641800\epsilon^{14}\\
               &+1754335440\epsilon^{15}+9126691200\epsilon^{16}+o(\epsilon^{16}).
\end{align*}
The performance of this code on the $SEC(q,\epsilon)$ channel is shown in Figure~\ref{fig_ErasurePerf12}. 
We used the standard finite field of size $q=256$. The union bound for the symbol error 
probability $P_{es}^{\G}$ is derived by weighting the summation term in (\ref{equ_Union_bound}) with $w/N$,
i.e. $P_{es}^{\G} \le \sum_{w=d_1d_2}^{N} \frac{w}{N}\tau_w \epsilon^w$.
As observed in the plot of Figure~\ref{fig_ErasurePerf12}, 
the union bound is sufficiently tight. Furthermore, the performance of the iterative algebraic
row-column decoder is very close to that of ML decoding in the whole range of $\epsilon$. 
For small $\epsilon$, the curves are superimposed as predicted by Corollary~\ref{cor_PewG_PewML}.\\ 

\begin{figure}[!t]
\begin{center}
\includegraphics[height=0.98\columnwidth,angle=270]{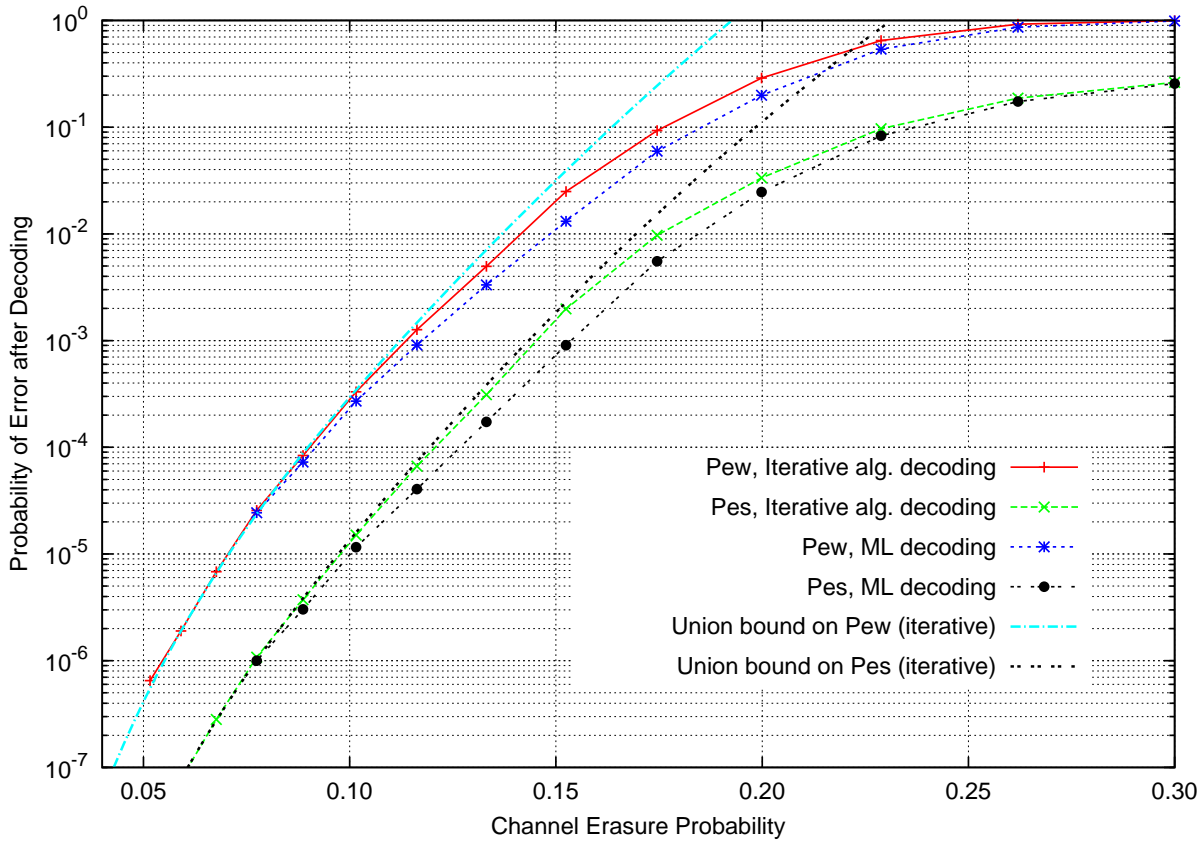}
\caption{Product code $[14,12]_q \otimes [16,14]_q$, no edge coloring. 
Word and symbol error rate performance for iterative decoding versus its union bound 
and ML decoding.\label{fig_ErasurePerf16}}
\end{center}
\end{figure}

\noindent
The union bound $P^U(\epsilon)$ for the $[14,12,3]_q \otimes [16,14,3]_q$ product code is 
\begin{align*}
P^U(\epsilon) =&203840\epsilon^9+44946720\epsilon^{12}+174894720\epsilon^{13}+
131171040\epsilon^{14}\\
               &+17839261440\epsilon^{15}+126887941180\epsilon^{16}+o(\epsilon^{16}).
\end{align*}
The performance of this code on the $SEC(q,\epsilon)$ channel is shown in Figure~\ref{fig_ErasurePerf16}. 
Similar to the previous code, the union bound is tight enough and iterative decoding performs
very close to ML decoding. Finally, let us interpret these results from a finite-length information
theoretical point of view \cite{Polyanskiy2010}. The $SEC(q,\epsilon)$ of Shannon capacity 
$\log_2(q)(1-\epsilon)$ behaves exactly like a $BEC(\epsilon)$ of capacity $(1-\epsilon)$
but erasures in the $SEC$ occur at the symbol level instead of the binary digit level. 
Finite-regime BEC bounds from \cite{Polyanskiy2010} are directly applicable 
to our product codes over the $SEC(q,\epsilon)$. The BEC channel dispersion 
is $V=\epsilon(1-\epsilon)$ and its maximal achievable rate is given by \cite{Polyanskiy2010}, Theorem~53,
\begin{equation}
R=(1-\epsilon)-\sqrt{\frac{V}{n}} Q^{-1}(P_{ew})+O(\frac{1}{n}),
\end{equation}
where $n$ is the code length, $Q(x)$ is the Gaussian tail function, $\epsilon$ is the channel 
erasure probability, and $P_{ew}$ is the target word error probability.
The next table shows how good is the proposed product code based on MDS components.

\begin{table}[!h]
\begin{center}
\begin{tabular}{|l|c|c|}
\hline
 & Coding Rate $R$ & Erasure Prob. $\epsilon$ \\ 
 & for $\epsilon=0.15$ & for $R=0.75$ \\ \hline \hline
Polyanskiy-Poor-Verd\'u      & $0.794$ : $P_{ew}=1.0\cdot10^{-2}$ & $0.189$ \\ \hline
$[14,12]_q \otimes [16,14]_q$ & $0.750$ : $P_{ew}=1.0\cdot10^{-2}$ & $0.150$ \\ \hline
Regular-$(3,12)$ LDPC         & $0.750$ : $P_{ew}=2.9\cdot 10^{-2}$ & $0.135$ \\ \hline
\end{tabular}
\end{center}
\caption{Finite-length performance of the $[14,12]_q\otimes [16,14]_q$ product code.
The value of $\epsilon$ in the third column is given for 
$P_{ew}=10^{-2}$ at all rows.\label{tab_polyanskiy}}
\end{table}

\subsection{Unequal probability erasures \label{perf_erasure_sub3}}
In communication and storage systems, erasure events of unequal probabilities may occur.
In order to observe the effect of a double-diversity coloring on the performance
in multiple erasure channels, we define the $SEC(q,\{\epsilon_i\}_{i=1}^M)$.
On this channel, symbol erasure events are independent but the probability of erasing a symbol
is $\epsilon_i$ if it is associated to an edge in $\G$ with color $\phi(e)=i$.
The union bound is easily modified to get
\begin{equation}
P_{ew}^{\G} \le P^U(\epsilon_1, \ldots, \epsilon_M),
\end{equation}
where
\begin{equation}
\label{equ_Union_bound_multiple}
P^U(\epsilon_1, \ldots, \epsilon_M) = \hspace{-1mm} \sum_{w=d_1d_2}^{N} 
\sum_{
{\scriptsize
\begin{array}{l}
w_1,\ldots,w_M\\ : \sum_i w_i =w
\end{array}
}
} \hspace{-3mm} \tau(w_1,\ldots,w_M) \prod_{i=1}^M \epsilon_i^{w_i}.
\end{equation}
The coefficient $\tau(w_1,\ldots,w_M)$ is the number of stopping sets of size $w=\sum_{i=1}^M w_i$,
where $i$ symbol edges have color $i$, $i=1 \ldots M$. Clearly, the coefficients
$\tau(w_1,\ldots,w_M)$ depend on the edge coloring $\phi$. For double-diversity colorings and $M\ge2$,
these coefficients satisfy the following property:\\
For any stopping set $\cS$ such that $|\cS|=w$, $\tau(w_1,\ldots,w_M)$ does exist
for $\sum_{i=1}^M w_i=w$ and $w_i>0$ only, i.e. no weak compositions of $w$ are authorized by $\phi$.\\
Hence, the product code should perform well if one of the $\epsilon_i$ is close to $1$
and the remaining $\epsilon_i$ are small enough. The extreme case is true thanks to double
diversity yielding $P^U(0^{M-1}, 1^1)=0$, where $(0^{M-1}, 1^1)$ represents all vectors
with all positions at $0$ except for one position set to $1$.
Figure~\ref{fig_ErasurePerf4colors1210}
shows the performance of $[12,10]_q^{\otimes 2}$ on the $SEC(q,\{\epsilon_i\}_{i=1}^M)$ channel with $M=4$ colors.
The edge coloring is the double-diversity coloring produced by the DECA algorithm 
and drawn in Figure~\ref{fig_matrix_color_6_6}. The expression of $P^U(\epsilon_1, \ldots, \epsilon_M)$
is determined by stopping sets enumeration as in Theorems~\ref{th_stopping_sets_d} and \ref{th_stopping_sets_d1_d2}. 
Details are omitted and the very long expression of $P^U(\epsilon_1, \ldots, \epsilon_M)$ is not shown.
The special case $\epsilon_1=\epsilon_2=\epsilon_3$
is considered and the performance is plotted as a function of $\epsilon_4$.
For a fixed $\epsilon_1$, double diversity dramatically improves the performance with respect to $\epsilon_4$.

\begin{figure}[!t]
\begin{center}
\includegraphics[height=0.98\columnwidth,angle=270]{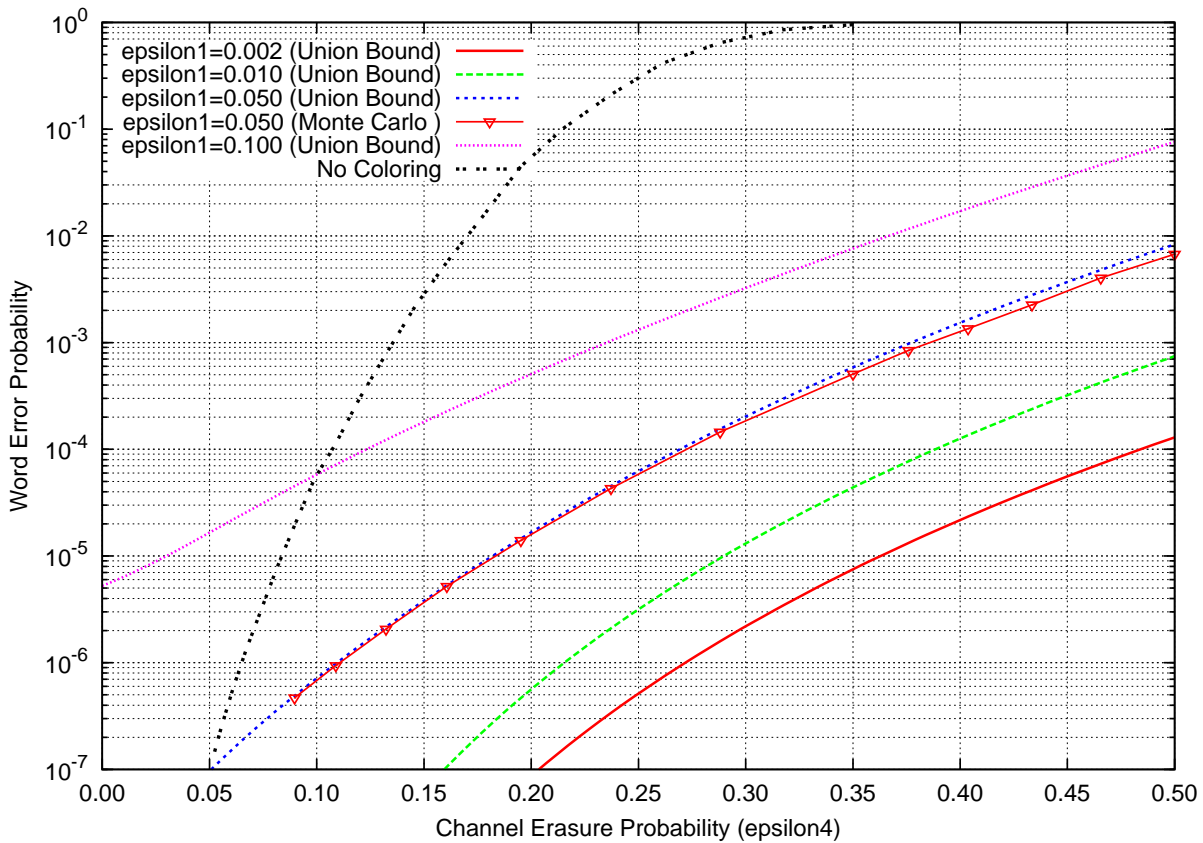}
\caption{Product code $[12,10]_q^{\otimes 2}$ with double-diversity edge coloring. 
Word error rate performance versus $\epsilon_4$, 
for iterative decoding on the $SEC(q,\{\epsilon_i\}_{i=1}^4)$ channel with $\epsilon_1=\epsilon_2=\epsilon_3$.\label{fig_ErasurePerf4colors1210}}
\end{center}
\end{figure}

\section{Conclusions \label{sec_conclusions}}
Non-binary product codes with MDS components are studied in this paper
in the context of iterative row-column algebraic decoding. 
Channels with both independent and block erasures are considered. 
The rootcheck concept and associated double-diversity edge colorings 
were described after introducing a compact graph representation for product codes.
For solving erased symbols, an upper bound of the number of decoding iterations 
is given as a function of the graph size and the color palette size $M$. 
Stopping sets are defined in the context of MDS components and 
a relationship is established with the graph representation of the product code.
A full characterization of these stopping sets is given up to a weight
$(d_1+1)(d_2+1)$. Then, we proposed a differential evolution edge coloring algorithm
to design colorings with a large population of minimal rootcheck order symbols.
The complexity of this algorithm per iteration is $o(M^{\aleph})$, where $\aleph$
is the differential evolution parameter. The performance of MDS-based product codes 
with and without double-diversity coloring is analyzed. In addition, 
ML and iterative decoding are proven to coincide at small channel erasure probability.
Original results found in this paper are listed in Section~\ref{sec_main_results}.

A complete enumeration of product code codewords is still an open problem in coding theory.
Following the enumeration of bipartite graphs in Section~\ref{sec_stopping_sets_sub4} 
(see also Table~\ref{tab_seq_spec_parts}) and following
the DECA algorithm that aims at improving $\eta(\phi)$ in Section~\ref{sec_edge_coloring_sub2}, 
two open problems can be stated.\\
$\bullet$ In number theory. There exists no recursive or closed form expression
for the special partition function, i.e. the number of special partitions of an integer.
Also, in a way similar to the Hardy-Ramanujan formula, the asymptotic behavior
is unknown for the number of special partitions. Special partitions are introduced
in Definition~\ref{def_special_partition}.\\
$\bullet$ In graph theory and combinatorics. Consider a matrix of size $H \times W$
and a coloring palette of size $M$. For simplicity, assume that $H\cdot W$ is multiple of $M$.
A matrix entry is called {\em edge}. A color is assigned to each edge in the matrix.
All $M$ colors are equally used. 
A matrix edge/entry $(i,j)$ of color $c$ is said to be {\em good} if it is the unique entry
with color $c$ either on row $i$ or on column $j$. The number of good entries
is denoted by $\eta(\phi)$, see also (\ref{equ_eta_phi}). 
Given the matrix height $H$, width $W$,
and the palette size $M$, find the maximum achievable number of good entries $\eta(\phi)$ 
over the set of all edge colorings $\phi$. A simpler problem would be to find 
an upper bound of $\eta(\phi)$.
\newpage
\section*{Appendix A\\Proof of Theorem~\ref{th_stopping_sets_d1_d2}}
Of course, $C_1$ and $C_2$ are interchangeable which explains why we stated the theorem for $d_1 < d_2$.
From Lemma~\ref{lem_max_support}, the maximal rectangle height satisfies $\ell_1^0 \le (d_1+2)$.
Similarly, under the condition $d_2 < 3d_1-1$, the maximal rectangle width satisfies $\ell_2^0 \le (d_2+3)$.
From $d_1 \times d_2$ up to the maximal size $(d_1+2) \times (d_2+3)$, 
there are twelve different sizes listed in Figure~\ref{fig_rectangle_d1_d2}.
The most right column tells us when sizes located on the same row are equal.
Also, the first entries on rows $4$ and $5$ are equal if $d_2=d_1+1$.
For these rectangular supports, the stopping set weight $w$
takes values from rows $1$-$4$ (and the ranges between these values) in the table 
drawn in Figure~\ref{fig_rectangle_d1_d2}, i.e. $d_1d_2 \le w \le (d_1+1)(d_2+1)$. 
\vspace*{1cm}
\begin{figure}[!h]
\begin{center}
\includegraphics[width=0.7\columnwidth]{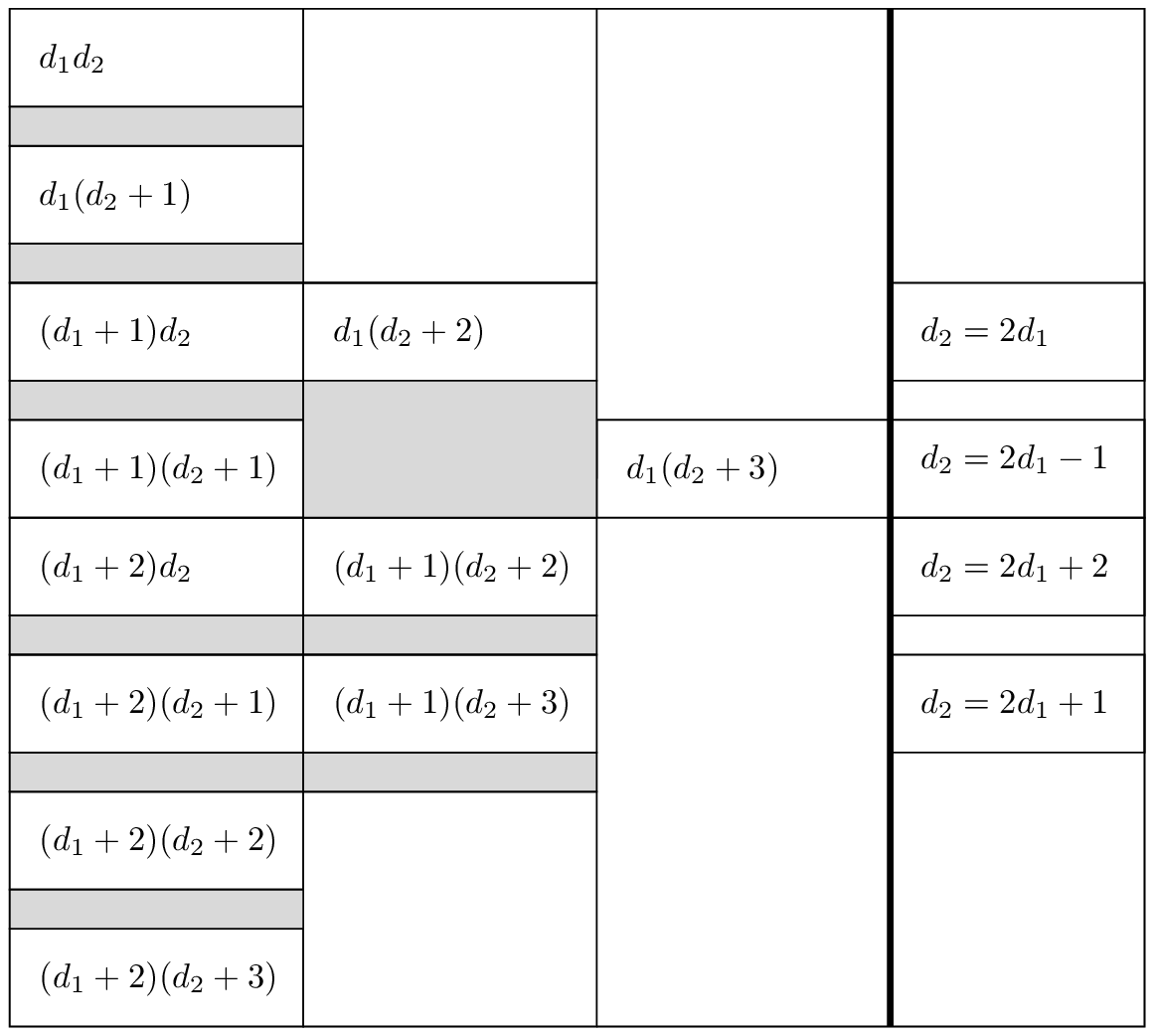}
\caption{Size of the rectangular support $\cR(\cS)$ given in the three left columns. 
The twelve different sizes are listed in increasing order within each column. The right
column of this table indicates when sizes on the same row are equal.\label{fig_rectangle_d1_d2}}
\end{center}
\end{figure}

The proof shall consider $d_2>2d_1$ in its sub-section C.
There exists no integer $d_2$ in the range $]2d_1, 3d_1-1[$ for $d_1=2$.
Only for sub-section C and $d_1=2$, we consider a rectangular support
with a width up to $d_2+4$ which enlarges the range of $d_2$
to $2d_1 < d_2 < 4d_1-1$ and permits to keep the case $d_1=2$ valid in sub-section C.\\ 

\begin{itemize}
\itemsep=3mm
\item The case $w < d_1d_2$.\\
The proof is similar to $w<d^2$ in Theorem~\ref{th_stopping_sets_d}. Here,
we just deduce that  $ w \ge d_1\ell_2$ and $ w \ge d_1\ell_2$ leading to the contradiction
$w \ge d_1d_2$. Therefore $\tau_w=0$ for $w < d_1d_2$ under type II iterative decoding.
\item The case $w=d_1d_2$.\\
We use similar inequalities as in the previous case which resembles the proof in Theorem~\ref{th_stopping_sets_d}
for $w=d^2$. We get that $\cR(\cS)=\cS$. All stopping set of size $d_1d_2$ are obvious. 
Their number is given by choosing $d_1$ rows out of $n_1$ and $d_2$ columns out of $n_2$.
\item The case $ d_1d_2 < w < d_1(d_2+1)$.\\
Given that $\ell_1\ell_2 \ge w > d_1d_2$, we get $\ell_1 \ge d_1$ and $\ell_2 \ge d_2$,
since $d_1 \times d_2$ is the smallest $\cR(\cS)$.
Take $\ell_1=d_1$, then $\ell_2 \ge d_2+1$ because $w > d_1d_2$. The weight of each column must be at least $d_1$
giving us $w \ge d_1\ell_2 \ge d_1(d_2+1)$, which is a contradiction unless $\tau_w=0$.
The same arguments hold for $\ell_1>d_1$.
\item The case $w=d_1(d_2+1)$.\\
The admissible rectangular support can have all sizes $\ell_1 \times \ell_2$ 
listed in Figure~\ref{fig_rectangle_d1_d2} starting from $d_1 \times (d_2+1)$.
\begin{itemize}
\item The smallest $\cR(\cS)$ is $d_1 \times (d_2+1)$. All corresponding stopping sets are obvious.
Their number is
\[
{n_1 \choose d_1}  {n_2 \choose d_2+1}.
\]
\item The next $\cR(\cS)$ has size $(d_1+1)\times d_2$. 
The number of zeros is $\beta=(d_1+1)d_2-d_1(d_2+1)=d_2-d_1>0$. This result
contradicts Lemma~\ref{lem_max_zeros} where $\beta=0$.
Hence, this size of the rectangular support yields no stopping sets, $\tau_w=0$ in this sub-case.
\item The next $\cR(\cS)$ has size $d_1\times (d_2+2)$. Again,  Lemma~\ref{lem_max_zeros} on
the existence of a stopping set tells us that $\beta=0$, but $\beta=d_1(d_2+2)-d_1(d_2+1)=d_1>0$.
Then $\tau_w=0$.
\item All rectangle sizes from rows $4-8$ in the table in Figure~\ref{fig_rectangle_d1_d2}
are larger than the previous case and make a contradiction on $\beta$ unless $\tau_w=0$.
\end{itemize}
\end{itemize}
Starting from this point, $d_2$ should be compared to $2d_1$ in order 
to sort the values of the stopping set size as given in the table in Figure~\ref{fig_rectangle_d1_d2}.

\subsection*{A. Minimum distances satisfying $d_2 < 2d_1$}
\begin{itemize}
\item The case $d_1(d_2 + 1)<w<(d_1+1)d_2$.\\
The smallest $\cR(\cS)$ has size $(d_1+1)\times d_2$ and the largest has size
$(d_1+2)\times (d_2+3)$. All these stopping sets contradict Lemma~\ref{lem_max_zeros}
if $\beta$ is computed from the size of $\cR(\cS)$ and $w$. Then $\tau_w=0$.
\item The case $w=(d_1+1)d_2$.
\begin{itemize}
\item $\cR(\cS)$ has size $(d_1+1)\times d_2$. Stopping sets are obvious and their number is
\[
{n_1 \choose d_1+1}  {n_2 \choose d_2}.
\]
\item $\cR(\cS)$ has size $d_1\times (d_2+2)$. Lemma~\ref{lem_max_zeros} gives $\beta=0$
but $\beta=d_1(d_2+2)-(d_1+1)d_2=2d_1-d_2 \ge 1$. Then $\tau_w=0$ in this sub-case.
\item $\cR(\cS)$ has size $(d_1+1)\times (d_2+1)$. We have $\beta=d_1+1$
and $d_2-d_1$ columns have no $0$. The $\beta$ zeros should be
in a $(d_1+1)\times(d_1+1)$ permutation matrix in the remaining $\beta$ columns. 
These stopping sets are not obvious and their number is 
\[
(d_1+1)!  {d_2+1 \choose d_2-d_1}  {n_1 \choose d_1+1}  {n_2 \choose d_2+1}.
\]
\item All other $\cR(\cS)$ greater than the previous case have $\tau_w=0$
because of a contradiction on $\beta$.
\end{itemize}
\item The case $(d_1+1)d_2 < w < d_1(d_2+2)$.\\
Let us write $w=(d_1+1)d_2+\lambda$, where $\lambda$ belongs to $[1, 2d_1-d_2-1]$.
If $d_2=2d_1-1$ this range for $\lambda$ is empty and we obtain $\tau_w=0$.
Then, we consider $d_2 < 2d_1-1$.
\begin{itemize}
\item $\cR(\cS)$ has size $d_1\times (d_2+2)$. The number of zeros
is $\beta=d_1(d_2+2)-w>0$ which contradicts Lemma~\ref{lem_max_zeros}.
There are no stopping sets for this rectangular size.
\item $\cR(\cS)$ has size $(d_1+1)\times (d_2+1)$. 
We have $\beta=d_1+1-\lambda \in [d_2-d_1+2, d_1]$. The non-obvious
stopping sets are built by selecting $\beta$ columns and $\beta$ rows
and then embedding any $0$-permutation matrix, their number is
\[
(d_1+1-\lambda)!  {d_1+1 \choose \lambda}  {d_2+1 \choose d_1+1-\lambda}  {n_1 \choose d_1+1}  {n_2 \choose d_2+1}. \nonumber
\]
\item All other $\cR(\cS)$ greater than the previous case have $\tau_w=0$
because of a contradiction on $\beta$.
\end{itemize}
\item The case $w=d_1(d_2+2)$.
\begin{itemize}
\item For $\cR(\cS)$ with size $d_1(d_2+2)$, we the following number
of obvious stopping sets
\[
{n_1 \choose d_1}  {n_2 \choose d_2+2}.
\]
\item The next size for $\cR(\cS)$ to be considered is $(d_1+1) \times (d_2+1)$. 
The number of zeros is $\beta=d_2-d_1+1$. As usual, these non-obvious stopping
sets are constructed by a $0$-permutation matrix of size $\beta$ inside $\cR(\cS)$.
Their number is
\[
(d_2-d_1+1)!   {d_1+1 \choose d_2-d_1+1}  {d_2+1 \choose d_2-d_1+1} 
              {n_1 \choose d_1+1}  {n_2 \choose d_2+1}. \nonumber
\]
\item Both $d_1\times (d_2+3)$ and $(d_1+2)\times d_2$ lead to a contradiction on $\beta$.
Now we consider the rectangle of size $(d_1+1) \times (d_2+2)$. The number
of zeros is $\beta=d_2+2$. All columns must have a single $0$.
Regarding the rows, let $r_0$, $r_1$, and $r_2$ be the number of rows
with $0$, $1$, and $2$ zeros respectively. We have $r_0+r_1+r_2=d_1+1$
and $\beta=2r_2+r_1$. Combining the two previous equalities yields
$2r_0+r_1=2d_1-d_2$. Many similar cases where encountered in the proof of Theorem~\ref{th_stopping_sets_d}.
The number of these non-obvious stopping sets becomes
\[
\sum_{2r_0+r_1=2d_1-d_2}  {d_1+1 \choose r_0} {d_1+1-r_0 \choose r_1} \frac{(d_2+2)!}{2^{r_0+d_2-d_1+1}} 
 {n_1 \choose d_1+1}  {n_2 \choose d_2+2}. \nonumber 
\]
\end{itemize}
\item The case $d_1(d_2+2)<w<(d_1+1)(d_2+1)$.\\
Let us write $w=d_1(d_2+2)+\lambda$, where $\lambda$ belongs to the interval $[1, d_2-d_1]$.
\begin{itemize}
\item The smallest rectangle has size $(d_1+1)(d_2+1)$. 
The number of zeros is $\beta=d_2-d_1+1-\lambda \in [1, d_2-d_1]$. The number of these
non-obvious stopping sets is found by counting all $\beta \times \beta$ permutation matrices
in all positions,
\[
(d_2-d_1+1-\lambda)!   {d_1+1 \choose 2d_1-d_2+\lambda}  {d_2+1 \choose d_1+\lambda} 
  {n_1 \choose d_1+1}  {n_2 \choose d_2+1}. \nonumber
\]
\item The next $\cR(\cS)$ has size $(d_1+1)(d_2+2)$ according to the table in Figure~\ref{fig_rectangle_d1_d2},
since both sizes $d_1(d_2+3)$ and $(d_1+2)d_2$ lead to a contradiction on $\beta$. 
The number of zeros for this rectangular support is $\beta=d_2+2-\lambda \in [d_1+2,d_2+1]$.
The $(d_2+2)$ columns satisfy: $\lambda$ columns have no zero and $\beta$
columns have a unique zero. As usual, we solve $\beta=2r_2+r_1$ and $r_0+r_1+r_2=d_1+1$
to get $2r_0+r_1=2d_1-d_2+\lambda$ and $r_2=r_0+\lambda-d_1-1$. 
The number of non-obvious stopping sets in this case is
\[
\sum_{2r_0+r_1=2d_1-d_2+\lambda}  {d_1+1 \choose r_0} {d_1+1-r_0 \choose r_1} \frac{(d_2+2)!}{2^{r_2}  \lambda!} 
  {n_1 \choose d_1+1}  {n_2 \choose d_2+2}. \nonumber
\]
\item Rectangular supports larger than $(d_1+1)(d_2+2)$ do not correspond to stopping sets
for the given range of $w$, i.e. $\tau_w=0$.
\end{itemize}
\item The last case $w=(d_1+1)(d_2+1) \le d_1(d_2+3)$.\\
\begin{itemize}
\item The number of obvious stopping sets for the smallest $\cR(\cS)$ is
\[
{n_1 \choose d_1+1}  {n_2 \choose d_2+1}.
\]
If $d_2=2d_1-1$ then $(d_1+1)(d_2+1)=d_1(d_2+3)$ and corresponds to
the following obvious stopping sets
\[
{n_1 \choose d_1}  {n_2 \choose d_2+3}.
\]
Similarly, if $d_2=d_1+1$ then $(d_1+1)(d_2+1)=(d_1+2)d_2$ and corresponds to
the obvious stopping sets with number
\[
{n_1 \choose d_1+2}  {n_2 \choose d_2}.
\]
Notice that $d_1\times (d_2+3)$ and $(d_1+2)\times d_2$ have no non-obvious 
stopping sets (from Lemma~\ref{lem_max_zeros}). 
\item The next $\cR(\cS)$ is $(d_1+1)(d_2+2)$. The corresponding number of zeros
is $\beta=d_1+1$. Similar cases were encountered before.
The number of these non-obvious stopping sets is
\[
\sum_{2r_0+r_1=d_1+1}  {d_1+1 \choose r_0} {d_1+1-r_0 \choose r_1} \frac{(d_2+2)!}{2^{r_0}(d_2-d_1+1)!} 
  {n_1 \choose d_1+1}  {n_2 \choose d_2+2}. \nonumber
\]
\item Consider $\cR(\cS)$ with size $(d_1+2)(d_2+1)$. We have $\beta=d_2+1$.
If $d_2 > d_1+1$, then we find $\tau_w=0$ by contradicting arguments on $\beta$.
But if $d_2=d_1+1$, the number of non-obvious stopping sets becomes
\[
\sum_{2r_0+r_1=d_2+1}  {d_2+1 \choose r_0} {d_2+1-r_0 \choose r_1} \frac{(d_1+2)!}{2^{r_0}} 
  {n_1 \choose d_1+2}  {n_2 \choose d_2+1}. \nonumber
\]
\item Consider $\cR(\cS)$ with size $(d_1+1)(d_2+3)$. We have $\beta=2(d_1+1)$.
If $d_2 < 2d_1-1$ there are no stopping sets. When $d_2=2d_1-1$, we
get $\beta=2(d_1+1)=d_2+3$. The number of non-obvious stopping sets is found to be
(method as in previous cases)
\[
\sum_{3r_0+2r_1+r_2=d_1+1}  {d_1+1 \choose r_0} {d_1+1-r_0 \choose r_1} 
 {d_1+1-r_0-r_1 \choose r_2} \frac{(d_2+3)!}{2^{r_2}6^{r_3}}
   {n_1 \choose d_1+1}  {n_2 \choose d_2+3}, 
\]
where $r_3=d_1+1-r_0-r_1-r_2$.
\item Consider the next $\cR(\cS)$ with size $(d_1+2)(d_2+2)$ 
as given in the table in Figure~\ref{fig_rectangle_d1_d2}. 
We have $\beta=d_1+d_2+3$. No stopping sets are found (by contradiction on $\beta$)
except for $d_2=d_1+1$. In this case, we get $\beta=2(d_1+2)$.
The rectangle has two zeros in each row.
This problem is solved in a similar method as in the proofs of
Lemma~\ref{lem_graph_bipartite_deg2} and Lemma~\ref{lem_graph_bipartite_deg2_1}. 
Indeed, we have to enumerate bipartite graphs with $d_1+2$ left vertices
all of degree $2$. These graphs have $d_1+3$ right vertices. 
Two cases should be distinguished: a- The extra vertex on the right
has no edges, b- The extra vertex at the right has one edge.
The number of these stopping sets is
\[
\left((d_2+2)x_{d_1+2} + \frac{(d_2+2)y_{d_1+2}}{2} \right) 
 {n_1 \choose d_1+2}  {n_2 \choose d_2+2}, \nonumber
\]
where $x_{d_1+2}$ and $y_{d_1+2}$ are determined 
from Lemma~\ref{lem_graph_bipartite_deg2} and Lemma~\ref{lem_graph_bipartite_deg2_1}.
\item The largest rectangular support for $w=(d_1+1)(d_2+1)$ is $(d_1+2)(d_2+3)$.
The number of zeros is $\beta=d_2+2d_1+5$. From Lemma~\ref{lem_max_zeros} we
get that $\beta$ must be less than or equal to both $2(d_2+3)$ and $3(d_1+2)$.
The first condition is satisfied if $d_2=d_1+1$ and $d_1=2$, also the second condition
is satisfied if $d_2=2d_1-1$ and $d_1=2$. Consequently, for this $w$ and this size
of $\cR(\cS)$, non-obvious stopping sets exist only for $d_1=2$, $d_2=3$, and $\beta=12$
in a rectangle of size $4 \times 6$. Their number is
\[
1860  {n_1 \choose d_1+2}  {n_2 \choose d_2+3}.
\]
\end{itemize}
\end{itemize}

\subsection*{B. Minimum distances satisfying $d_2=2d_1$}
\begin{itemize}
\item The case $d_1(d_2 + 1)<w<(d_1+1)d_2=d_1(d_2+2)$.\\
Write $w=d_1(d_2 + 1)+\lambda$, where $\lambda$ is in the range $[1,d_1-1]$.
For all sizes of $\cR(\cS)$ in the table in Figure~\ref{fig_rectangle_d1_d2}, 
we find $\beta=\ell_1\ell_2-w$ and we notice that it contradicts Lemma~\ref{lem_max_zeros}.
Thus, there are no stopping sets for $w$ in the range $]d_1(d_2 + 1), (d_1+1)d_2[$.
\item The case $w=(d_1+1)d_2=d_1(d_2+2)$.\\
\begin{itemize}
\item Obvious stopping sets do exist and their number is
\[
{n_1 \choose d_1+1}  {n_2 \choose d_2} + {n_1 \choose d_1}  {n_2 \choose d_2+2}.
\]
\item For rectangles larger than $(d_1+1) \times d_2$ and $d_1\times (d_2+2)$,
all sizes yield no stopping sets (by contradiction on $\beta$) except
for $(d_1+1) \times (d_2+1)$ and $(d_1+1) \times (d_2+2)$ where the number
of non-obvious stopping sets is respectively 
\[
(d_1+1)!  {d_2+1 \choose d_1+1}  {n_1 \choose d_1+1}  {n_2 \choose d_2+1},
\]
and
\[
\frac{(d_2+2)!}{2^{d_1+1}}  {n_1 \choose d_1+1}  {n_2 \choose d_2+2}.
\]
\end{itemize}
\item The case $(d_1+1)d_2=d_1(d_2+2) < w < d_1(d_2+3)$.\\
Write $w=(d_1+1)d_2+\lambda$, where $\lambda$ is in the range $[1,d_1-1]$.
\begin{itemize}
\item The smallest rectangular support with a non-zero number of stopping sets
is $(d_1+1)\times (d_2+1)$. We have $\beta=d_1+1-\lambda$ belonging
to the range $[2, d_1]$. The number of corresponding non-obvious stopping sets
is
\[
(d_1+1-\lambda)!  {d_1+1 \choose \lambda}  {d_2+1 \choose d_1+\lambda} 
                  {n_1 \choose d_1+1}  {n_2 \choose d_2+1} \nonumber.
\]
\item For $\cR(\cS)$ with size $(d_1+1)\times (d_2+2)$, we have $\beta=d_2+2-\lambda$ varying
in the range $[d_1+3,d_2+1]$. The rectangle have $\lambda$ columns without zeros.
Given $r_2=d_1+1-r_0-r_1=d_1+1+r_0-\lambda$,
the number of non-obvious stopping sets is
\[
\label{equ_d_1_1_d_2_2}
\sum_{2r_0+r_1=\lambda}  {d_1+1 \choose r_0} {d_1+1-r_0 \choose r_1} \frac{(d_2+2)!}{2^{r_2} \lambda!} 
 {n_1 \choose d_1+1}  {n_2 \choose d_2+2}.
\]
Larger rectangles $\cR(\cS)$ lead to a contradiction on $\beta$, so they do not create stopping
sets for this given weight $w$.
\end{itemize}
\item The case $w=d_1(d_2+3)$.\\
Obvious stopping sets are given by
\[
{n_1 \choose d_1}  {n_2 \choose d_2+3}.
\]
\begin{itemize}
\item Take $\cR(\cS)$ with size $(d_1+1)(d_2+1)$. Then $\beta=1$ (recall that $d_2=2d_1$
in this sub-section). The number of non-obvious stopping sets with a unique zero
in their rectangular support is
\[
(d_1+1)  (d_2+1)  {n_1 \choose d_1+1}  {n_2 \choose d_2+1}.
\]
\item Take $\cR(\cS)$ with size $(d_1+1)(d_2+2)$. Then $\beta=d_1+2$.
The number of stopping sets is given by (\ref{equ_d_1_1_d_2_2})
after setting $\lambda=d_1$.
\item Take $\cR(\cS)$ with size $(d_1+1)(d_2+3)$. Then $\beta=d_2+3$.
In $\cR(\cS)$, all columns have a unique zero. 
Define $r_3=d_1+1-r_0-r_1-r_2=2r_0+r_1+1$, then 
the number of non-obvious stopping sets in this sub-case becomes
\end{itemize}
\[
\sum_{3r_0+2r_1+r_2=d_1}  {d_1+1 \choose r_0} {d_1+1-r_0 \choose r_1} {d_1+1-r_0-r_1 \choose r_2} 
 \frac{(d_2+3)!}{2^{r_2}6^{r_3}}  {n_1 \choose d_1+1}  {n_2 \choose d_2+3} \nonumber
\]
All remaining rectangle sizes (smaller or larger) have no stopping sets.
\item For $d_2=2d_1$ the range $]d_1(d_2+3), (d_1+1)(d_2+1)[$ is empty.
We complete this sub-section with the last case $w=(d_1+1)(d_2+1)$.
\begin{itemize}
\item The number of obvious stopping sets for the smallest $\cR(\cS)$ is
\[
{n_1 \choose d_1+1}  {n_2 \choose d_2+1}.
\]
The size $(d_1+2)\times d_2$ rectangle has no stopping sets.
\item The next $\cR(\cS)$ is $(d_1+1)(d_2+2)$. The corresponding number of zeros
is $\beta=d_1+1$. The number of non-obvious stopping sets is (expression identical
to the case $d_2 \le 2d_1-1$):
\[
\sum_{2r_0+r_1=d_1+1} {d_1+1 \choose r_0} {d_1+1-r_0 \choose r_1} \frac{(d_2+2)!}{2^{r_0}(d_1+1)!} 
  {n_1 \choose d_1+1}  {n_2 \choose d_2+2}. \nonumber
\]
\item Consider $\cR(\cS)$ with size $(d_1+2)(d_2+1)$. We have $\beta=d_2+1$.
For $d_2=2d_1$ this $\beta$ contradicts the upper bound in Lemma~\ref{lem_max_zeros}.
Then $\tau_w=0$. 
\item Consider $\cR(\cS)$ with size $(d_1+1)(d_2+3)$. We have $\beta=2(d_1+1)=d_2+2$.
In $\cR(\cS)$, all columns must have at most one zero but rows can afford up to three zeros.
The number of non-obvious stopping sets is found to be (method as in previous cases)
\end{itemize}
\[
\sum_{3r_0+2r_1+r_2=d_1+1}  {d_1+1 \choose r_0} {d_1+1-r_0 \choose r_1} {d_1+1-r_0-r_1 \choose r_2} 
  \frac{(d_2+3)!}{2^{r_2}6^{2r_0+r_1}}  {n_1 \choose d_1+1}  {n_2 \choose d_2+3}. \nonumber
\]
\begin{itemize}
\item Consider the next supports $\cR(\cS)$ with size $(d_1+2)(d_2+2)$ and $(d_1+2)(d_2+3)$
as given in the table in Figure~\ref{fig_rectangle_d1_d2}. 
The $\beta$ for both sizes contradicts the upper bound in Lemma~\ref{lem_max_zeros}.
We deduce that $\tau_w=0$ in these cases.
\end{itemize}
\end{itemize}

\subsection*{C. Minimum distances satisfying $d_2>2d_1$}
For $2<d_1<d_2 < 3d_1-1$, the width of $\cR(\cS)$ cannot exceed $d_2+3$.
In the special case $d_1=2$, as stated earlier, a width up to $d_2+4$ should be considered.
Then, for $d_1=2$, the rectangular supports are ordered in increasing size according to Table~\ref{tab_d1=2}.
The first and second rows list the stopping set weight $w$ in increasing order.\\
\vspace{-3mm}
\begin{table}[!h]
\begin{center}
\begin{tabular}{|l|}
\hline
$d_1d_2 < d_1(d_2+1)<d_1(d_2+2)<(d_1+1)d_2$\\
$< d_1(d_2+3)<\mathbf{d_1(d_2+4)}\le (d_1+1)(d_2+1)$\\
$< (d_1+2)d_2 \le (d_1+1)(d_2+2) < (d_1+1)(d_2+3)$\\
$\le (d_1+2)(d_2+1)< \mathbf{(d_1+1)(d_2+4)} < (d_1+2)(d_2+2)$ \\
$< (d_1+2)(d_2+3)<\mathbf{(d_1+2)(d_2+4)}$\\ \hline
\end{tabular}
\end{center}
\caption{Table of rectangular sizes for the special case where the first
component code has $d_1=2$.\label{tab_d1=2}}
\end{table}

\begin{itemize}
\item The case $d_1(d_2 + 1)<w<d_1(d_2+2)<(d_1+1)d_2$.\\
Write $w=d_1(d_2 + 1)+\lambda$, where $\lambda$ is in the range $[1,d_1-1]$.
For all sizes of $\cR(\cS)$ in the table in Figure~\ref{fig_rectangle_d1_d2}, 
we find $\beta=\ell_1\ell_2-w$ and we notice that it contradicts Lemma~\ref{lem_max_zeros}.
There are no stopping sets for $w$ in the range $]d_1(d_2 + 1), d_1(d_2+2)[$.
\item The case $w=d_1(d_2+2)$.\\
\begin{itemize}
\item Obvious stopping sets do exist and their number is
\[
{n_1 \choose d_1}  {n_2 \choose d_2+2}.
\]
\item For rectangles larger than $d_1(d_2+2)$ we found no other stopping sets,
by contradiction on $\beta$.
\end{itemize}
\item The case $d_1(d_2+2) < w < (d_1+1)d_2$.\\
Write $w=d_1(d_2+2)+\lambda$, where $\lambda$ is in the range $[1,d_2-2d_1-1]$.
For all rectangular supports, from $\beta$ we deduce that $\tau_w=0$.
\newpage
\item The case $w=(d_1+1)d_2$.\\
Obvious stopping sets are given by
\[
{n_1 \choose d_1+1}  {n_2 \choose d_2}.
\]
\begin{itemize}
\item Take $\cR(\cS)$ with size $(d_1+1)(d_2+1)$. Then $\beta=d_1+1$. 
The number of non-obvious stopping sets for this sub-case is
\[
(d_1+1)!  {d_2+1 \choose d_1+1}  {n_1 \choose d_1+1}  {n_2 \choose d_2+1}.
\]
\item Take $\cR(\cS)$ with size $(d_1+1)(d_2+2)$. Then $\beta=2d_1+2$.
All rows have two zeros. Also, $d_2-2d_1$ columns have no zeros,
while the remaining columns include a unique zero.
The number of non-obvious stopping sets in this sub-case is
\[
\frac{(d_2+2)!}{2^{d_1+1}(d_2-2d_1)!}  {n_1 \choose d_1+1}  {n_2 \choose d_2+2}.
\]
\item Take $\cR(\cS)$ with size $(d_1+1)(d_2+3)$. From Lemma~\ref{lem_max_zeros}
we find that no stopping sets exist, except for $d_1=2$ and $d_2=6$.
In this case, $\beta=3d_1+3=9$. Each row in the rectangle
have three zeros. The number of stopping sets is $\tau_w=\frac{9!}{6^3}=1680$
for $d_1=2$ and $d_2=6$.

\item Take $\cR(\cS)$ with size $(d_1+2)(d_2+1)$. Then $\beta=d_2+d_1+2$.
The reader can easily check that the bound in Lemma~\ref{lem_max_zeros} is not satisfied.
We deduce that $\tau_w=0$ for this rectangle size and this weight $w$.
All remaining rectangle sizes have no stopping sets.
\end{itemize}
\item The case $(d_1+1)d_2 < w < d_1(d_2+3)$.\\
$\tau_w=0$ for $d_1=2$. We pursue this case for $d_1>2$.\\
Write $w=(d_1+1)d_2+\lambda$ where $\lambda$ belongs
to the non-empty interval $[1, 3d_1-d_2-1]$.
\begin{itemize}
\item The next rectangular support with a non-zero number of stopping sets
is $(d_1+1)(d_2+1)$. The number of zeros is $\beta=d_1+1-\lambda$
varying in the range $[d_2-2d_1+2, d_1]$. 
Non-obvious stopping sets
are enumerated by selecting the location and permuting the $\beta$ zeros
inside $\cR(\cS)$. Their number is
\[
(d_1+1-\lambda)!  {d_2+1 \choose d_1+1-\lambda}  {d_1+1 \choose \lambda} 
                   {n_1 \choose d_1+1}  {n_2 \choose d_2+1} \nonumber.
\]
\item Take $\cR(\cS)$ with size $(d_1+1)(d_2+2)$.
We have $\beta=2d_1+2-\lambda$ inside the interval $[d_2-d_1+3, 2d_1+1]$.
As made before, we find $2r_0+r_1=\lambda$ and $r_2=d_1+1+r_0-\lambda$. The columns in $\cR(\cS)$
have at most one zero and rows have at most two zeros. 
The number of these non-obvious stopping sets is
\end{itemize}
\[
\sum_{2r_0+r_1=\lambda}  {d_1+1 \choose r_0} {d_1+1-r_0 \choose r_1} 
 \frac{(d_2+2)!}{2^{r_2}(d_2-2d_1+\lambda)!}  {n_1 \choose d_1+1}  {n_2 \choose d_2+2}. 
\]
\begin{itemize}
\item Consider $\cR(\cS)$ with size $(d_1+2)(d_2+1)$. 
Here $\beta=d_2+d_1+2-\lambda$ contradicts Lemma~\ref{lem_max_zeros} 
because $d_2 > 2d_1$. We have $\tau_w=0$.
All remaining rectangles (smaller or larger) have no stopping sets.
\end{itemize}
\item The case $w=d_1(d_2+3)$.
\begin{itemize}
\item The number of obvious stopping sets in a $d_1 \times (d_2+3)$ rectangular
support is
\[
{n_1 \choose d_1}  {n_2 \choose d_2+3}.
\]
\item Consider  $\cR(\cS)$ with size $(d_1+1)(d_2+1)$.
Here $\beta=d_2-2d_1+1$ is restricted to the interval $]1,d_1[$
given the constraints $2d_1 < d_2 < 3d_1-1$ for $d_1>2$.
For $d_1=2$, $d_2=6$ is the only valid value, with $\beta=3$. 
The number of non-obvious stopping sets is (for $d_1 \ge 2$)
\[
\beta! {d_1+1 \choose 3d_1-d_2} {d_2+1 \choose 2d_1} {n_1 \choose d_1+1} {n_2 \choose d_2+1}.
\]
\item Consider  $\cR(\cS)$ with size $(d_1+1)(d_2+2)$. 
The number of zeros is $\beta=d_2-d_1+2 \in ]d_1+2, 2d_1+1[$. The number
of non-obvious stopping sets is given by
\end{itemize}
\[
\sum_{2r_0+r_1=3d_1-d_2}  {d_1+1 \choose r_0} {d_1+1-r_0 \choose r_1} 
 \frac{(d_2+2)!}{2^{r_2}d_1!}  {n_1 \choose d_1+1}  {n_2 \choose d_2+2} ,
\]
where $r_2=d_2-2d_1+1+r_0$. For $d_2=2$, the above expression is valid
for $d_2=6$ only. 
\begin{itemize}
\item Consider $\cR(\cS)$ with size $(d_1+1)(d_2+3)$. Here $\beta=d_2+3$.
All columns in the rectangle have one zero. Rows can have up to three zeros.
The number of non-obvious stopping sets is 
\[
\sum_{3r_0+2r_1+r_2=3d_1-d_2} {d_1+1 \choose r_0} {d_1+1-r_0 \choose r_1}
 {d_1+1-r_0-r_1 \choose r_2}
\frac{(d_2+3)!}{2^{r_2}6^{r_3}} {n_1 \choose d_1+1} {n_2 \choose d_2+3}, \nonumber
\]
where $r_3=d_1+1-r_0-r_1-r_2$. The above expression is also valid for $(d_1,d_2)=(2,6)$.
\item The rectangular support $\cR(\cS)$ of size $(d_1+2)(d_2+1)$ gives no stopping sets.
The remaining rectangular supports from Figure~\ref{fig_rectangle_d1_d2} 
and Table~\ref{tab_d1=2} yield no stopping sets for $w=d_1(d_2+3)$.\\
\end{itemize}
\noindent
Recall that a maximal rectangle width of $d_2+3$ should be considered for $d_1>2$
and it goes up to $d_2+4$ for $d_1=2$ as shown in Table~\ref{tab_d1=2}.
New obvious stopping sets are found, they appear for $d_1=2$ only with a rectangular
width equal to $d_2+4$. Their rectangular support corresponds to the sizes in boldface
in Table~\ref{tab_d1=2} for $w>d_1(d_2+3)$:
${ n_1 \choose d_1} {n_2 \choose d_2+4}$ obvious stopping sets of size $d_1 \times (d_2+4)$,
${ n_1 \choose d_1+1} {n_2 \choose d_2+4}$ obvious stopping sets of size $(d_1+1) \times (d_2+4)$,
and ${ n_1 \choose d_1+2} {n_2 \choose d_2+4}$ obvious stopping sets of size $(d_1+2) \times (d_2+4)$.
Given that this theorem enumerates stopping sets for $w\le (d_1+1)(d_2+1)$, 
one should only count obvious $d_1 \times (d_2+4)$ sets.
\item The case $d_1(d_2+3)<w<(d_1+1)(d_2+1)$.\\
Write $w=d_1(d_2+3)+\lambda$ where $\lambda \in [1, d_2-2d_1]$.
The results for the three rectangles listed below are valid for $d_1 \ge 2$.
\begin{itemize}
\item Consider $\cR(\cS)$ of size $(d_1+1)(d_2+1)$. We have $\beta=d_2-2d_1+1-\lambda$
varying in the range $[1, d_2-2d_1]$. The number of non-obvious stopping sets is
\[
(d_2-2d_1+1-\lambda)! {d_1+1 \choose 3d_1-d_2+\lambda}  
 {d_2+1 \choose 2d_1+\lambda} 
  {n_1 \choose d_1+1}  {n_2 \choose d_2+1}. 
\]
\item Now consider $\cR(\cS)$ of size $(d_1+1)(d_2+2)$. 
We have $\beta=d_2-d_1+2-\lambda \in [d_1+2, d_2-d_1+1]$.
As done before, the expression of the number of non-obvious stopping sets
involves $r_0$ and $r_1$ as follows.
\end{itemize}
\[
\sum_{2r_0+r_1=3d_1-d_2+\lambda}  {d_1+1 \choose r_0} {d_1+1-r_0 \choose r_1} 
\frac{(d_2+2)!}{2^{r_2}(d_1+\lambda)!}  {n_1 \choose d_1+1}  {n_2 \choose d_2+2}, 
\]
where $r_2=d_1+1-r_0-r_1$.
\begin{itemize}
\item We consider the next $(d_1+1)(d_2+3)$ rectangular support. 
Now $\beta=d_2+3-\lambda \in [2d_1+3, d_2+2]$. The number of non-obvious stopping sets
is
\[
\sum_{3r_0+2r_1+r_2=3d_1-d_2+\lambda} {d_1+1 \choose r_0} {d_1+1-r_0 \choose r_1} 
{d_1+1-r_0-r_1 \choose r_2}
\frac{(d_2+3)!}{2^{r_2}6^{r_3}\lambda!}{n_1 \choose d_1+1}{n_2 \choose d_2+3}, 
\]
where $r_3=d_1+1-r_0-r_1-r_2$.
\item The remaining larger rectangular supports give no stopping sets,
except for $(d_1+1)(d_2+4)$ for $d_1=2$ and $d_2=6$ where $\tau_w=22050$.
These $22050$ rectangles of size $3 \times 10$, where $\beta=10$ and $w=20$, 
have one zero in each column but a row may have up to four zeros.
\end{itemize}
\newpage
\item The last case $w=(d_1+1)(d_2+1)$.
\begin{itemize}
\item Obvious stopping sets are given by
\[
{n_1 \choose d_1+1}  {n_2 \choose d_2+1}.
\]
\item The next $\cR(\cS)$ from the table is $(d_1+1)\times (d_2+2)$.
We have $\beta=d_1+1$. The number of stopping sets is
\[
(d_1+1)!  {d_2+2 \choose d_1+1}   {n_1 \choose d_1+1}  {n_2 \choose d_2+2}.
\]
\item Now consider the $(d_1+1)\times (d_2+3)$ rectangle. We have
$\beta=2d_1+2$. The number of stopping sets is
\end{itemize}
\begin{align}
\sum_{3r_0+2r_1+r_2=d_1+1} & {d_1+1 \choose r_0} {d_1+1-r_0 \choose r_1}{d_1+1-r_0-r_1 \choose r_2}\nonumber \\
&  \frac{(d_2+3)!}{2^{r_2}6^{r_3}(d_2-2d_1+1)!}{n_1 \choose d_1+1}{n_2 \choose d_2+3},\nonumber 
\end{align}
where $r_3=d_1+1-r_0-r_1-r_2$, the above expression being valid
for $d_1\ge 2$.\\
\begin{itemize}
\item No stopping sets are found for the remaining three rectangular supports for $d_1>2$.
On the other hand, for $d_1=2$, stopping sets are found only with a rectangle $(d_1+1)(d_2+4)$.
In this case, we have $\beta=3(d_1+1)$. The number of non-obvious stopping sets is
$\tau_w=11130$ for $(d_1,d_2)=(2,5)$ and $\tau_w=111300$ for $(d_1,d_2)=(2,6)$.
\end{itemize}
\end{itemize}

Q.E.D.
~\\
~\\
\section*{Acknowledgment}
The work of Joseph J. Boutros was supported by the Qatar National Research Fund (QNRF), 
a member of Qatar Foundation, 
under NPRP project 5-401-2-161 on layered coding.
The authors would like to thank Dr. Mireille Sarkiss, from CEA-LIST Paris, for her
precious support.

\end{document}